%% file: cas-sc-template.tex
\documentclass[a4paper,fleqn]{cas-sc}
\ExplSyntaxOn
\cs_gset:Npn \__first_footerline: { \mbox{} }
\ExplSyntaxOff

\usepackage[authoryear,longnamesfirst]{natbib}

\def\tsc#1{\csdef{#1}{\textsc{\lowercase{#1}}\xspace}}
\tsc{WGM}
\tsc{QE}
\input{preamble}

\begin{document}
\let\WriteBookmarks\relax
\def\floatpagepagefraction{1}
\def\textpagefraction{.001}

% Short title
\shorttitle{}    

% Short author
\shortauthors{}  

% Main title of the paper
\title [mode = title]{Model Predictive Control for Dynamic Hydraulic Balancing in Building Radiator Heating Systems: Modeling, Design, and Experimental Validation}

%Experimental Validation of Hydraulic-Balancing-Aware Model Predictive Control for Radiator Systems in Buildings}  

% Title footnote mark
% eg: \tnotemark[1]
%\tnotemark[1] 

% Title footnote 1.
% eg: \tnotetext[1]{Title footnote text}
%\tnotetext[1]{} 

% First author
%
% Options: Use if required
% eg: \author[1,3]{Author Name}[type=editor,
%       style=chinese,
%       auid=000,
%       bioid=1,
%       prefix=Sir,
%       orcid=0000-0000-0000-0000,
%       facebook=<facebook id>,
%       twitter=<twitter id>,
%       linkedin=<linkedin id>,
%       gplus=<gplus id>]

\author[1]{Erfan Shakhesi}%[<options>]

% Corresponding author indication
\cormark[1]

% Corresponding author text
\cortext[1]{Corresponding author}

% Footnote of the first author
%\fnmark[1]

% Email id of the first author
\ead{e.shakhesi@tue.nl}

% URL of the first author
%\ead[url]{}

% Credit authorship
% eg: \credit{Conceptualization of this study, Methodology, Software}
%\credit{Conceptualization, Formal analysis, Investigation, Methodology, Software, Validation, Writing -- Original Draft}
% Address/affiliation
\affiliation[1]{organization={Eindhoven University of Technology},
            addressline={Department of Mechanical Engineering}, 
            city={Eindhoven},
%          citysep={}, % Uncomment if no comma needed between city and postcode
            postcode={5600 MB}, 
            %state={North Brabant},
            country={The Netherlands}}

%\author[1]{Alexander Katriniok}%[]

% Footnote of the second author
%\fnmark[1]

% Email id of the second author
%\ead{a.katriniok@tue.nl}

% URL of the second author
%\ead[url]{}

% Credit authorship
%\credit{}
\author[1]{Maurice Heemels}%[]
%\fnmark[1]
\ead{m.heemels@tue.nl}
%\ead[url]{}
%\credit{Conceptualization, Formal analysis, Investigation, Methodology, Project administration, Validation, Writing -- Review \& Editing}

\author[1]{Paula Chanfreut}%[]
%\fnmark[1]
\ead{p.chanfreut.palacio@tue.nl}
%\ead[url]{}
%\credit{Conceptualization, Formal analysis, Investigation, Methodology, Funding acquisition, Project administration, Validation, Writing -- Review \& Editing}

%\fntext[1]{}

% For a title note without a number/mark
%\nonumnote{}

% Here goes the abstract
\begin{abstract}
Hydronic radiator systems are among the most widely used heating systems in buildings. In the literature, radiator heat outputs are often assumed to be independent of one another and arbitrarily adjustable in time when designing controllers. However, in practice, radiators are supplied by one or multiple common heat sources and are hydraulically coupled through the water circulation system.
Maintaining occupant comfort in multi-zone buildings therefore requires not only an appropriate supply temperature but also proper hydraulic balancing to distribute the available water flow according to the heating demand of each zone. 
To this end, we develop a grey-box thermal model that captures the hydraulic interactions among radiators and the effects of radiator valves, circulation pumps, and heat sources. In a real building, we show that accounting for hydraulic interactions reduces the root mean square error (RMSE) between the measured and modeled zone temperatures by approximately $11\%$ compared with a model that neglects these interactions. Additionally, we integrate the developed model into a model predictive control (MPC) framework for dynamic hydraulic balancing that jointly optimizes valve openings and the supply temperature to maintain thermal comfort in each zone while reducing energy consumption.
Through both real-world experiments and numerical case studies, we demonstrate that, compared with existing MPC formulations that neglect hydraulic interactions or do not control radiator valves, the proposed MPC reduces comfort-range violations by at least $27\%$ while requiring a similar or even lower cumulative supply temperature.
\end{abstract}

% Research highlights
%\begin{highlights}
%\item Proposed grey-box model captures hydraulic–thermal coupling in hydronic radiators.
%\item Accounting for radiator hydraulic interactions reduces zone-temperature RMSE by 11\%.
%\item Proposed Model Predictive Control (MPC) enables dynamic hydraulic balancing.
%\item Proposed MPC reduces comfort violations by 27\% while maintaining low energy use.
%\item Real-world experiments validate the proposed MPC in a six-zone building.
%\end{highlights}

% Keywords
% Each keyword is seperated by \sep

\maketitle

% Main text
\section{Introduction}\label{}
\input{s1}
% Numbered list
% Use the style of numbering in square brackets.
% If nothing is used, default style will be taken.
%\begin{enumerate}[a)]
%\item 
%\item 
%\item 
%\end{enumerate}  

% Unnumbered list
%\begin{itemize}
%\item 
%\item 
%\item 
%\end{itemize}  

% Description list
%\begin{description}
%\item[]
%\item[] 
%\item[] 
%\end{description}  

% Uncomment and use as the case may be
%\begin{theorem} 
%\end{theorem}

% Uncomment and use as the case may be
%\begin{lemma} 
%\end{lemma}

%% The Appendices part is started with the command \appendix;
%% appendix sections are then done as normal sections
%% \appendix

\section{Building Thermal--Hydraulic Model}\label{sec:model}

\input{s2}

\section{Grey-Box Model Identification}\label{sec:grey-box}
\input{s3_new}

\input{s3}

\section{Model Predictive Control}\label{sec:mpc}
\input{s4}

\section{Conclusion and Future Work}
\input{s5}

%\printcredits

\section*{Acknowledgements}
The authors would like to thank Shalika Walker from Kropman, Christian Portilla from OpenToControl, and Hugo Theunissen and Ivo van Wolferen from Rensen for their valuable feedback and suggestions, as well as their assistance in conducting the experiments. This research was supported by the Netherlands Enterprise Agency [project number MOOI224004], and was partially supported by the Dutch Research Council (NWO) under the VENI TTW Grant 21957, ``Dynamic Home Energy Management: Optimizing for a Sustainable Future (DHEMOS)'' [grant ID \url{https://doi.org/10.61686/MVDUJ87852}].

%\section*{Declaration of competing interest}
%The authors declare that they have no known competing financial interests or personal relationships that could have appeared to influence the work reported in this paper.

% To print the credit authorship contribution details

\section*{Declaration of generative AI and AI-assisted technologies}
The authors used ChatGPT (GPT-5.6 Sol) for the graphical representation of the radiator metal and water in Figure~\ref{fig:water_rad}, and to add graphical details to Figure~\ref{fig:solar_window} to improve its presentation. These figures are used only for explanatory purposes.
%% Loading bibliography style file
%\bibliographystyle{model1-num-names}
\bibliographystyle{cas-model2-names}

%\newpage
% Loading bibliography database
\bibliography{cas-refs}

% Biography
%\bio{}
% Here goes the biography details.
%\endbio

%\bio{pic1}
% Here goes the biography details.
%\endbio

%\appendix
%\section{A1}
%\input{Sections/Appendix/A1}

\end{document}

%% file: preamble.tex
\usepackage{amsthm}
\usepackage{graphicx}
\usepackage{amsmath}
\usepackage{algorithm}
\usepackage{algpseudocode}

\usepackage{amssymb}
\usepackage{siunitx}
\usepackage{float}

\usepackage{booktabs}
\usepackage{tabularx}
\usepackage{siunitx}
\usepackage{mathtools}
\usepackage{bbm}
\usepackage[bb=boondox]{mathalfa}
\usepackage{subcaption}
\DeclareSIUnit\bar{bar}
\newtheorem{remark}{Remark}

\newtheorem{theorem}{Theorem}

\usepackage[table]{xcolor}
\newcommand{\rowline}{\arrayrulecolor{gray!50}\specialrule{0.15pt}{0pt}{0pt}\arrayrulecolor{black}}
\usepackage{makecell}
\usepackage{array}
\newcolumntype{M}[1]{>{\raggedright\arraybackslash}m{#1}}

%% file: s1.tex
Buildings are increasingly required to provide thermal comfort while simultaneously reducing their environmental footprint, including energy consumption. In many buildings, space heating is provided by hydronic (water-based) heat emitters, and radiators are the most common type. Across the European Union (EU), hydronic heat emitters are installed in nearly 130 million buildings \citep{EHI2021HeatingMarket}. This is further supported by statistics from the United Kingdom, where approximately $90\%$ of dwellings use hydronic radiators \citep{DLUHC2023EnglishHousingSurvey}. 

Hydronic radiators are supplied with hot water from one or more central heat sources, such as boilers and heat pumps, and as the water circulates through them, heat is transferred to the surrounding air, thereby heating the space. Most radiators are equipped with valves that regulate their water flow rate based on the difference between the zone temperature and its temperature setpoint. This setpoint can be selected manually by occupants at the radiator, as in thermostatic radiator valves (TRVs) \citep{TRV_simulation}, or adjusted remotely by a controller.
Although radiator valves provide basic zone-level thermal comfort, they operate locally and do not coordinate with other heating system components. In addition, adjusting valve positions alone is not sufficient for improving thermal comfort and energy efficiency, and optimizing the water supply temperature, often shared across multiple zones, is essential. These limitations of local valve-based control motivate the use of more integrated control systems.

One promising integrated control strategy is model predictive control (MPC), which can also anticipate external factors such as changes in outdoor temperature and occupancy patterns, providing an additional advantage for improving thermal comfort and energy efficiency \citep{survey_MPC2025, survey_MPC2020, Castilla2014}. In particular, MPC determines the values of selected control inputs at the current time by optimizing an objective function subject to predefined constraints, based on predictions of the future behavior of the system over a specified prediction horizon, for example, the next two hours, using forecasts of exogenous inputs such as the outdoor temperature \citep{camacho2026mpc}. In the literature, existing MPC schemes for building radiator heating systems often compute high-level control inputs, such as the required heat output of each zone \citep{BRCM_toolbox, EMPC_radiator}, which must then be translated into implementable control actions through an additional calculation, for example, by converting the required heat output into a pulse-width modulation (PWM) signal for the radiator valves. Other schemes use a cascade framework \citep{PEDERSEN2019772, LIU2025112936, KNUDSEN2024110694, SIROKY20113079}, in which the MPC computes high-level setpoints, such as the water supply temperature or zone-temperature references, while low-level proportional–integral (PI) or proportional–integral–derivative (PID) controllers associated with the valves or boilers track these setpoints. In these works, however, the dynamics of the low-level controllers are not explicitly represented in the MPC prediction model.

Both classes of MPC schemes have two main drawbacks. First, a mismatch may arise between the high-level decisions computed by the MPC and their implementation by the low-level controllers. Second, these schemes are often formulated as if each radiator operated independently, thereby neglecting the hydraulic interactions among the radiators. In practice, however, the radiators draw hot water from one or more shared heat sources through pumps, meaning that adjusting one valve can affect the flow distribution and heat output of the others. For example, a fully opened valve in one zone may reduce the water flow and heating capacity available to radiators in other zones \citep{en12173215}. This can lead to situations in which satisfying the comfort requirements of one zone causes comfort violations in another, which is inherently a multi-zone interaction. In the literature, the importance of modeling hydraulic interactions in
building simulations has been discussed in \cite{TRV_simulation,
hydraulic_thermal_sim, en12173215}. However, these works either consider no control strategy or rely only on simple control approaches such as rule-based control. Even with rule-based control, \cite{hydraulic_thermal_sim} reported energy savings of $8\%-13\%$ through terminal radiator flow-rate control using a thermal-hydraulic model that accounts for hydraulic interactions. To achieve further energy savings and/or improve thermal comfort, more advanced control strategies, such as MPC, can be used. In \cite{GUO2024123951}, it is demonstrated that hydraulic imbalance can substantially degrade MPC performance. However, hydraulic balancing in this work is handled by a separate control scheme rather than being integrated into the MPC formulation, and the valve positions are therefore not dynamically optimized by the MPC.
This motivates the development of a modeling framework that captures the coupled thermal--hydraulic dynamics while remaining suitable for the real-world MPC implementation. In particular, when radiator valves are remotely controllable, this modeling approach enables the MPC to accurately adjust the valve positions to improve the heat distribution across zones and, consequently, improving thermal comfort while reducing energy use.

To model building thermal dynamics, different approaches have been used in the literature. These include models developed entirely from physical principles, referred to as \textit{white-box models}, in which building thermal dynamics are typically represented using linear resistor--capacitor circuits \citep{BRCM_toolbox}. Moreover, \textit{grey-box models} are used in \cite{EMPC_radiator, MPC_public_building, thilker2021nonlinear, TRV_vs_MPC, WANG2024130812, 01KGF3XPRTF4S7MGWJ6WT6G01J}, where physical knowledge is used to define the model structure and measurement data are used to estimate the model parameters. On the other hand, \textit{black-box} models assume no prior knowledge of the system dynamics and are used, for example, in \cite{KNUDSEN2021117227}. These models are then integrated into MPC frameworks to achieve predefined objectives. 

Regarding existing MPC frameworks with experimental validation, the survey in \cite{survey_MPC2025} identifies 90 studies on MPC for heating, ventilation, and air-conditioning (HVAC) systems, of which six consider hydronic radiator-based heating systems \citep{EMPC_radiator, PARISIO2014599, 7040202, 7087366, STAUFFER2017127, 5299104}. In all of these works, the MPC decision variables are high-level control inputs, such as the heat delivered to a zone. Moreover, none of these works optimizes the distribution of water flow among targeted zones or discusses the effects of dynamic hydraulic balancing on thermal comfort or energy savings. More specifically, with the exception of \cite{7087366}, these works consider a single aggregated temperature for the corresponding case-study building or each floor, offering a useful simplification of the control problem. However, comfort violations in smaller zones may remain hidden at the aggregated scale. 

Motivated by the control challenges discussed above, this article considers multi-zone buildings with radiators equipped with PI(D)-controlled valves and supplied by a common boiler. To the best of our knowledge, no MPC formulation that explicitly accounts for hydraulic interactions among the radiators while jointly optimizing the references of the valve PI(D) controllers and the supply temperature has been experimentally validated in a real multi-zone building. To fill this gap in the literature, we present four main contributions. First, we present a physics-based thermal–hydraulic building model, in which the radiator heat output is explicitly modeled as a function of the water flow rate and the temperature difference across each radiator, accounting for radiator valves, water pumps, and boilers (Section~\ref{sec:model}). Second, we adopt two distinct grey-box modeling approaches, using real data to adjust the parameters that are difficult to accurately obtain in practice. Using real measurements, we also show that accounting for hydraulic interactions among the radiators reduces the root mean square error (RMSE) between the measured and modeled zone temperatures by 11\% (Section~\ref{sec:grey-box}). Third, we propose an MPC strategy for dynamic hydraulic balancing that actively optimizes the radiator valve positions to modify the hydraulic resistance of each radiator branch and thereby distribute flow according to the heating demand of each zone. The objective is to provide occupant comfort in each zone while simultaneously reducing the water supply temperature, which is related to the building energy consumption. Fourth, through numerical case studies and experiments in a real six-zone building with 10 radiators, we show that the proposed MPC for dynamic hydraulic balancing significantly outperforms MPC formulations that either cannot control the radiator valves or neglect the hydraulic interactions among the radiators. In particular, it reduces comfort violations by at least $27\%$ while requiring the same or even lower cumulative water supply temperature (Section~\ref{sec:mpc}).

%% file: s2.tex
\input{symbols}

This section presents the thermal--hydraulic model structure of a building heated by hydronic radiators with valves. The focus of this section is on the physical model formulation, and the identification/adjustments of the model parameters from data is discussed later in Section~\ref{sec:grey-box}.
\subsection{Building}
Each zone in a building is modeled using multiple thermal states: one representing the zone air temperature and one representing the temperature at the mid-plane of each wall, including the ceiling and floor, enclosing the zone. Similar to \cite{BRCM_toolbox}, the air temperature $T_{\text{z}}^{(i)}$ of zone $i$, ~\mbox{$i \in \{1,2,\hdots, n_{\mathrm{z}}\}$}, is described by
\begin{align}
    C_{\text{z}}^{(i)} \frac{dT_{\text{z}}^{(i)}}{dt} &=  \sum_{j \in \mathcal{N}_{\mathrm{w}}^{(i)}}\frac{A_{\text{w}}^{(j)}}{R_{\text{w}}^{(j)}}\left(T_{\text{w}}^{(j)} - T_{\mathrm{z}}^{(i)}\right) 
    +\sum_{j \in\mathcal{N}_{\mathrm{w}}^{(i)}}\frac{A_{\text{wd}}^{(j)}}{R_{\text{wd}}^{(j)}}\left(T_{\text{o}} - T_{\text{z}}^{(i)}\right) + G_{\text{inf}}^{(i)}\left(T_\mathrm{o} - T_{\mathrm{z}}^{(i)}\right) + Q_{\text{int}}^{(i)}  \nonumber\\
    & \quad + \sum_{j \in \mathcal{N}_{\mathrm{rad}}^{(i)}}Q_{\text{rad}}^{(j)} + \sum_{j \in\mathcal{N}_{\mathrm{w}}^{(i)}}f_{\mathrm{sec}}^{(j)}f_{\mathrm{SHGC}}^{(j)}A_{\mathrm{wd}}^{(j)}q_{\text{sol}}^{(j)}, \label{eq:room_temp}
\end{align}
where all variables are defined in Table~\ref{tab:symbols_thermal_model}. Note that for ease of notation, we omit the time argument $t$ in \eqref{eq:room_temp} and also throughout this section. On the right-hand side of \eqref{eq:room_temp}, the terms represent, respectively, heat transfer between the zone air and adjacent walls, heat exchange with the outdoor environment through windows, air infiltration from outdoors, internal heat gains (heat released inside the building by occupants, lighting, etc.), radiator heating, and inward secondary solar heat gains through windows (see Figure~\ref{fig:solar_window}). The outdoor temperature $T_{\text{o}}$ (prediction or measurement) is obtained from a weather forecasting station, the internal heat gain $Q_{\text{int}}^{(i)}$ of zone $i$ is estimated (e.g., from historical data or room booking information), and the computation of $q_{\text{sol}}^{(j)}$ for the window(s) on wall $j$ is described in Section~\ref{sec:solar_computation} below. These three variables are treated as exogenous inputs, while $Q_{\text{rad}}^{(j)}$, $j \in \mathcal{N}_{\mathrm{rad}}^{(i)}$, is the only (indirectly) controllable variable, as explained in Section~\ref{sec:radiator_heat_output} below.

Following \cite{BRCM_toolbox}, the dynamics of the wall temperature $T_{\mathrm{w}}^{(i)}$, $i \in \{1,2,\hdots, n_{\mathrm{w}}\}$, can be written as
\begin{align}
    C_{\mathrm{w}}^{(i)} \frac{dT_{\mathrm{w}}^{(i)}}{dt} &=  \sum_{j \in \mathcal{N}_{\mathrm{z}}^{(i)}} \frac{A_{\mathrm{w}}^{(i)}}{R_{\mathrm{w}}^{(i)}} \left( T_{\mathrm{z}}^{(j)} - T_{\mathrm{w}}^{(i)} \right) + \gamma_{\text{abs}}^{(i)}A_{\mathrm{w}}^{(i)}q_{\text{sol}}^{(i)} \nonumber \\
    &\quad +
    \sum_{j \in \mathcal{N}_{\mathrm{z}}^{(i)}} \sum_{k \in \mathcal{N}_{\mathrm{w}}^{(j)}} \left(1 - f_{\text{sec}}^{(k)} \right)f_{\text{SHGC}}^{(k)}A_{\mathrm{wd}}^{(k)} \frac{A_{\mathrm{w}}^{(i)}}{\sum_{l \in \mathcal{N}_{\mathrm{w}}^{(j)}} A_{\mathrm{w}}^{(l)}} q_{\text{sol}}^{(k)}, \label{eq:wall_temp}
\end{align}
where $T_{\mathrm{z}}^{(0)} \coloneqq T_{\mathrm{o}}$. On the right-hand side of \eqref{eq:wall_temp}, the first term models the heat transfer between the wall and its adjacent zones, and the second term accounts for solar radiation absorbed directly by the wall. Moreover, the third term represents direct solar gains transmitted through the windows and redistributed to the surrounding walls (see Figure~\ref{fig:solar_window}). In particular, the outer summation runs over all zones \mbox{$j \in \mathcal{N}_{\mathrm{z}}^{(i)}$} adjacent to wall $i$. For each such zone $j$, the inner summation runs over the windows on all walls $k \in \mathcal{N}_{\mathrm{w}}^{(j)}$ that surround zone $j$. The transmitted solar gain is then distributed among the walls enclosing zone~$j$ proportionally to their areas.

\begin{figure}
    \centering
    \includegraphics[width=0.55\linewidth]{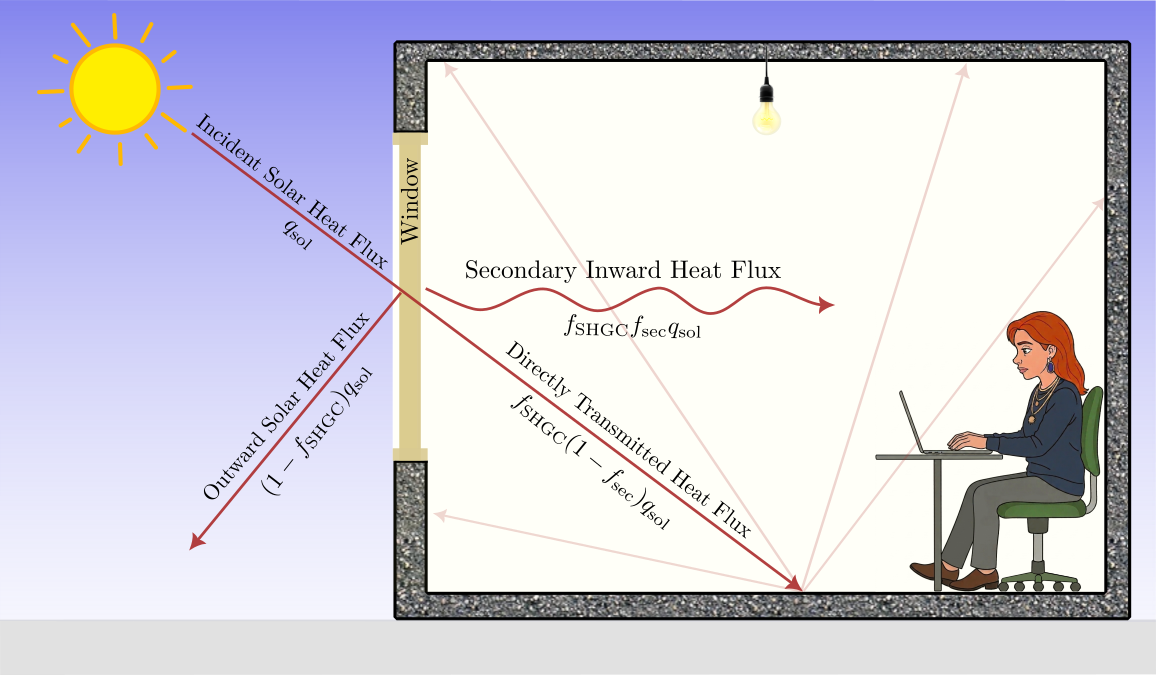}
    \caption{Schematic of the directly transmitted solar heat flux, its reflection toward the surrounding walls, and the secondary inward solar heat flux.}
    \label{fig:solar_window}
\end{figure}

\subsection{Solar Irradiance} \label{sec:solar_computation}
To compute the solar irradiance $q_{\mathrm{sol}}^{(i)}$ incident on wall~$i$, we adapt the approach originally developed for solar panels in \cite{solar}. This approach relies on the following data, which can be obtained from a weather forecasting station:
\begin{itemize}
    \item Direct Normal Irradiation (DNI) $[\si{\watt \hour \per\square\meter }]$ is the amount of solar irradiation received on a surface oriented perpendicular to the sun’s rays.
    \item Diffuse Horizontal Irradiation (DHI) $[\si{\watt \hour \per\square\meter }]$ is the amount of solar irradiation scattered by the atmosphere and received on a horizontal surface.
    \item Global Horizontal Irradiation (GHI) $[\si{\watt \hour \per\square\meter }]$ is the total solar irradiation received on a horizontal surface.
\end{itemize}
\begin{figure}
    \centering
    \includegraphics[width=0.55\linewidth]{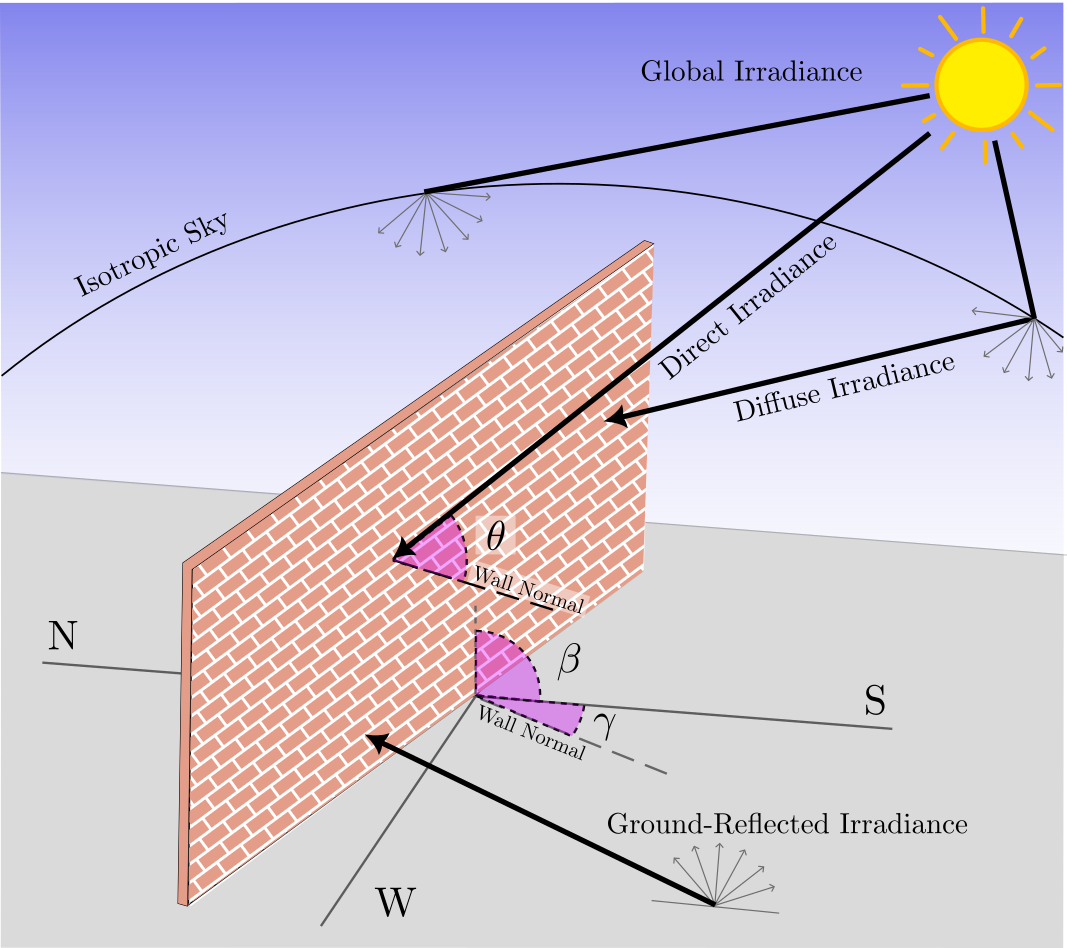}
    \caption{Schematic of direct, diffuse, and ground-reflected solar radiation incident on a wall.}
    \label{fig:solar_wall}
\end{figure}
Assuming isotropic sky, the total solar irradiation $I_{\mathrm{T}}^{(i)} \; [\si{\watt \hour \per\square \meter}]$ on wall $i$ is estimated as 
\begin{align}
    I_{\mathrm{T}}^{(i)} = \text{DNI}\cos{\theta^{(i)}} + \text{DHI}\left(\frac{1+\cos{\beta^{(i)}}}{2}\right) + \text{GHI}\rho_{\mathrm{g}}\left(\frac{1-\cos{\beta^{(i)}}}{2}\right), \nonumber
\end{align}
where $\beta^{(i)} \; [\si{\radian}]$ is the angle between the plane of wall $i$ and the horizontal (typically $\pi/2$), $\theta^{(i)} \; [\si{\radian}]$ is the angle of incidence between the beam radiation and the surface normal (see Figure~\ref{fig:solar_wall}), and $\rho_{\mathrm{g}} \in [0, 1] \subset \mathbb{R}$ is the ground reflectance coefficient (typically $\rho_{\mathrm{g}} \approx 0.4$). To compute $\theta^{(i)}$, we use
\begin{align}
    \cos{\theta^{(i)}} &= \sin{\delta} \sin{\phi} \cos{\beta^{(i)}} - \sin{\delta} \cos{\phi} \sin{\beta^{(i)}} \cos{\gamma^{(i)}} + \cos{\delta} \cos{\phi} \cos{\beta^{(i)}}\cos{\omega} \nonumber \\ 
    &\quad + \cos{\delta} \sin{\phi} \sin{\beta^{(i)}} \cos{\gamma^{(i)}} \cos{\omega} + \cos{\delta} \sin{\beta^{(i)}} \sin{\gamma^{(i)}} \sin{\omega}, \nonumber
\end{align}
where $\phi \; [\si{\radian}]$ is the latitude of the location (north positive), and $\gamma^{(i)} \; [\si{\radian}]$ is the surface azimuth angle of wall~$i$ (with zero due south, east negative, and west positive, $-\pi \leq \gamma^{(i)} \leq \pi$; see Figure~\ref{fig:solar_wall}). Moreover, $\delta \; [\si{\radian}]$ is the angular position of the sun at solar noon and is computed as 
\begin{align}
    \delta &= 0.0069 - 0.40 \cos{B} + 0.07 \sin{B} - 0.0068 \cos{2B} + 0.00091 \sin{2B} \nonumber \\
    &\quad - 0.0027 \cos{3B} + 0.0015 \sin{3B}, \nonumber
\end{align}
where
%\begin{align}
    $B \coloneqq (n_{\text{doy}} - 1) \frac{360}{365}$
%\end{align}
with $n_{\text{doy}} \in \mathbb{N} \coloneqq \{1,2,3,\hdots\}$ denoting the day of the year. Lastly, \(\omega \; [\si{\radian}]\) denotes the hour angle and is computed as 
\begin{align}
    \omega = 15\frac{\pi}{180}\left( t_\text{solar} - 12 \right), \nonumber
\end{align}
where $t_\text{solar} \, [\si{\hour}]$ is the solar time and is computed as
%Solar time is computed as 
\begin{align}
    t_\text{solar} = \frac{4(L_{\mathrm{st}} - L_{\mathrm{loc}}) + E}{60} + t_{\mathrm{std}}. \nonumber
\end{align}
Here, $t_{\mathrm{std}} \, [\si{\hour]}$ is the standard time, $L_{\mathrm{st}} \, [\si{\degree}]$ is the standard meridian for the local time zone, $L_{\mathrm{loc}} \, [\si{\degree}]$ is the longitude of the
location, and longitudes are in degrees west, that is, \mbox{$0 \leq L_{\mathrm{loc}} \leq 360^{\circ}$}. Moreover, 
\begin{align}
E &\coloneqq 229.2\bigl(0.000075 + 0.001868 \cos{B}  - 0.032077 \sin{B} \nonumber\\
&\quad - 0.014615 \cos{2B}  - 0.04089 \sin{2B}\bigr). \nonumber
\end{align}
Assuming that DNI, DHI, and GHI are hourly integrated and that the solar irradiance is taken to be constant within
each one-hour interval, the corresponding solar
irradiance $q_{\mathrm{sol}}^{(i)}$ incident on wall~$i$ and its window(s), if any, is computed as
\begin{align}
    q_{\mathrm{sol}}^{(i)} \, [\si{\watt\per\square\meter}] =
\frac{I_{\mathrm{T}}^{(i)} \,[\si{\watt \hour \per\square\meter }]}{1 \,[\si{\hour}]}. \nonumber
\end{align}

\subsection{Radiator Heat Output} \label{sec:radiator_heat_output}
As discussed, the heat output of each radiator represents a controllable variable to be optimized for achieving the desired building thermal response.
For water-based radiators, however, a desired heat output cannot be commanded arbitrarily, and therefore the radiator dynamics must also be modeled.

\subsubsection{Radiator Modeling} \label{sec:radiator_modeling}
We model a radiator as a two-node lumped model \citep{two_node_radiator}: (i) the radiator metal and (ii) the water inside the radiator; see Figure~\ref{fig:water_rad}.
The radiator metal temperature $T_{\mathrm{rad}}^{(i)}$ for radiator $i$, $i \in \{1,2,\ldots,n_{\mathrm{rad}}\}$, evolves as
\begin{align}
    C_{\text{rad}}^{(i)}\frac{d T_{\text{rad}}^{(i)}}{dt} = U_{\mathrm{rw}}^{(i)}A_{\mathrm{rw}}^{(i)}\left(T_{\text{wat}}^{(i)} - T_{\text{rad}}^{(i)} \right) -  Q_{\mathrm{rad}}^{(i)}, \label{eq:rad_dynamics}
\end{align}
where the first term on the right-hand side represents the heat transferred from the water inside the radiator to the radiator metal, while the second term represents the radiator heat output. This heat output $Q_{\mathrm{rad}}^{(i)}$ is modeled as
\begin{align}
    Q_{\text{rad}}^{(i)} = Q_{\text{nom}}^{(i)}\left( \frac{ \max \left( T_{\text{rad}}^{(i)} - T_{\mathrm{z}}^{(j)},~0 \right) }{\Delta T_{\mathrm{N}}} \right)^{\alpha}, \label{eq:heat_output}
\end{align}
where $Q_{\mathrm{nom}}^{(i)} \, [\si{\watt}]$ is the nominal heat output from the datasheet, $T_{\mathrm{z}}^{(j)}$ is the temperature of the corresponding zone, $\alpha\in\mathbb{R}$ is the radiator exponent (typically $\alpha\approx 1.3$; see \cite{TUNZI2016413,en12173215}), and $\Delta T_{\mathrm{N}}$ is the nominal temperature difference between the radiator metal and the zone air (often $\Delta T_{\mathrm{N}} \approx 42~\si{\kelvin}$).

Moreover, the dynamics of the water node inside radiator~$i$ can be written as
\begin{align}
    C_{\mathrm{wat}}^{(i)}\,\frac{\mathrm{d}T_{\mathrm{wat}}^{(i)}}{\mathrm{d}t}
    &= \dot m_{\mathrm{wat}}^{(i)}\,c_{\mathrm{p},\mathrm{wat}}\bigl(T_{\mathrm{sup}}-T_{\mathrm{ret}}^{(i)}\bigr) + U_{\mathrm{rw}}^{(i)}A_{\mathrm{rw}}^{(i)}\bigl(T_{\mathrm{rad}}^{(i)} - T_{\mathrm{wat}}^{(i)}\bigr),
    \label{eq:water_dynamics_revised}
\end{align}
where the first term on the right-hand side represents the net heat supplied by the water flowing through the radiator, and the second term represents the heat transferred from the water to the radiator metal. As $T_{\mathrm{wat}}^{(i)}$ is approximated as the arithmetic mean of the supply and return temperatures, we have
\begin{align}
    T_{\text{ret}}^{(i)} = 2T_{\text{wat}}^{(i)} - T_{\text{sup}}. \nonumber%\label{eq:transient-eq2}
\end{align}
Note that, for brevity, we assume that there is a single heat source (e.g., a boiler). The model presented in this article can be readily extended to systems with multiple heat sources.

\begin{figure}
    \centering
    \includegraphics[width=0.52\linewidth]{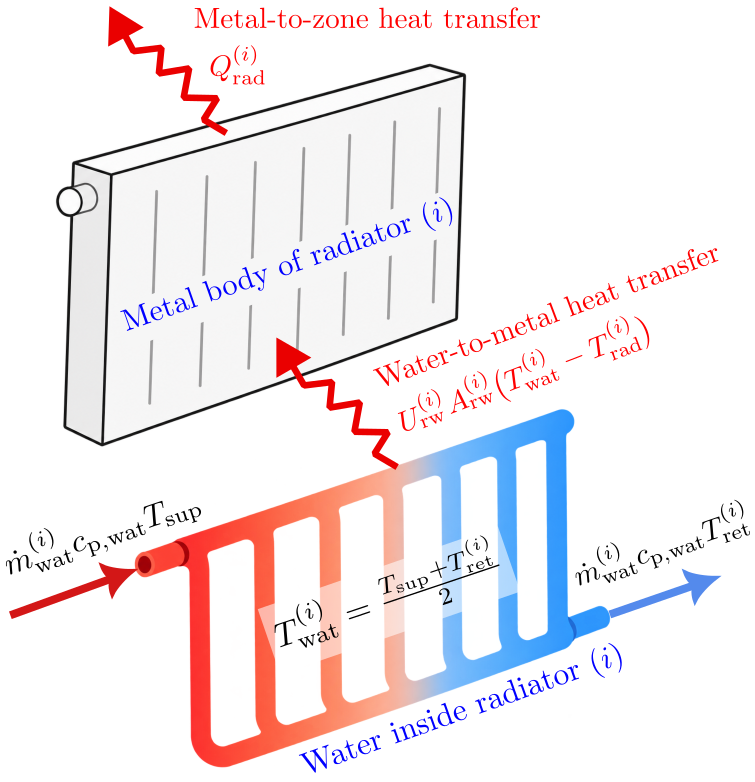}
    \caption{Two-node model of radiator dynamics.}
    \label{fig:water_rad}
\end{figure}

\begin{remark}
    Radiators are assumed to be installed in parallel (see Figure~\ref{fig:control_scheme}). Thus, assuming no heat losses in pipes, the water supply temperature $T_{\mathrm{sup}}$ is the same for all radiators.
\end{remark}

To compute the water-temperature dynamics of radiator~$i$ in~\eqref{eq:water_dynamics_revised} and thus consequently its metal temperature and the resulting heat output~\eqref{eq:heat_output}, the corresponding water mass flow rate $\dot{m}_{\mathrm{wat}}^{(i)}$ must be determined. This requires a hydraulic model of the radiator network, which is presented in Section~\ref{sec:hydraulic_model} below.

\subsubsection{Hydraulic Model}\label{sec:hydraulic_model} The typical configuration of a hydronic radiator heating system, consisting of a heat source (e.g., a boiler), connected radiators, and a circulation pump that drives the water flow, is shown in Figure~\ref{fig:control_scheme}. To compute the water flow rate through each radiator, we first model the dynamics of the water pump.
For water pumps, the pump head--flow rate ($H$--$Q$) curve is typically provided in the manufacturer's datasheet for different operating modes. By converting the volumetric flow rate $Q \; [\si{\meter^3 \per \second}]$ to the water mass flow rate $\dot{m}_{\mathrm{wat}} \; [\si{\kilogram \per \second}]$, the pump head $H \; [\si{\meter}]$ can be approximated as a quadratic function of $\dot{m}_{\mathrm{wat}}$, yielding
\begin{align}
    H = a_{\mathrm{p}} - b_{\mathrm{p}}\dot{m}_{\mathrm{wat}}^2, \label{eq:pump-head}
\end{align}
where $a_{\mathrm{p}}, b_{\mathrm{p}} \in \mathbb{R}$ are curve-fitting coefficients. In a parallel configuration (see Figure \ref{fig:control_scheme}), each radiator experiences approximately the same pressure drop $\Delta p \; [\si{\pascal}]$, which can be written as
\begin{align}
    \Delta p = K^{(i)} {\left(\dot{m}_{\mathrm{wat}}^{(i)}\right)}^2, \label{eq:delta-p}
\end{align}
where $K^{(i)} \; [\si{\pascal \square \second \per \square \kilogram}]$ denotes the hydraulic resistance of radiator~$i$, $i \in \{1,2,\hdots, n_{\mathrm{rad}}\}$, computed as \eqref{eq:hydraulic_resistance} below. Thus, the overall flow rate $\dot{m}_{\mathrm{wat}}$ can be written as
\begin{align}
    \dot{m}_{\mathrm{wat}} = \sum_{i=1}^{n_{\mathrm{rad}}} \dot{m}_{\mathrm{wat}}^{(i)}  = \sqrt{\Delta p} \underbrace{\sum_{i=1}^{n_{\mathrm{rad}}} \frac{1}{\sqrt{K^{(i)}}} }_{\coloneqq \frac{1}{\sqrt{K_{\mathrm{eq}}}}}. \label{eq:sum-q}
\end{align}
Converting \eqref{eq:pump-head} to pressure form and substituting \eqref{eq:sum-q} yields
\begin{align}
    \rho_\mathrm{wat} g a_{\mathrm{p}} - \rho_\mathrm{wat} g b_{\mathrm{p}} \left( \frac{\sqrt{\Delta p} }{\sqrt{K_{\mathrm{eq}}}} \right) ^2 = \Delta p,\label{eq:delta-p-2}
\end{align}
where $\rho_{\mathrm{wat}} \; [\si{\kilogram \per \cubic \meter}]$ is the density of water (assumed constant), and $g \; [\si{\meter \per \square \second}]$ is the standard acceleration due to gravity. By rewriting \eqref{eq:delta-p-2} and combining it with \eqref{eq:delta-p}, the water flow rate $\dot{m}_{\mathrm{wat}}^{(i)}$ through radiator~$i$, $i \in \{1,2,\hdots, n_{\mathrm{rad}}\}$, is given by
\begin{align}
    \dot{m}_{\mathrm{wat}}^{(i)} = \sqrt{\frac{\rho_\mathrm{wat} ga_{\mathrm{p}}}{1 + \rho_\mathrm{wat} gb_{\mathrm{p}} \frac{1}{K_{\mathrm{eq}}}}} \frac{1}{\sqrt{K^{(i)}}}. \label{eq:flow_rate}
\end{align}

To compute $\dot{m}_{\mathrm{wat}}^{(i)}$ using \eqref{eq:flow_rate}, the hydraulic resistances of all radiators in the system must first be determined. For each radiator equipped with a controllable valve, the hydraulic resistance $K^{(i)}$ depends on the valve position, which is determined by a local PI controller based on the temperature error between the setpoint and the measured zone temperature and the time integral of this error. In particular, let the temperature error of radiator~$i$ be
\begin{align}
e^{(i)} \coloneqq T_{\mathrm{des}}^{(i)} - T_{\mathrm{z}}^{(i)}, \nonumber
\end{align}
where $T_{\mathrm{des}}^{(i)}$ is the temperature setpoint. Then, the unsaturated valve
position is
\begin{align}
\phi^{(i)} = G_{\mathrm{p}}^{(i)}\,e^{(i)} + I^{(i)},
\label{eq:pi_command}
\end{align}
where $G_{\mathrm{p}}^{(i)} \; [\si{\per\kelvin}]$ is the proportional gain and
$I^{(i)}$ is the integral state.
The applied valve position $\eta^{(i)}$ is the saturation of $\phi^{(i)}$ to the admissible
range $[0,~1] \subset \mathbb{R}$, where $1$ corresponds to the fully open valve. Thus,
\begin{align}
\eta^{(i)} 
\coloneqq
\begin{cases}
0, & \phi^{(i)} < 0,\\
\phi^{(i)}, & 0 \le \phi^{(i)} \le 1,\\
1, & \phi^{(i)} > 1.
\end{cases}
\label{eq:valve_position}
\end{align} 
Additionally, the integral state evolves according to the back-calculation anti-windup law \begin{align}
\dot{I}^{(i)} &= \frac{G_{\mathrm{p}}^{(i)}e^{(i)} + \eta^{(i)}-\phi^{(i)}}{\tau^{(i)}} = \frac{\eta^{(i)}-I^{(i)}}{\tau^{(i)}}, 
\label{eq:integral_dynamics} 
\end{align} 
where $\tau^{(i)}\,[\si{\second}]$ denotes the integral and anti-windup tracking time. In particular, when the valve is unsaturated, $\eta^{(i)}=\phi^{(i)}$, \eqref{eq:integral_dynamics} reduces to 
\begin{align} 
\dot{I}^{(i)} = \frac{G_{\mathrm{p}}^{(i)}}{\tau^{(i)}}e^{(i)}, \nonumber
\end{align} 
thus, \eqref{eq:pi_command}--\eqref{eq:integral_dynamics} recover the standard PI control law. In contrast, when the valve is saturated, the correction term $\eta^{(i)}-\phi^{(i)}$ drives the integral state toward the saturated valve opening, thereby preventing integrator windup. 

To avoid nonsmooth dynamics, \eqref{eq:valve_position} is
approximated as
\begin{align}
\tilde{\eta}^{(i)}
\coloneqq \sigma_\varepsilon(\phi^{(i)}) - \sigma_\varepsilon(\phi^{(i)} - 1),
\label{eq:valve_position_smooth}
\end{align}
where
\begin{align}
\sigma_\varepsilon(y) \coloneqq \tfrac{1}{2}\Bigl(y + \sqrt{y^2 + \varepsilon}\Bigr), \nonumber
\end{align}
with $\varepsilon \in \mathbb{R}_{>0}$ being sufficiently small.
Thus, the hydraulic
resistance $K^{(i)}$ of radiator~$i$ is computed as
\begin{align}
K^{(i)}
&=
\bigg(
\frac{1}{\tilde{\eta}^{(i)}}\frac{3600\sqrt{10^5}}{\rho_{\mathrm{wat}}\,k_{\mathrm{v}}^{(i)}}
\bigg)^2
+
\bigg(
\frac{3600\sqrt{10^5}}{\rho_{\mathrm{wat}}\,k_{\mathrm{v}_0}^{(i)}}
\bigg)^2,
\label{eq:hydraulic_resistance}
\end{align}
where $k_{\mathrm{v}}^{(i)} \; [\si{\meter^3 / \hour}]$ is the nominal flow coefficient of the fully open controllable valve, and $k_{\mathrm{v}_0}^{(i)} \; [\si{\meter^3 / \hour}]$ is the equivalent fixed flow coefficient associated with radiator $i$, collectively representing the pressure losses across the radiator, pipes, and other fixed hydraulic restrictions.
\begin{remark}
    Flow coefficients are standard parameters in datasheets that specify the volumetric flow rate of water through the corresponding components at a pressure drop of \SI{1}{\bar}. In particular,
    \begin{align}
    q_{\mathrm{h}} = 
    k_{\mathrm{v}}
    \sqrt{\frac{\Delta p_{\mathrm{v}}}{10^{5}}}, \nonumber
    \end{align}
    where $q_{\mathrm{h}} \, [\si{\meter^3/ \hour}]$ is the volumetric flow rate, and $\Delta p_{\mathrm{v}} \, [\si{\pascal}]$ is the pressure drop across the corresponding component, and $k_{\mathrm{v}} \; [\si{\meter^3 / \hour}]$ denotes the nominal flow coefficient. Thus, in \eqref{eq:hydraulic_resistance}, the factor $3600$
    accounts for the conversion from hours to seconds, and $\rho_{\mathrm{wat}}$ is used to
    convert the volumetric flow rate to the mass flow rate used
    in~\eqref{eq:flow_rate}. Moreover, since $k_{\mathrm{v}}^{(i)}$ in \eqref{eq:hydraulic_resistance}
    corresponds to the fully open controllable valve, its effective flow
    coefficient is assumed to vary proportionally with the valve position,
    i.e., $\tilde{\eta}^{(i)}k_{\mathrm{v}}^{(i)}$.
\end{remark}

\begin{remark}
    It can be observed from \eqref{eq:flow_rate} and \eqref{eq:hydraulic_resistance} that the mass flow rate $\dot{m}^{(i)}_{\mathrm{wat}}$ through radiator $i$ depends not only on its own valve position $\eta^{(i)}$, but also on the positions of the other valves connected to the same circulating pump.
\end{remark}

In the next subsection, we express all the models developed in the previous subsections in a conventional state-space representation so that they can be more clearly incorporated into the MPC framework.

\subsection{State-Space Representation}
In compact form, the overall system can be written as
\begin{equation}
\dot{x}(t) = f\big(x(t), u(t), d(t)\big), \label{eq:overall_ct_system}
\end{equation}
where $f: \mathbb{R}^{n} \times \mathbb{R}^{n_{\mathrm{u}}} \times \mathbb{R}^{n_{\mathrm{d}}} \rightarrow \mathbb{R}^{n}$ is the vector field obtained by combining the building, radiator, and valve dynamics, and the stacked state vector $x$ is given by
\begin{align}
x(t) \coloneqq
\bigl [T_{\mathrm{z}}^{\top}(t) ~~ T_{\mathrm{w}}^{\top}(t) ~~ T_{\mathrm{rad}}^{\top}(t) ~~ T_{\mathrm{wat}}^{\top}(t) ~~ I^{\top}(t) \bigr ]
^{\top} \in \mathbb{R}^n. \nonumber
\end{align}
Here,
\begin{equation*}
\begin{alignedat}{2}
T_{\mathrm{z}}(t)
&\coloneqq \bigl[T_{\mathrm{z}}^{(i)}(t)\bigr]_{i=1}^{n_{\mathrm{z}}}, 
&\qquad
T_{\mathrm{w}}(t)
&\coloneqq \bigl[T_{\mathrm{w}}^{(i)}(t)\bigr]_{i=1}^{n_{\mathrm{w}}}, \\[2mm]
T_{\mathrm{rad}}(t)
&\coloneqq \bigl[T_{\mathrm{rad}}^{(i)}(t)\bigr]_{i=1}^{n_{\mathrm{rad}}},
&\qquad
T_{\mathrm{wat}}(t)
&\coloneqq \bigl[T_{\mathrm{wat}}^{(i)}(t)\bigr]_{i=1}^{n_{\mathrm{rad}}}, \\[2mm]
I(t)
&\coloneqq \bigl[I^{(i)}(t)\bigr]_{i=1}^{n_{\mathrm{rad}}},
\end{alignedat}
\end{equation*}
where the notation $\bigl [T_{\mathrm{z}}^{(i)} (t) \bigr]_{i=1}^{n_{\mathrm{z}}} \in \mathbb{R}^{n_{\mathrm{z}}}$, e.g., denotes the column vector obtained by stacking the components $T_{\mathrm{z}}^{(i)}$, ~$i\in \{1,2,\ldots, n_{\mathrm{z}}\}$.
Moreover, the input is defined as 
\begin{align}
u(t) := \begin{bmatrix}
  T_{\text{sup}}(t) &  T_{\mathrm{des}}^{\top}(t)
\end{bmatrix}^{\top} \in \mathbb{R}^{n_{\mathrm{u}}}, \nonumber
\end{align}
where $n_{\mathrm{u}} \coloneqq 1 + n_{\mathrm{rad}}$, and $T_{\mathrm{des}}(t) \coloneqq \bigl[ T_{\mathrm{des}}^{(i)}(t) \bigr]_{i = 1}^{n_{\mathrm{rad}}}$. Recall that the vector of temperature setpoints $T_{\mathrm{des}}$ controls the valve position, as discussed in Section~\ref{sec:hydraulic_model}. Additionally, the exogenous input is
\begin{align}
    d(t) := \begin{bmatrix}
        T_{\text{o}}(t) & Q_{\text{int}}^{\top}(t) & q_{\text{sol}}^{\top}(t)
    \end{bmatrix}^{\top} \in \mathbb{R}^{ n_{\mathrm{d}} }, \nonumber
\end{align}
where $Q_{\mathrm{int}}(t) \coloneqq \bigl [Q_{\mathrm{int}}^{(i)}(t) \bigr]_{i = 1}^{n_{\mathrm{z}}}$, $q_{\mathrm{sol}}(t) \coloneqq \bigl [q_{\mathrm{sol}}^{(i)}(t) \bigr]_{i = 1}^{n_{\mathrm{w}}}$, and $n_{\mathrm{d}} \coloneqq 1 + n_{\mathrm{z}} + n_{\mathrm{w}}$.
\begin{remark}
    Walls with identical orientation (e.g., surface azimuth) receive the same solar irradiation, so only unique values need to be stored in $q_{\text{sol}}(t)$ to avoid redundancy.
\end{remark}

Thus far, we have discussed the general modeling of buildings heated by hydronic radiators. We now consider a case-study building and formulate its model in the form of \eqref{eq:overall_ct_system}.

\subsection{Case-Study Modeling}
In this article, the case study is the Rensen building, located in Uden, the Netherlands. The exterior of the building is shown in Figure~\ref{fig:exterior_scheme}. We consider only the first floor of the building, which consists of six zones. In this subsection, we discuss how the building-related system parameters in \eqref{eq:room_temp} and  \eqref{eq:wall_temp} and the radiator-related parameters in \eqref{eq:rad_dynamics} and \eqref{eq:water_dynamics_revised} are obtained for this building, as well as the number of states considered in the overall system \eqref{eq:overall_ct_system}.

\begin{figure}
     \centering
     \begin{subfigure}{0.55\linewidth}
         \centering
         \includegraphics[width=\linewidth]{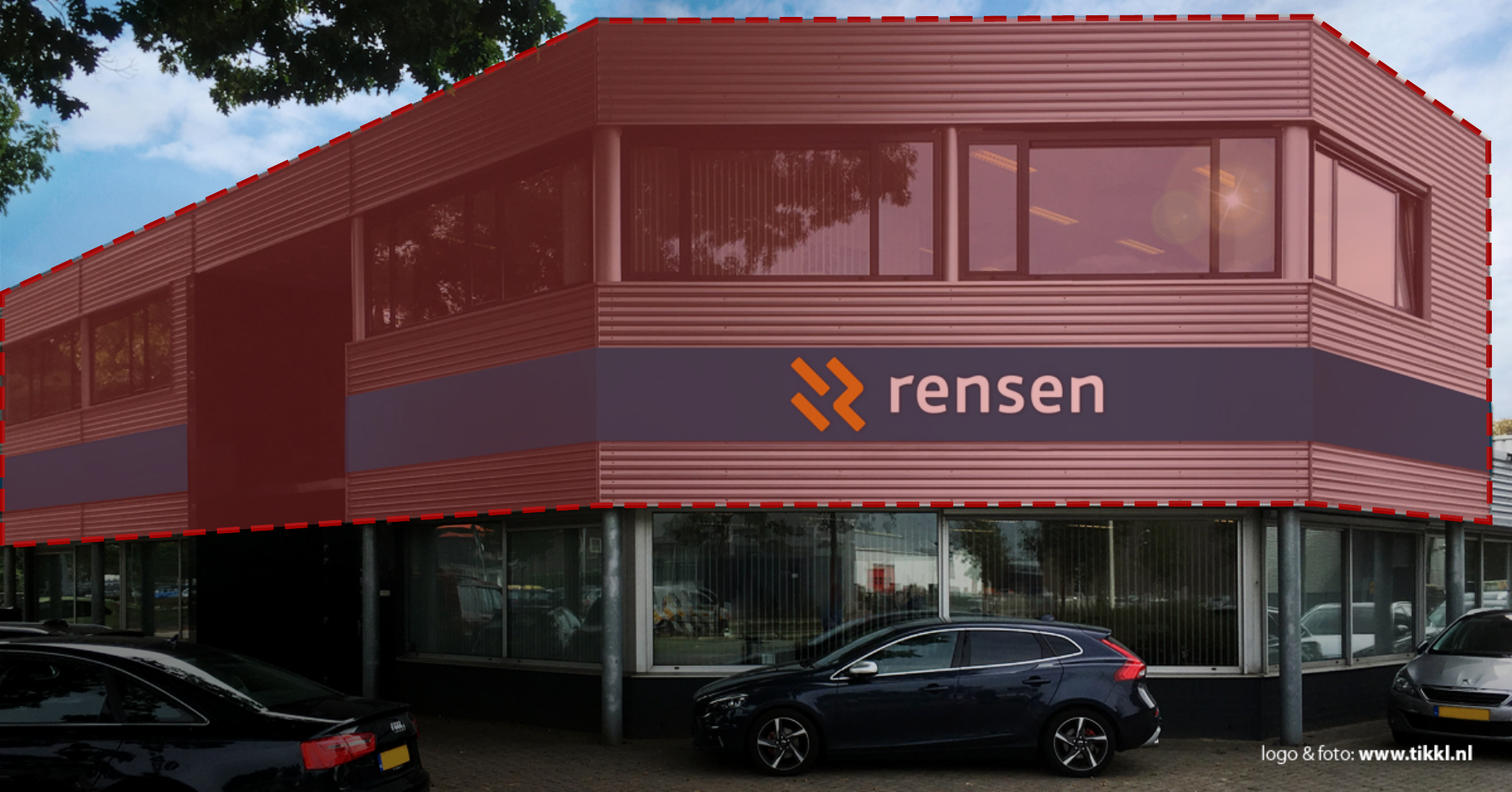}
         \caption{Exterior view of the Rensen building, highlighting the first floor considered in the case study.}
         \label{fig:exterior_case_study}
     \end{subfigure}
     \hfill
    \begin{subfigure}{0.78\linewidth}
        \centering
        \includegraphics[
        width=0.8\linewidth
        ]{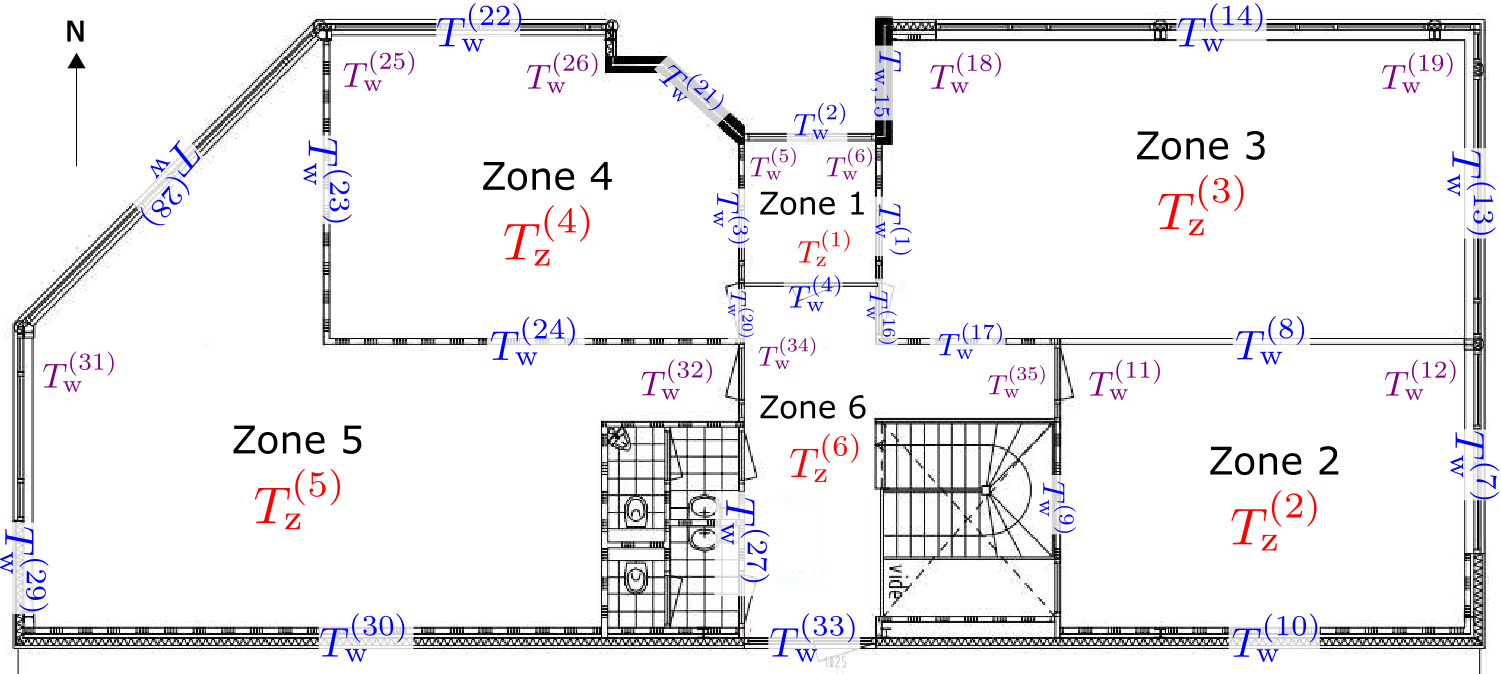}
        \caption{Floor plan and corresponding states of the first floor of the building considered in the case study.}
        \label{fig:scheme_building}
    \end{subfigure}
        \caption{The experimental case study: the Rensen Building, located in the Netherlands.}
    \label{fig:exterior_scheme}
\end{figure}

\subsubsection{Building-Related System Parameters}

The floor plan and the building-related states, including the zone and wall temperatures, are shown in Figure~\ref{fig:scheme_building}. In this figure, the wall-temperature states shown in violet correspond to the ceiling and floor. Overall, the model includes six zone-temperature states, $n_{\mathrm{z}} \coloneqq 6$, and 35 wall-temperature states, $n_{\mathrm{w}} \coloneqq 35$.

To obtain the system parameters in \eqref{eq:room_temp} and \eqref{eq:wall_temp}, such as the wall thermal resistances $R_{\mathrm{w}}^{(i)}$ and the heat capacities of the zones $C_{\mathrm{z}}^{(i)}$ and walls $C_{\mathrm{w}}^{(i)}$, we use the building construction drawings. For example, Figure~\ref{fig:ext_wall} illustrates the material composition of the exterior walls. The exterior walls consist of a
$0.75~\mathrm{mm}$ steel sheet, $130~\mathrm{mm}$ Rockwool insulation with
$R_{\mathrm{rw}} = 2.5~\mathrm{m^2K/W}$, a $0.70~\mathrm{mm}$ steel inner tray, and a
$12.5~\mathrm{mm}$ gypsum board. Since the temperature of wall $i$ in \eqref{eq:wall_temp} is considered at its midplane, the thermal resistance~$R_{\mathrm{w}}^{(i)}$ is
computed as 
\begin{align}
R_{\mathrm{w}}^{(i)} = \frac{R_{\mathrm{w,tot}}^{(i)}}{2}, \nonumber
\end{align}
where
\begin{align}
R_{\mathrm{w}, \mathrm{tot}}^{(i)}
&=
\frac{0.00075}{\lambda_{\mathrm{steel}}}
+
R_{\mathrm{rw}}
+
\frac{0.00070}{\lambda_{\mathrm{steel}}}
+
\frac{0.0125}{\lambda_{\mathrm{gypsum}}} 
=
\frac{0.00075}{50}
+
2.5 \nonumber \\
&\quad +
\frac{0.00070}{50}
+
\frac{0.0125}{0.25}
\approx 2.55~\mathrm{m^2K/W}, \nonumber
\end{align}
with $\lambda_{\mathrm{steel}}~[\mathrm{W/(mK)}]$ and
$\lambda_{\mathrm{gypsum}}~[\mathrm{W/(mK)}]$ being the thermal conductivities
of steel and gypsum board, respectively. Moreover, the heat capacity $C_{\mathrm{w}}^{'(i)} ~[\mathrm{J/(m^2K)}]$ of wall $i$ per unit area is computed as
\begin{align}
C_{\mathrm{w}}^{'(i)}
&=
\rho_{\mathrm{steel}} c_{\mathrm{p},\mathrm{steel}}(0.00075+0.00070)
+
\rho_{\mathrm{rw}} c_{\mathrm{p},\mathrm{rw}}(0.13)
+
\rho_{\mathrm{gypsum}} c_{\mathrm{p},\mathrm{gypsum}}(0.0125) \nonumber \\
&=
7850 \cdot 500 \cdot 0.00145
+
50 \cdot 840 \cdot 0.13
+
800 \cdot 1090 \cdot 0.0125 \approx 2.2 \times 10^{4}~\mathrm{J/(m^2K)}, \nonumber
\end{align}
where $\rho_{\mathrm{steel}}~[\mathrm{kg/m^3}]$,
$\rho_{\mathrm{rw}}~[\mathrm{kg/m^3}]$, and
$\rho_{\mathrm{gypsum}}~[\mathrm{kg/m^3}]$ denote the densities of steel,
Rockwool, and gypsum board, respectively, and
$c_{\mathrm{p},\mathrm{steel}}~[\mathrm{J/(kgK)}]$,
$c_{\mathrm{p},\mathrm{rw}}~[\mathrm{J/(kgK)}]$, and
$c_{\mathrm{p},\mathrm{gypsum}}~[\mathrm{J/(kgK)}]$ denote their specific heat
capacities, respectively. Thus, the heat capacity $C_{\mathrm{w}}^{(i)}$ of wall $i$ is
\begin{align}
C_{\mathrm{w}}^{(i)} = A_{\mathrm{w}}^{(i)} C_{\mathrm{w}}^{'(i)}. \nonumber
\end{align}
\begin{figure}
    \centering
    \includegraphics[width=0.52\linewidth]{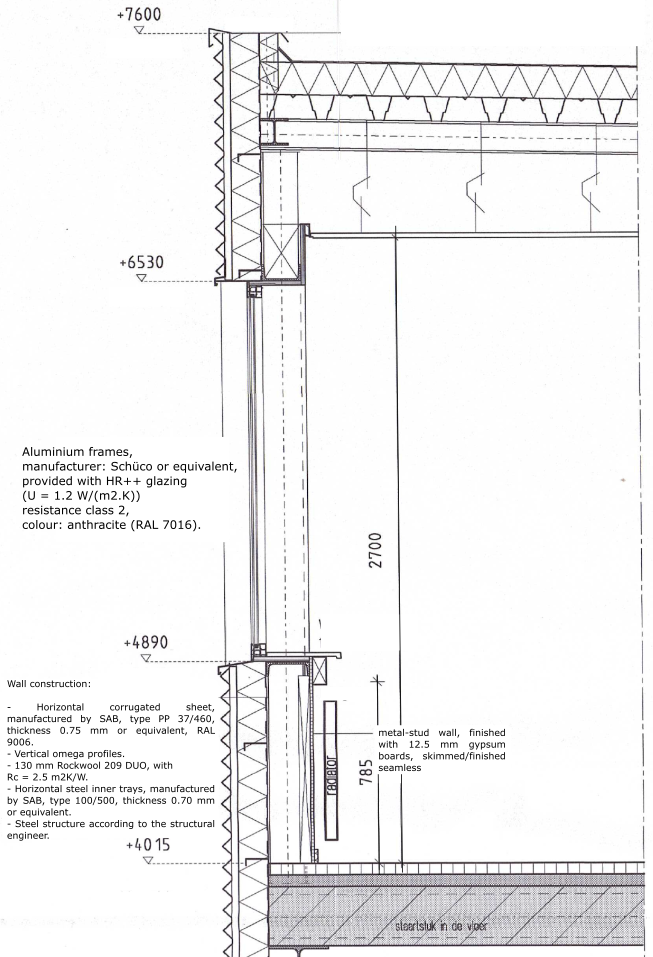}
    \caption{Material composition of the exterior walls and properties of the windows.}
    \label{fig:ext_wall}
\end{figure}
The same procedure is applied to the interior walls, ceiling, and floor. Regarding the windows and based on the building construction drawing in Figure \ref{fig:ext_wall}, the thermal resistance $R_{\mathrm{wd}}^{(i)}$ of window $i$, $i \in \{1,2,\ldots, 8\}$, is obtained as
\begin{align}
    R_{\mathrm{wd}}^{(i)} = \frac{1}{U} = \frac{1}{1.2} = 0.83~\mathrm{m^2K/W}. \nonumber
\end{align}
Moreover, because the windows are HR++ double glazed, the solar heat gain coefficient and secondary heat gain fraction are estimated as $f^{(i)}_{\mathrm{SHGC}}=0.55$ and $f^{(i)}_{\mathrm{sec}}=0.8$, $i \in \{1,2,\ldots, 8\}$, respectively.

Additionally, the heat capacity $C_{\mathrm{z}}^{(i)}$ of zone $i$, $i \in \{1,2,\ldots, 6\}$, is approximated based on the thermal capacitance of the air contained in the zone as
\begin{align}
    C_{\mathrm{z}}^{(i)}
    \coloneqq
     \alpha^{(i)}_{\mathrm{z}}\rho_{\mathrm{air}} c_{\mathrm{p},\mathrm{air}} V_{\mathrm{z}}^{(i)}, \nonumber
\end{align}
where $\rho_{\mathrm{air}} \, [\si{\kilogram \meter ^{-3}}]$ is the density of air, and $c_{\mathrm{p},\mathrm{air}} \, [\si{\joule \per \kilogram \per \kelvin}]$ is the specific heat capacity of air, and $V_{\mathrm{z}}^{(i)} \, [\si{\meter ^3}]$ is the volume of zone $i$. Moreover, $\alpha^{(i)}_{\mathrm{z}} \in \mathbb{R}_{> 0}$ denotes the effective zone thermal-capacitance coefficient, typically ranging from 6 to 10 for office buildings \citep{HONG201923}.

\subsubsection{Radiator-Related System Parameters}
The radiator configuration of the considered floor is as follows:
\begin{itemize}
\item Zone 1: one radiator with a nominal heat output of $\qty{1600}{\watt}$.
\item Zone 2: one radiator with a nominal heat output of $\qty{1706}{\watt}$.
\item Zone 3: four radiators, each with a nominal heat output of $\qty{1706}{\watt}$.
\item Zone 4: one radiator with a nominal heat output of $\qty{4785}{\watt}$.
\item Zone 5: three radiators with nominal heat outputs of $\qty{1706}{\watt}$, $\qty{2047}{\watt}$, and $\qty{2047}{\watt}$.
\item Zone 6: no radiator.
\end{itemize}
All radiators are supplied by a single boiler, namely a Remeha Quinta 35c equipped with a Grundfos UPMO 25-60 water pump; see Figure~\ref{fig:water_pump}. The pump is running at its lowest speed (the third light of the water pump is yellow in Figure~\ref{fig:water_pump}), so we use the constant $H$--$Q$ Curve~I from the datasheet \citep{grundfos_upm3}. By converting the volumetric flow rate $Q$ into the mass flow rate $\dot{m}_{\mathrm{wat}}$ and fitting a quadratic function, we obtain the coefficients $a_{\mathrm{p}}$ and $b_{\mathrm{p}}$ in \eqref{eq:pump-head}, as shown in Figure~\ref{fig:HQ_water_pump}.

\begin{figure}
     \centering
     \begin{subfigure}[t]{0.49\linewidth}
        \centering
        \vspace{0pt}
        \includegraphics[
            height=5cm,
            trim={0cm 0cm 0cm 0cm},
            clip
        ]{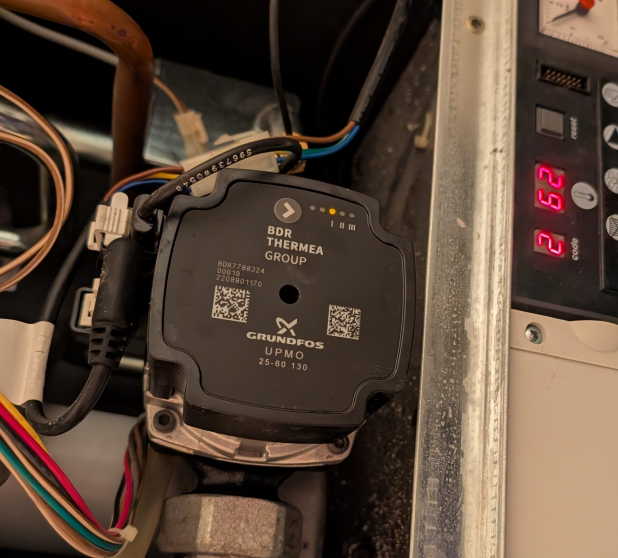}
        \caption{Water pump inside the boiler, which circulates water through the radiators on the first floor of the building.}
        \label{fig:water_pump}
     \end{subfigure}
      \hfill
     \begin{subfigure}[t]{0.49\linewidth}
        \centering
        \vspace{0pt}
        \includegraphics[
            height=5cm,
            trim={10cm 25cm 0cm 30cm},
            clip
        ]{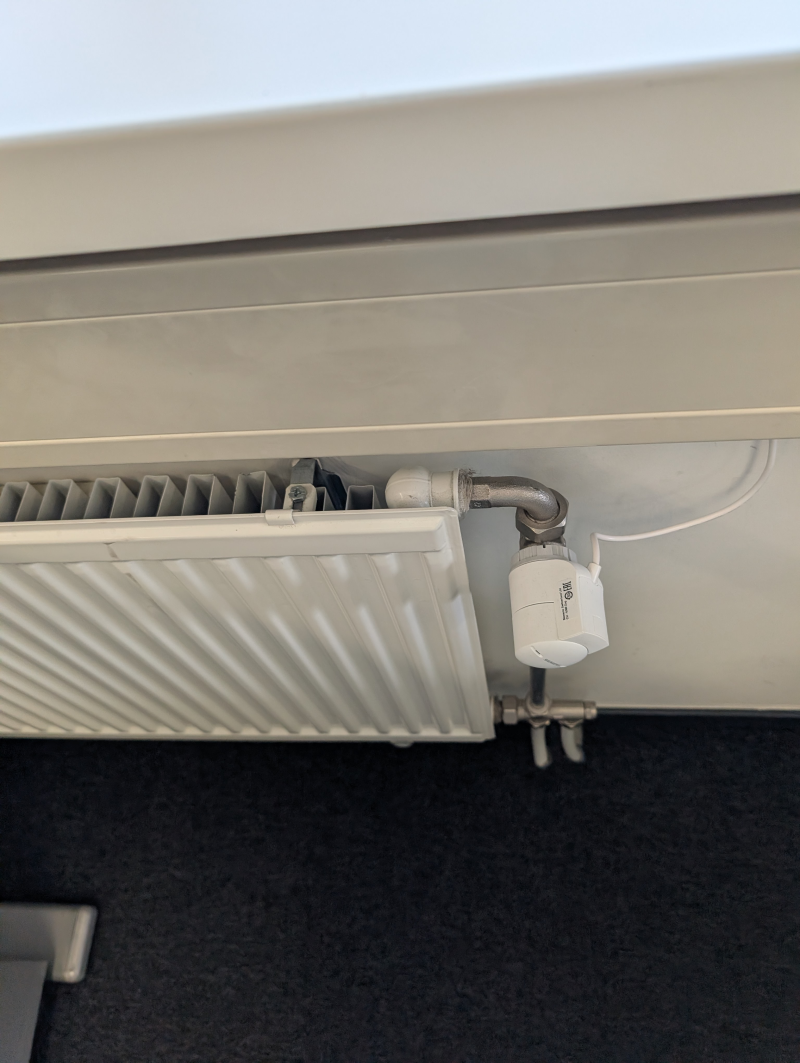}
        \caption{One of the controllable radiator valves, which acts as a PI controller with a remotely configurable setpoint.}
        \label{fig:valve}
     \end{subfigure}
     \caption{Main components of the first-floor hydronic heating system considered in the case study.}
\end{figure}

\begin{figure}
    \centering
    \includegraphics[width=0.55\linewidth]{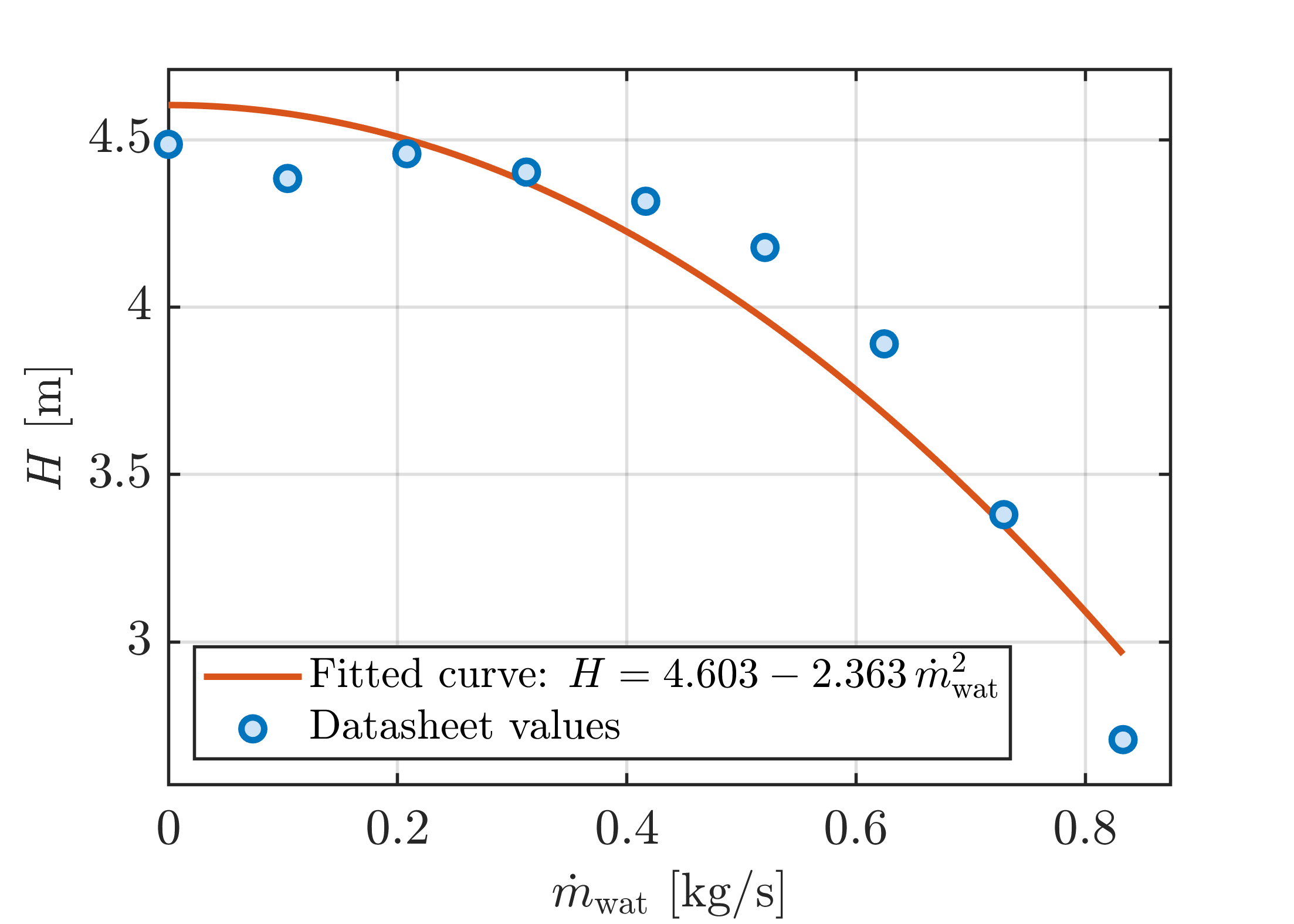}
    \caption{Identification of the water-pump parameters in \eqref{eq:pump-head} for the case-study building.}
    \label{fig:HQ_water_pump}
\end{figure}

In the case-study building, the valves of all radiators in zones 2, 3, 4, and 5 are remotely controllable, whereas the radiator in zone 1 is equipped with a conventional TRV. Moreover, all controllable valves within the same zone share a common temperature setpoint. Therefore, the number of considered control inputs is
$n_{\mathrm{u}} = 1 + 4 = 5$,
where one input corresponds to the boiler supply temperature $T_{\mathrm{sup}}$ and the remaining four inputs correspond to the setpoints $T_{\mathrm{des}}^{(i)}$, $i \in \{2,3,4,5\}$. Figure~\ref{fig:valve} shows one of the remotely controllable radiator valves. %in Room 3.
In addition, since all controllable valves in the same zone exhibit the same behavior because of their common setpoints, one valve-integral state~$I^{(j)}$ is assigned to each controlled zone $j$. Consequently, the model includes 24 radiator-related states. Together with the zone- and wall-temperature states, the total number of states in the system \eqref{eq:overall_ct_system} is~65.

However, finding accurate values for some of the parameters introduced in this section, such as the flow coefficients $k^{(i)}_\mathrm{v}$ and $k^{(i)}_{\mathrm{v}_0}$ of the radiator branches in \eqref{eq:hydraulic_resistance}, can be challenging. Moreover, although some parameters can be estimated based on physical properties, their actual values may deviate from their nominal values over time due to aging, faults, or other factors. For example, the thermal resistances $R_{\mathrm{w}}^{(i)}$, $i \in \{1,2,\ldots,n_{\mathrm{w}}\}$, of the walls are estimated from the wall materials, but the wall color may also affect their effective thermal resistance. Therefore, in the next section, we use real measurement data to slightly adjust these parameters so that the model better matches the actual system behavior. This approach is referred to as \textit{grey-box modeling}.

%% file: symbols.tex
\begin{table}[pos=htbp]
\centering
\caption{Thermal model parameters and variables.}
\label{tab:symbols_thermal_model}
\setlength{\tabcolsep}{6pt}
\renewcommand{\arraystretch}{1.15}
%\begin{tabularx}{\linewidth}{@{} M{\symw} l X @{}} \toprule
\begin{tabularx}{\linewidth}{@{} l l X @{}}\toprule
\textbf{Symbol} & \textbf{Unit} & \textbf{Description} \\ 
\midrule

$A_{\mathrm{w}}^{(i)}$ & $\si{\square\meter}$ & Area of wall $i$, excluding window areas \\ \rowline

$A_{\mathrm{wd}}^{(i)}$ & $\si{\square\meter}$ & Area of window(s) located on wall $i$ (if none, $A_{\mathrm{wd}}^{(i)} = 0$) \\ \rowline

$A_{\mathrm{rw}}^{(i)}$ & $\si{\square\meter}$ & Heat-transfer area between the water inside radiator~$i$ and its metal surface
%Radiator–water heat-transfer area of radiator $i$ 
\\ \rowline

$C_{\mathrm{z}}^{(i)}$ & $\si{\joule\per\kelvin}$ & Heat capacity of zone $i$ \\ \rowline

$C_{\mathrm{w}}^{(i)}$ & $\si{\joule\per\kelvin}$ & Heat capacity of wall $i$ \\ \rowline

$C_{\text{rad}}^{(i)}$ & $\si{\joule\per\kelvin}$ & Heat capacity of radiator $i$ \\ \rowline

$C_{\text{wat}}^{(i)}$ & $\si{\joule\per\kelvin}$ & Heat capacity of water inside radiator $i$ \\ \rowline

$c_{\mathrm{p},\mathrm{wat}}$ & $\si{\joule \per \kilogram \per \kelvin}$ & Specific heat capacity of water \\ \rowline

$f_{\text{SHGC}}^{(i)}$ & --& Solar heat gain coefficient of window(s) on wall~$i$, defined as the fraction of incident solar radiation on the glazing that contributes to indoor heat gain (directly transmitted and inward secondary heat gains), $f_{\text{SHGC}}^{(i)}\in[0,1] \subset \mathbb{R}$; see Figure~\ref{fig:solar_window} \\ \rowline

$f_{\text{sec}}^{(i)}$ & -- & Secondary heat gain fraction of window(s) on wall $i$ (fraction of the window solar heat gain released to the zone air), $f_{\text{sec}}^{(i)}\in[0,1] \subset \mathbb{R}$; see Figure~\ref{fig:solar_window} \\ \rowline

$G_{\mathrm{inf}}^{(i)}$ & $\si{\watt\per\kelvin}$ & Infiltration conductance of zone $i$ \\ \rowline

$K^{(i)}$ & $\si{\pascal \square \second \per \square \kilogram}$ & Hydraulic resistance of radiator~$i$ \\ \rowline

$\dot{m}_{\mathrm{wat}}^{(i)}$ & $\si{\kilogram \per \second}$ & Water mass flow rate through radiator~$i$ \\ \rowline

$\mathcal{N}_{\mathrm{z}}^{(i)}$ & -- & Set of indices of zones adjacent to wall $i$; index~$0$ corresponds to the outdoor environment \\ \rowline

$\mathcal{N}_{\mathrm{w}}^{(i)}$ & -- & Set of indices of walls surrounding zone $i$ \\ \rowline

$\mathcal{N}_{\mathrm{rad}}^{(i)}$ & -- & Set of indices of radiators in zone $i$ \\ \rowline

$n_{\mathrm{z}}$ & -- & Number of zones \\ \rowline

$n_{\mathrm{w}}$ & -- & Number of walls \\ \rowline

$n_{\mathrm{rad}}$ & -- & Number of  radiators \\ \rowline %\\ \rowline

$Q_{\mathrm{rad}}^{(i)}$ & $\si{\watt}$ & Heat transfer rate supplied by radiator $i$ \\ \rowline

$Q_{\mathrm{int}}^{(i)}$ & $\si{\watt}$ & Internal heat gain rate in zone $i$ \\ \rowline

$q_{\mathrm{sol}}^{(i)}$ & $\si{\watt\per\square\meter}$ & Solar irradiance incident on wall $i$ and its window(s) \\ \rowline

$R_{\mathrm{w}}^{(i)}$ & $\si{\kelvin \square \meter \per\watt}$ & Thermal resistance of wall $i$ \\ \rowline

$R_{\mathrm{wd}}^{(i)}$ & $\si{\kelvin \square \meter\per\watt}$ & Thermal resistance of window(s) on wall $i$ \\ \rowline

$T_{\mathrm{des}}^{(i)}$ & $\si{\kelvin}$ & Zone-temperature setpoint for the valve controller of radiator $i$ \\ \rowline

$T_{\mathrm{o}}$ & $\si{\kelvin}$ & Outdoor temperature \\ \rowline

$T_{\mathrm{rad}}^{(i)}$ & $\si{\kelvin}$ & Temperature of radiator $i$ \\ \rowline

$T_{\mathrm{ret}}^{(i)}$ & $\si{\kelvin}$ & Return water temperature from radiator $i$ \\ \rowline

$T_{\mathrm{sup}}$ & $\si{\kelvin}$ & Water supply temperature \\ \rowline %of boiler $i$ \\ \rowline

$T_{\mathrm{w}}^{(i)}$ & $\si{\kelvin}$ & Temperature of wall $i$ \\ \rowline

$T_{\mathrm{wat}}^{(i)}$ & $\si{\kelvin}$ & Average water temperature in radiator~$i$ \\ \rowline

$T_{\mathrm{z}}^{(i)}$ & $\si{\kelvin}$ & Temperature of zone $i$ \\ \rowline

%$T_{\mathrm{sen},i}$ & $\si{\kelvin}$ & Temperature measured by the TRV in zone $i$ \\ \rowline

%$A_{\mathrm{z},i}$ & $\si{\square\meter}$ & Floor area of zone $i$ \\ \rowline

$U_{\mathrm{rw}}^{(i)}$ & $\si{\watt \per \kelvin \per \square \meter}$ & Heat-transfer coefficient between the water inside radiator $i$ and its metal surface \\ \rowline

$\eta^{(i)}$ & -- &  Valve position of radiator $i$, $\eta^{(i)} \in[0,1] \subset \mathbb{R}$ \\ \rowline

$\gamma_{\mathrm{abs}}^{(i)}$ & -- & \mbox{Solar absorptance of wall $i$, $\gamma_{\mathrm{abs}}^{(i)} \in[0,1] \subset \mathbb{R}$} \\

% & -- & Number of  \\
%$n_{\mathrm{wd}}^{(i)}$ & -- & The number of windows in zone $i$ \\
\bottomrule
\end{tabularx}
\end{table}

%% file: s3_new.tex
The idea of grey-box modeling is to find optimal values for the system parameters, constrained around the values estimated from physical properties, such that the difference between the model outputs and the real measurements is minimized. In the case-study building, at each sampling instant $t_k \coloneqq kt_s$, where $t_s \in \mathbb{R}_{> 0}$ is the sampling time and $k \in \mathbb{N}_0$, the temperatures of zones 2, 3, 4, and 5, denoted by $T_{\mathrm{meas},k}^{(i)}$, $i \in \mathcal{Z} \coloneqq \{2,3,4,5\}$, the boiler supply temperature $T^{\mathrm{sup}}_{\mathrm{meas},k}$, and the valve positions $\eta^{(i)}_{\mathrm{meas},k}$, $i \in \mathcal{Z}$, are measured. Note that $\eta^{(i)}_{\mathrm{meas},k}$ denotes the common valve position of the radiators in zone~$i$ at $t_k$, as all radiators within the same zone have the same valve position. Since not all states in \eqref{eq:overall_ct_system} are measured in the case-study building, the unmeasured states at the beginning of the identification period must be estimated to initialize the model and subsequently identify the system parameters.

For the grey-box identification, we discretize the continuous-time model \eqref{eq:overall_ct_system} and consider uncertainty in the discrete-time model as
\begin{equation}
x_{k+1} = F(x_k, u_k, d_k; \theta) + w_k,
\label{eq:EKF_ct_system}
\end{equation}
where $x_k := x(t_k)$, $u_k := u(t_k)$, and $d_k := d(t_k)$ denote the corresponding values at $t_k \coloneqq k t_{\mathrm{s}}$. Moreover, $\theta \in \mathbb{R}^{n_{\theta}}$ denotes the vector of all system parameters that we aim to identify, and $w_k \in \mathbb{R}^n$ is considered zero-mean Gaussian process noise associated with the transition from $t_k$ to $t_{k+1}$, with $w_k \sim \mathcal{N}(0,Q)$ and covariance \mbox{$Q \in \mathbb{R}^{n\times n}$}. Note that, in this article, we discretize \eqref{eq:overall_ct_system} using the fourth-order Runge--Kutta method
(RK4) because compared with, for example, the Euler discretization, which is
first order, RK4 provides a more accurate
approximation of the continuous-time dynamics for a given sampling time \citep{RK4}. In particular, we discretize $f$ as
\begin{equation}
\begin{aligned}
x_{k+1} = F(x_k,u_k,d_k; \theta) + w_k
= x_k + \frac{t_s}{6}
\left(
\kappa_1 + 2\kappa_2 + 2\kappa_3 + \kappa_4
\right) + w_k,
\end{aligned}
\label{eq:rk4_discretization}
\end{equation}
where
\begin{alignat}{2}
\kappa_1
&\coloneqq f(x_k,u_k,d_k;\theta),
\qquad \qquad&
\kappa_2
&\coloneqq f\left(x_k+\frac{t_s}{2}\kappa_1,u_k,d_k;\theta\right),
\nonumber\\
\kappa_3
&\coloneqq f\left(x_k+\frac{t_s}{2}\kappa_2,u_k,d_k;\theta\right),
\qquad \qquad&
\kappa_4
&\coloneqq f\left(x_k+t_s\kappa_3,u_k,d_k;\theta\right).
\nonumber
\end{alignat}
Additionally, the measurement model is nonlinear and is described by 
\begin{align}
    y_k = h(x_k, T_{\mathrm{des},k}) + v_k, \label{eq:EKF_measurement}
\end{align}
where 
\begin{align}
y_k \coloneqq
\begin{bmatrix}
    \bigl[T_{\mathrm{meas},k}^{(i)}\bigr]_{i\in\mathcal{Z}}^{\top} &  \bigl[\eta_{\mathrm{meas},k}^{(i)}\bigr]_{i\in\mathcal{Z}}^{\top} 
\end{bmatrix}^{\top}
\in \mathbb{R}^{n_{\mathrm{y}}}, \qquad \mathcal{Z} \coloneqq \{2,3,4,5\}, \nonumber
\end{align}
denotes the measurements obtained from the sensors at time $t_k \coloneqq k t_s$. In \eqref{eq:EKF_measurement}, $T_{\mathrm{des},k} \coloneqq T_{\mathrm{des}}(t_k)$, and $v_k \in \mathbb{R}^{n_{\mathrm{y}}}$ denotes zero-mean Gaussian measurement noise with $v_k \sim \mathcal{N}(0,R)$, $R \in \mathbb{R}^{n_{\mathrm{y}}\times n_{\mathrm{y}}}$.

\begin{remark}
    The nonlinearity in the measurement model \eqref{eq:EKF_measurement} comes from the valve position outputs, and specifically from their saturation. If only the zone temperatures are extracted from the states, the measurement model \eqref{eq:EKF_measurement} becomes linear.
\end{remark}

We also collect two successive datasets. The first dataset is used for the warm-up phase to estimate all the states at the end of this phase and is defined as
\begin{align}
\mathcal{D}_{\mathrm{wp}}
\coloneqq
\left\{ \left(u_k,\, d_k,\, y_{k+1}\right) \right\}_{k=0}^{N_{\mathrm{wp}}-1}, \nonumber
\end{align}
where $N_{\mathrm{wp}} \in \mathbb{N}$ is the number of data points. The second dataset follows the warm-up dataset and is used to identify the system parameters $\theta$ in \eqref{eq:rk4_discretization} and is defined as
\begin{align}
    \mathcal{D}_{\mathrm{fit}}\coloneqq
\left\{ \left(u_k,\, d_k,\, y_{k+1}\right) \right\}_{k=N_{\mathrm{wp}}}^{N_{\mathrm{fit}}-1}, \nonumber
\end{align}
where $N_{\mathrm{fit}} \in \mathbb{N}$, with $N_{\mathrm{fit}} \geq N_{\mathrm{wp}}+1$, and $N_{\mathrm{fit}} - N_{\mathrm{wp}}$ is the number of data points for the identification. 

For the grey-box modeling, we adopt two approaches. The first is a two-block optimization method with theoretical guarantees. The second is a heuristic approach that is significantly faster, particularly for buildings with a large number of zones. We then compare the results obtained with the two approaches.

%% file: s3.tex
\subsection{Monotonic Two-Block Optimization Method} \label{sec:theoretical_approach}

Given the data collected for the warm-up phase $\mathcal{D}_{\mathrm{wp}}$, we aim to estimate the state $x_{N_{\mathrm{wp}}} \in \mathbb{R}^n$ at the end of this phase using maximum \textit{a posteriori} (MAP) estimation. In particular, the MAP estimate is defined as
\begin{align}
\hat{x}_{N_{\mathrm{wp}}}
\in
\underset{x_{N_{\mathrm{wp}}}}{\operatorname{arg \, max}}
\;
p(
x_{N_{\mathrm{wp}}}
 \mid 
\mathcal{D}_{\mathrm{wp}}
), \label{eq:MAP_probability}
\end{align}
where $p(x_{N_{\mathrm{wp}}}\mid \mathcal{D}_{\mathrm{wp}})$ denotes the posterior probability density of $x_{N_{\mathrm{wp}}}$ after incorporating the measurements in the warm-up dataset. Thus, $\hat{x}_{N_{\mathrm{wp}}} \in \mathbb{R}^n$ is the state value with the highest posterior probability density given $\mathcal{D}_{\mathrm{wp}}$. 
Following \cite{rawlings2017mpc}, for a given parameter vector $\theta$, the MAP estimation problem \eqref{eq:MAP_probability} can be formulated as
\begin{subequations}
\label{eq:MAP_est}
\begin{align}
(\hat{x}_0,\hat{W})
\in 
\operatorname*{arg\,min}_{\substack{
x_0\in\mathbb{R}^{n}\\
W\in\mathbb{R}^{n\times N_{\mathrm{wp}} }
}}
\quad &
\mathcal{J}_{\mathrm{wp}}(x_0,W; \theta, \mathcal{D}_{\mathrm{wp}})
\\
\text{subject to}\quad
&
x_{k+1}
=
F\!\left(x_k,u_k,d_k;\theta\right)+w_k,
\qquad k \in \{0,1,\ldots,N_{\mathrm{wp}}-1\},
\label{eq:MAP_est_const1} \\
&W \; \coloneqq \; \begin{bmatrix} w_0 & w_1 & \cdots & w_{N_{\mathrm{wp}}-1}\end{bmatrix}\in\mathbb{R}^{n\times N_{\mathrm{wp}}}, \\
&
x_0^{\mathrm{lb}}
\leq x_0
\leq x_0^{\mathrm{ub}},
\label{eq:MAP_est_const2} \\
&
-\bar{w}
\leq w_k
\leq \bar{w},
\qquad k \in \{0,1,\ldots,N_{\mathrm{wp}}-1\}, \label{eq:MAP_est_const3}
\end{align}
\end{subequations}
where
\begin{equation}
\mathcal{J}_{\mathrm{wp}}(x_0,W; \theta, \mathcal{D}_{\mathrm{wp}})
\coloneqq
\left\|x_0-\bar{x}_0\right\|_{P_0^{-1}}^2
+
\sum_{k=1}^{N_{\mathrm{wp}}}
\left\|
y_k-h\!\left(x_k,T_{\mathrm{des},k}\right)
\right\|_{R^{-1}}^2
+
\sum_{k=0}^{ N_{\mathrm{wp}}-1 }
\left\|w_k\right\|_{Q^{-1}}^2, \nonumber
\end{equation}
with $P_0\in\mathbb{R}^{n \times n}$ being the covariance of the Gaussian prior on
the initial state, i.e.\ the confidence placed in the initial guess $\bar{x}_0 \in \mathbb{R}^n$. In \eqref{eq:MAP_est}, $\bar{w} \in \mathbb{R}^n_{> 0}$ is the vector containing the maximum physically admissible magnitude of each component of the process noise, and $x_0^{\mathrm{lb}}, x_0^{\mathrm{ub}} \in \mathbb{R}^n$ are element-wise lower and upper bounds on the initial state, chosen to restrict $x_0$ to physically meaningful values.
Then, the state estimate $\hat{x}_{N_{\mathrm{wp}}} \in \mathbb{R}^n$ is obtained by propagating the system dynamics in \eqref{eq:MAP_est_const1} from the estimated initial state $\hat{x}_0$ using the estimated process-noise sequence $\hat{W}$.

\begin{remark}
    Note that to reduce the number of decision variables in \eqref{eq:MAP_est}, one can
    assume that the process noise $w_k$ enters only every $m$ steps, $m \in \mathbb{N}$, rather than at every
    step, and that it acts only on the zone temperatures.
\end{remark}

To identify the system parameter $\theta$, given $\mathcal{D}_{\mathrm{fit}}$, we aim to minimize the error between the measured data and the model outputs, with the model initialized at the estimated state $\hat{x}_{N_{\mathrm{wp}}}$. Mathematically, we solve
\begin{subequations}
\label{eq:MAP_iden}
\begin{align}
\hat{\theta}
\in 
\operatorname*{arg\,min}_{\substack{
\theta\in\mathbb{R}^{n_\theta}
}}
\quad &
\mathcal{J}_{\mathrm{fit}}\bigl(\theta \,;\, \hat{x}_{N_\mathrm{wp}}, \mathcal{D}_{\mathrm{fit}}\bigr)
\\
\text{subject to}\quad
&
x_{k+1}
=
F\!\left(x_k,u_k,d_k;\theta\right),
\qquad k \in \{N_{\mathrm{wp}}, N_{\mathrm{wp}} + 1, \ldots,N_{\mathrm{fit}}-1\},
\label{eq:MAP_iden_const1} \\
&  x_{N_{\mathrm{wp}}} = \hat{x}_{N_{\mathrm{wp}}},  \\
&
\theta_{\mathrm{min}}
\leq \theta
\leq \theta_{\mathrm{max}}, \label{eq:MAP_iden_const2}
\end{align}
\end{subequations}
where $\theta_{\min},\theta_{\max}\in\mathbb{R}^{n_\theta}$ denote the considered element-wise lower and upper bounds on the system parameters, and
\begin{align}
    \mathcal{J}_{\mathrm{fit}}\bigl(\theta \,;\, \hat{x}_{N_\mathrm{wp}}, \mathcal{D}_{\mathrm{fit}}\bigr) & \coloneqq \sum_{k=N_{\mathrm{wp}}}^{N_{\mathrm{fit}} - 1}
   \left\| y_{k+1}
-
h(x_{k+1}, T_{\mathrm{des}, k+1}) \right\| _2^{2}. \nonumber
\end{align}

Given that the MAP estimation \eqref{eq:MAP_est} depends on $\theta$ and the identification problem \eqref{eq:MAP_iden} depends on the estimated state~$\hat{x}_{N_{\mathrm{wp}}}$, both can be estimated simultaneously by combining \eqref{eq:MAP_est} and \eqref{eq:MAP_iden} into a single optimization problem as 
\begin{align} \label{eq:MAP_overall}
    (\hat{\theta}, \hat{x}_0, \hat{W}) \in \operatorname*{arg\,min}_{\substack{
x_0\in\mathbb{R}^{n}, \theta \in \mathbb{R}^{n_\theta}\\
W\in\mathbb{R}^{n\times N_{\mathrm{wp}} } \\
}}
\quad &
\overbrace{
\mathcal{J}_{\mathrm{wp}}(x_0,W; \theta, \mathcal{D}_{\mathrm{wp}})
+ \mathcal{J}_{\mathrm{fit}}\bigl(\theta; x_{N_{\mathrm{wp}}}(x_0, W, \theta), \mathcal{D}_{\mathrm{fit}}\bigr)
}^{\textstyle \eqqcolon\; \mathcal{J}(x_0,W,\theta)}
\\[0.5ex]
\text{subject to}\quad
& \eqref{eq:MAP_est_const1}\text{--}\eqref{eq:MAP_est_const3}, \, \eqref{eq:MAP_iden_const1}\text{--} \eqref{eq:MAP_iden_const2}. \nonumber
\end{align}
However, the optimization problem \eqref{eq:MAP_overall} may involve a large number of decision variables and can therefore be computationally demanding. To alleviate this computational burden, the decision variables can be divided into two blocks and updated iteratively using a so-called \textit{two-block optimization method} \citep{GRIPPO2000127}. In particular, starting from the parameters $\theta^{(0)}$ obtained from the physical properties, the initial state estimate $x^{(0)}$ and the noise matrix $W^{(0)}$ are obtained by solving~\eqref{eq:MAP_overall} with $\theta=\theta^{(0)}$ fixed. By propagating \eqref{eq:MAP_est_const1}, the state $x^{(0)}_{N_{\mathrm{wp}}}$ is then computed. Next, fixing $x^{(0)}_{N_{\mathrm{wp}}}$, the system parameters are updated to $\theta^{(1)}$ by solving \eqref{eq:MAP_overall} with respect to $\theta$. This procedure is repeated iteratively; see Algorithm~\ref{alg:two_block}. Since the same objective function is minimized in both blocks, the cost generated by the two-block optimization method is nonincreasing over the iterations; see Theorem~\ref{thm:monotone_decrease} below. Consequently, each parameter update does not worsen the overall fit to the measurements, according to the considered objective function $\mathcal{J}$, compared with the parameters from the previous iteration.

\begin{theorem} \label{thm:monotone_decrease} 
The sequence $\bigl\{(x_0^{(j)},W^{(j)},\theta^{(j)})\bigr\}_{j\geq 0}$ generated by Algorithm~\ref{alg:two_block}, where $j \in \mathbb{N}_0$ denotes the iteration number, satisfies \begin{equation} 
\begin{aligned} 
\mathcal{J}\bigl( x_0^{(j+1)},W^{(j+1)},\theta^{(j+1)} \bigr) &\leq \mathcal{J}\bigl( x_0^{(j+1)},W^{(j+1)},\theta^{(j)} \bigr) \leq \mathcal{J}\bigl( x_0^{(j)},W^{(j)},\theta^{(j)} \bigr), \qquad \text{~~for all~~~} j\geq 0. 
\end{aligned} \label{eq:MAP_theorem} 
\end{equation} 
\end{theorem} 
\begin{proof}
At iteration $j \in \mathbb{N}_{0}$ of Algorithm~\ref{alg:two_block}, Block~1 fixes $\theta^{(j)}$ and minimizes $\mathcal{J}(x_0,W,\theta^{(j)})$ with respect to $x_0$ and $W$. Since the previous iterate $(x_0^{(j)},W^{(j)})$ is feasible for this subproblem, the optimality of $(x_0^{(j+1)},W^{(j+1)})$ implies that \begin{equation} \mathcal{J}\bigl( x_0^{(j+1)},W^{(j+1)},\theta^{(j)} \bigr) \leq \mathcal{J}\bigl( x_0^{(j)},W^{(j)},\theta^{(j)} \bigr). \label{eq:MAP_theorem_proof_1} \end{equation} Similarly, in Block~2, $x_0^{(j+1)}$ and $W^{(j+1)}$ are fixed, and $\mathcal{J}(x_0^{(j+1)},W^{(j+1)},\theta)$ is minimized with respect to $\theta$. Since $\theta^{(j)}$ is feasible for this subproblem, the optimality of $\theta^{(j+1)}$ yields \begin{equation} \mathcal{J}\bigl( x_0^{(j+1)},W^{(j+1)},\theta^{(j+1)} \bigr) \leq \mathcal{J}\bigl( x_0^{(j+1)},W^{(j+1)},\theta^{(j)} \bigr). \label{eq:MAP_theorem_proof_2} \end{equation} Combining \eqref{eq:MAP_theorem_proof_1} and \eqref{eq:MAP_theorem_proof_2} gives \eqref{eq:MAP_theorem}. It is worth noting that global optimality of each block is not required here because each subproblem is warm-started from the solution obtained in the previous iteration assuming that local solvers do not return a solution whose objective value is greater than that of the warm start.
\end{proof}

\begin{algorithm}[t]
\caption{Two-block optimization for joint state estimation and parameter identification}
\label{alg:two_block}
\begin{algorithmic}[1]

\Require Warm-up dataset $\mathcal{D}_{\mathrm{wp}}$, identification dataset
$\mathcal{D}_{\mathrm{fit}}$; initial parameters $\theta^{(0)}$ satisfying
\eqref{eq:MAP_iden_const2}; prior $(\bar{x}_0,P_0)$; tolerance $\epsilon\in\mathbb{R}_{>0}$; maximum number of iterations
$j_{\max}\in\mathbb{N}$ %weights $Q$, $R$;

\Ensure Estimates $\hat{\theta}$, $\hat{x}_0$, $\hat{W}$

\Statex \hspace{\algorithmicindent}

\State Initialize $x_0^{(0)}\gets\bar{x}_0$, \;
$W^{(0)}\gets 0$

\For{$j=0,1,\dots,j_{\max}-1$}

\State \textbf{Block 1 (state estimation).} Fix $\theta=\theta^{(j)}$ and
warm-start with $\bigl(x_0^{(j)},W^{(j)}\bigr)$:

\Statex \hspace{\algorithmicindent}
$\displaystyle
\bigl(x_0^{(j+1)},W^{(j+1)}\bigr)
\in \operatorname*{arg\,min}_{x_0,W}\;
\mathcal{J}\bigl(x_0,W,\theta^{(j)}\bigr)$

\Statex \hspace{\algorithmicindent}\hspace{\algorithmicindent}
$\text{s.t.~~~}\eqref{eq:MAP_est_const1}\text{--}\eqref{eq:MAP_est_const3},
\eqref{eq:MAP_iden_const1}$

\State \textbf{Block 2 (parameter identification).} Fix
$x_0=x_0^{(j+1)}$ and $W=W^{(j+1)}$, and warm-start with $\theta^{(j)}$:

\Statex \hspace{\algorithmicindent}
$\displaystyle
\theta^{(j+1)}
\in\operatorname*{arg\,min}_{\theta}\;
\mathcal{J}\bigl(x_0^{(j+1)},W^{(j+1)},\theta\bigr)$

\Statex \hspace{\algorithmicindent}\hspace{\algorithmicindent}
$\text{s.t.~~~}\eqref{eq:MAP_est_const1},
\eqref{eq:MAP_iden_const1}\text{--}\eqref{eq:MAP_iden_const2}$

\If{$\lVert\theta^{(j+1)}-\theta^{(j)}\rVert\leq\epsilon$}
    \State \textbf{break}
\EndIf

\EndFor

\State $\hat{\theta}\gets\theta^{(j+1)}$, \;
$\hat{x}_0\gets x_0^{(j+1)}$, \;
$\hat{W}\gets W^{(j+1)}$

\State \Return $\hat{\theta}$, $\hat{x}_0$, $\hat{W}$

\end{algorithmic}
\end{algorithm}

The optimization problem \eqref{eq:MAP_overall} is applied to nine days of measurements collected from the Rensen building, consisting of a four-day warm-up phase and a five-day identification phase. The estimated state trajectory during the warm-up phase, obtained by propagating the system dynamics using the estimated initial state $\hat{x}_0$ and process-noise matrix $\hat{W}$, is depicted in Figure~\ref{fig:MAP_grey_box_wp}. The open-loop simulation over the identification phase, performed using the identified parameter vector $\hat{\theta}$ and initialized at $\hat{x}_{N_{\mathrm{wp}}}$, is depicted in Figure~\ref{fig:MAP_grey_box_iden}. Moreover, the per-zone root mean square error (RMSE) and maximum error between the model outputs and the corresponding measurements, together with the overall computational time, are reported in Table~\ref{tab:two_block_open_loop_rmse}. The identification is performed on a laptop equipped with a 12th Gen Intel Core i7-12700 processor and 32 GB of RAM.

\begin{figure}
     \centering
     \begin{subfigure}[t]{0.44\linewidth}
         \centering
         \includegraphics[height=0.25\textheight, keepaspectratio]{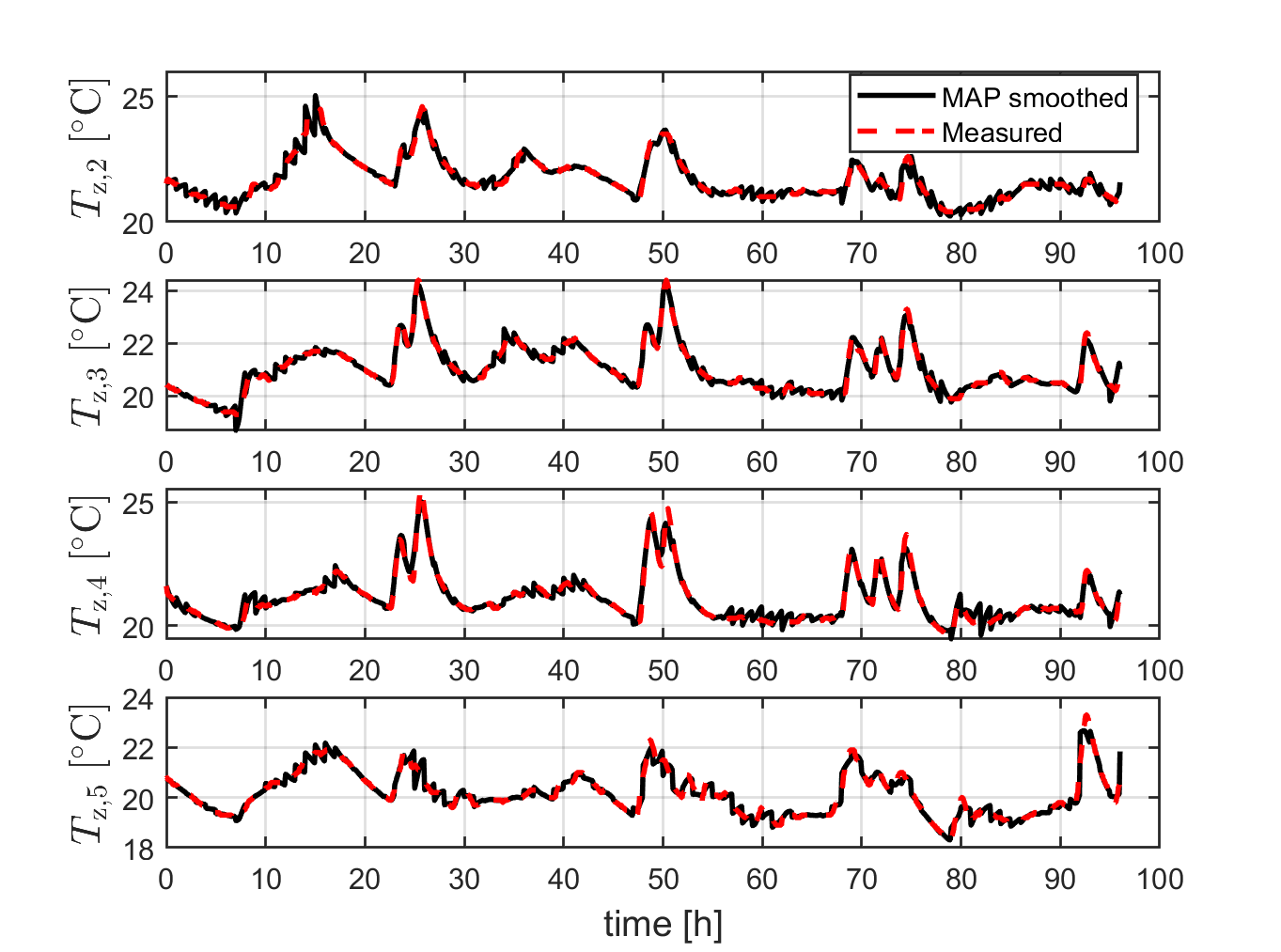}
         \caption{MAP estimation is applied to four days of measurements to estimate the state $\hat{x}_{N_{\mathrm{wp}}}$ at the end of this phase, which are then used in the identification step.}
         \label{fig:MAP_grey_box_wp}
     \end{subfigure}
     \hfill
    \begin{subfigure}[t]{0.48\linewidth}
        \centering
        \includegraphics[height=0.25\textheight, keepaspectratio]{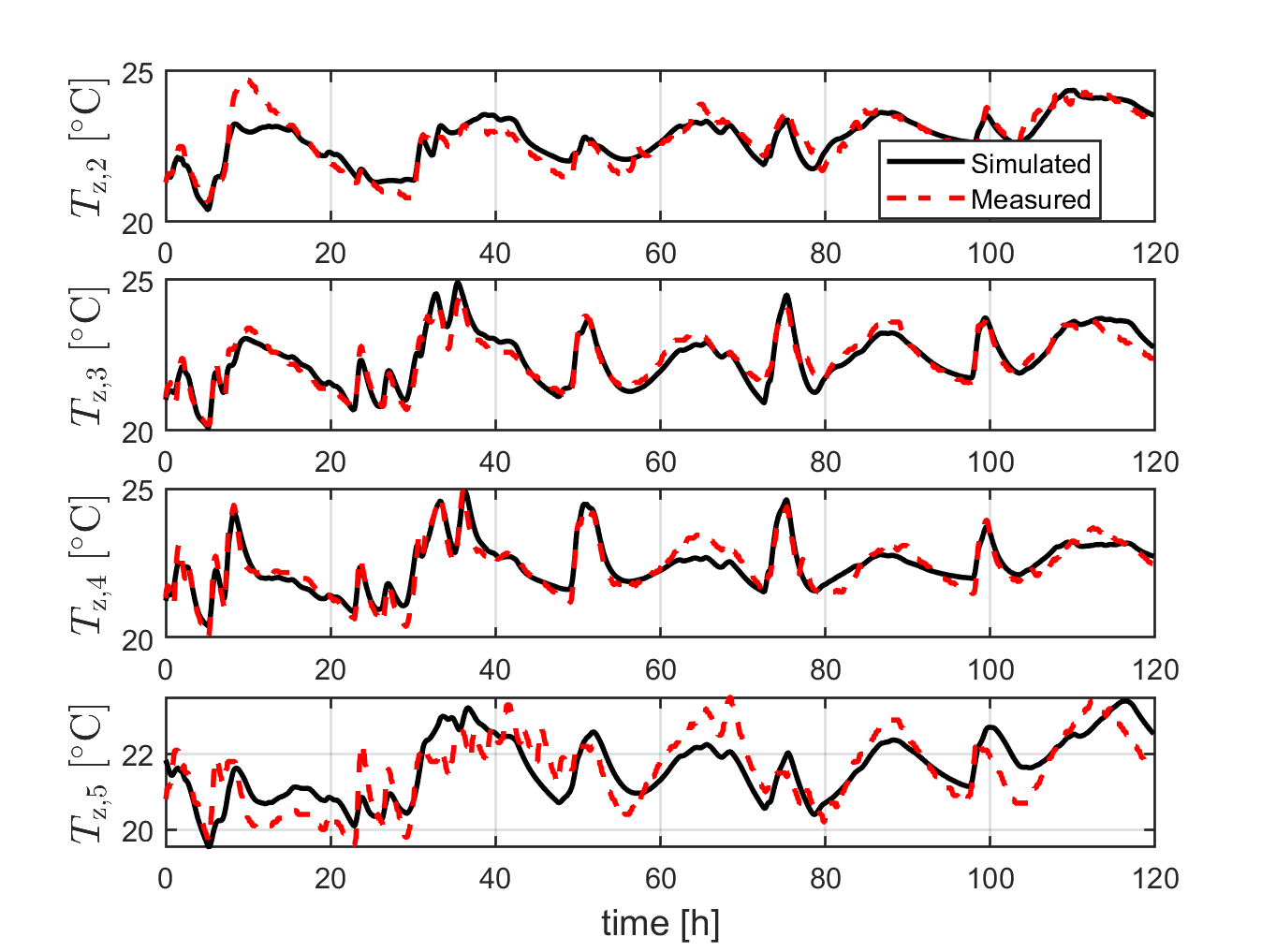}
        \caption{Five-day open-loop simulation initialized from $\hat{x}_{N_{\mathrm{wp}}}$ using the optimized $\theta$, together with the corresponding real measurements.}
        \label{fig:MAP_grey_box_iden}
    \end{subfigure}
        \caption{Monotonic two-block optimization method discussed in Section~\ref{sec:theoretical_approach} for identifying the system parameters of the case-study building.}
    \label{fig:MAP_grey_box}
\end{figure}

\begin{table}[t]
    \centering
    \caption{Per-zone error between modeled and measured zone temperatures using the system parameters estimated with \textit{the two-block optimization method}.}
    \label{tab:two_block_open_loop_rmse}
    \begin{tabular}{ccc}
        \toprule
        Zone & RMSE & maximum error \\
        & [\si{\degreeCelsius}]
        & [\si{\degreeCelsius}] \\
        \midrule
        2 & $0.392$ & $1.723$ \\
        3 & $0.298$ & $0.913$ \\
        4 & $0.317$ & $0.863$ \\
        5 & $0.621$ & $1.799$ \\
        \midrule
        \multicolumn{2}{r}{Total computation time} &  \SI{3}{\hour} \\
        \bottomrule
    \end{tabular}
\end{table}

However, even when using the two-block optimization method, solving the optimization problem~\eqref{eq:MAP_overall} remains computationally demanding, particularly for buildings with many zones. For the six-zone case-study building, the optimization required approximately $3~\si{\hour}$ to converge to a satisfactory solution. In the next subsection, we therefore propose an alternative method that, although it comes without theoretical guarantees, yields results very close to those of the two-block optimization method at a substantially lower computational time.

\subsection{Heuristic Iterative EKF–Identification Method} \label{sec:practical_approach}
In this method, the idea is similar to that of the two-block optimization method, but instead of solving the optimization problem \eqref{eq:MAP_est} with a large number of decision variables for MAP estimation, we use an observer, namely Extended Kalman Filter (EKF). In particular, for nonlinear systems, the state estimate obtained from the EKF is an approximation to the solution of the MAP estimation problem \eqref{eq:MAP_probability}. However, since the observer itself requires the model and its parameters, we propose a heuristic iterative scheme. Starting from an initial vector of parameter guesses, we run the EKF to obtain a state estimate. We then solve the identification problem, which minimizes the weighted error between the model outputs initialized at the estimated state and the real measurements. Using the updated parameters, the EKF is subsequently applied to obtain a new state estimate, and this procedure is repeated iteratively. 

\subsubsection{Warm-Up Phase} \label{sec:sec3:warm-up}
For the discretized system \eqref{eq:EKF_ct_system} with process noise $w_k \sim \mathcal{N}(0,Q)$ and the discrete-time measurement model \eqref{eq:EKF_measurement} with measurement noise $v_k \sim \mathcal{N}(0,R)$, given a parameter vector $\theta$ and the warm-up dataset $\mathcal{D}_{\mathrm{wp}}$, the EKF proceeds as follows (\cite{EKF}). Initially, we set
\begin{align}
\hat{x}_{0|0} \coloneqq \mathbb{E}[x_0], ~~~
P_{0|0} \coloneqq
\mathbb{E}[(x_0-\hat{x}_{0|0})(x_0-\hat{x}_{0|0})^\top], \nonumber
\end{align}
where $\mathbb{E}[\cdot]$ denotes the expected value.
For $k \in \{0,1,\ldots, N_{\mathrm{wp}} - 1\}$, the following prediction and update steps are performed iteratively:
\begin{enumerate}
    \item Prediction: Over the interval $t\in[t_k,t_{k+1}]$, with $t_k \coloneqq k t_{\mathrm{s}}$, the state estimate at $t_{k+1}$ is predicted as
   \begin{equation}
    \hat{x}_{k+1 | k} =
    F \left(\hat{x}_{k | k} , u_k, d_k; \theta \right). \nonumber
    \end{equation}
    Moreover, the corresponding error covariance matrix is predicted as
    \begin{equation}
        P_{k+1 | k} = A_k P_{k | k} A_k^{\top} + Q,\text{~~~~~~~where~~~~~~~}
    A_k \coloneqq
    \left.
    \frac{\partial F}{\partial x}
    \right|_{x=\hat{x}_{k |k},u=u_k,d=d_k}. \nonumber
    \end{equation}
    \item Update: Using the measurement $y_{k+1}$ at time $t_{k+1}$, the predicted state estimate and covariance matrix are corrected as
    \begin{align}
    \hat{x}_{k+1|k+1}
    &=
    \hat{x}_{k+1|k}
    +
    K_{k+1}
    \big(
    y_{k+1} - h(\hat{x}_{k+1|k}, T_{\mathrm{des},k+1})
    \big), \nonumber \\
    P_{k+1|k+1}
    &=
    (I_n-K_{k+1}H_{k+1})P_{k+1|k}, \nonumber
    \end{align}
    where $I_n \in \mathbb{R}^{n \times n}$ denotes the identity matrix, and
    \begin{align}
    K_{k+1}
    &\coloneqq
    P_{k+1|k}H_{k+1}^{\top}
    \left(
    H_{k+1}P_{k+1|k}H_{k+1}^{\top}+R
    \right)^{-1}, \nonumber \\[2mm]
    &\text{with~~~~~~~}
    H_{k+1}
    \coloneqq
    \left.
    \frac{\partial h}{\partial x}
    \right|_{x=\hat{x}_{k+1|k},T_{\mathrm{des}}=T_{\mathrm{des},k+1}}. \nonumber
    \end{align}
\end{enumerate}

Figure~\ref{fig:grey_box_wu} illustrates the four-day warm-up phase and presents the EKF-filtered state estimates $\hat{x}_{k | k}$, with \mbox{$k \in \{0, 1,\ldots, N_{\mathrm{wp}} \}$}, together with the measured zone temperatures for zones 2, 3, 4, and 5. At the end of this period, an estimate of the state vector is obtained, denoted by \mbox{$\hat{x}_{\mathrm{EKF}} \coloneqq \hat{x}_{N_{\mathrm{wp}} | N_{\mathrm{wp}}} \in \mathbb{R}^n$}.

\subsubsection{Identification Phase}
In the identification phase, we categorize the system parameters into two groups as
\begin{equation}
    \theta
    \coloneqq
    \begin{bmatrix}
        \theta_{\mathrm{bldg}}^\top &
        \theta_{\mathrm{rad}}^\top
    \end{bmatrix}^\top \in \mathbb{R}^{n_{\theta}}, \nonumber
\end{equation}
where $\theta_{\mathrm{bldg}} \in \mathbb{R}^{n_{\theta_{\mathrm{bldg}}}}$ and $\theta_{\mathrm{rad}} \in \mathbb{R}^{n_{\theta_{\mathrm{rad}}}}$, with $n_{\theta_{\mathrm{bldg}}} + n_{\theta_{\mathrm{rad}}} = n_{\theta}$, denote the
building and radiator parameters, respectively. For each group, we minimize the weighted residuals between the model outputs initialized at $\hat{x}_{\mathrm{EKF}}$ and the corresponding measurements with respect to the target parameters. In particular, for each parameter group
$g \in \{\mathrm{bldg},\mathrm{rad}\}$, we solve
\begin{equation}
\begin{aligned}
\min_{\theta} \quad
& 
J\bigl(\theta \,;\, \hat{x}_{\mathrm{EKF}}, \mathcal{D}_{\mathrm{fit}}\bigr)
\\
\text{subject to}\quad &
x_{k+1}
=
F(x_k, u_k,d_k; \theta),
\qquad k \in \{N_{\mathrm{wp}},N_{\mathrm{wp}}+1,\ldots,N_{\mathrm{fit}}-1\},
\\
&
x_{N_{\mathrm{wp}}} = \hat{x}_{\mathrm{EKF}},
\\
&
\theta_{g, \min} \leq \theta_g \leq \theta_{g, \max},
\end{aligned}
\label{eq:gray_box_weighted_optimization}
\end{equation}
where
\begin{align}
    J\bigl(\theta \,;\, \hat{x}_{\mathrm{EKF}}, \mathcal{D}_{\mathrm{fit}}\bigr) \coloneqq \frac{1}{\sum_{k=N_{\mathrm{wp}}}^{N_{\mathrm{fit}}-1}\sum_{i\in\mathcal{Z}} w^g_{i,k}}\sum_{k=N_{\mathrm{wp}}}^{N_{\mathrm{fit}}-1}
    \sum_{i\in\mathcal{Z}}
    w^g_{i,k}
    \left(
    T_{\mathrm{z},k+1}^{(i)}
    -
    T_{\mathrm{meas},k+1}^{(i)}
    \right)^2. \nonumber
\end{align}
Here, the weight $w^g_{i,k} \in \mathbb{R}_{> 0}$ reflects
the thermal regime of each sample. In particular, when a radiator valve is open, the zone-temperature dynamics are more strongly influenced by the radiator heat output. Therefore, these samples are informative for identifying the radiator parameters, but they are relatively less informative for identifying the building parameters. Conversely, when the valves are
closed, the zone temperature dynamics are governed by the
building thermal dynamics so these samples are most
informative for the building parameters. Each identification therefore
weights the residuals toward the regime in which its parameters are
excited. Therefore, for identifying the building parameters, we define the weighting constant $w^{\mathrm{bldg}}_{i,k} \in \mathbb{R}_{> 0}$ at each time instant $k$ and for each zone $i$ as
\begin{equation}
    w^{\mathrm{bldg}}_{i,k} \coloneqq w_{\mathrm{floor}} + (1 - w_{\mathrm{floor}})\,
    \bigl(1 - h_{i,k}\bigr) \;\in\; [w_{\mathrm{floor}},\, 1] \subset \mathbb{R}, \nonumber
\end{equation}
where $0 \leq w_{\mathrm{floor}} \ll 1$ denotes the minimum weight assigned to samples corresponding to the less informative thermal regime, and $h_{i,k}$ denotes the degree to which temperature of zone $i$ at time $k$ is dominated by the heating system. Specifically,
\begin{equation}
h_{i,k} \coloneqq e^{-\ell_{i,k} (t_s/\tau_{w})} \;\in\; [0,\,1] \subset \mathbb{R}, \nonumber
\end{equation}
where $\ell_{i,k} \in \mathbb{N}_0$ denotes the number
of samples elapsed since the valve of zone $i$ was last open
($\ell_{i,k}=0$ while the valve is open, and $\ell_{i,k}=\infty$ if no
opening has occurred yet), and $\tau_{w} ~[\si{ \second }]$ is a tunable time constant reflecting how quickly the heating influence becomes negligible after valve closure. Hence, $h_{i,k}=1$ when the valve of zone $i$ is open, and it decays exponentially to zero after the valve closes. Note that the exponential decay of $h_{i,k}$ reflects that the radiator continues to release stored heat for some time even after the valve closes. 
Similarly, for identifying the radiator
parameters, we define $w^{\mathrm{rad}}_{i,k} \in \mathbb{R}_{> 0}$ as
\begin{equation}
    w^{\mathrm{rad}}_{i,k} \coloneqq w_{\mathrm{floor}} + (1 - w_{\mathrm{floor}})\,
    h_{i,k}. \nonumber
\end{equation}

\begin{remark}
    By introducing this weighting scheme and setting $w_{\mathrm{floor}} \approx 0$, we reduce the number of (irrelevant) data samples used for identification. This significantly improves the computational efficiency of the identification problem.
\end{remark}

The overall algorithm is summarized in Algorithm \ref{alg:ekf-nlp}. For the case-study building, Algorithm~\ref{alg:ekf-nlp} is applied using real zone-temperature measurements over nine days: a four-day warm-up phase followed by a five-day identification phase. Figure~\ref{fig:grey_box_cost} shows the weighted RMSE, defined as the square root of the cost in \eqref{eq:gray_box_weighted_optimization}, associated with the building-parameter identification weights $w^{\mathrm{bldg}}_{i,k}$, together with the relative change in $\theta$ across iterations. Additionally, Figure~\ref{fig:grey_box_ol} and Figure~\ref{fig:gre_box_phase_radiator} show the open-loop modeled zone temperatures compared with the corresponding measurements, together with the building and radiator weights used for identification, respectively. The open-loop trajectories are initialized at $\hat{x}_{\mathrm{EKF}}$ and use the estimated parameters $\hat{\theta}$ obtained after Algorithm~\ref{alg:ekf-nlp} has converged. Table~\ref{tab:per_zone_open_loop_fit} reports the RMSE and weighted RMSE between the modeled and measured zone temperatures for each zone, together with the total computation time. By comparing Tables~\ref{tab:per_zone_open_loop_fit} and~\ref{tab:two_block_open_loop_rmse}, it can be observed that the parameters obtained with the proposed heuristic method yield almost the same error as those obtained with the two-block optimization method, while the computation time is substantially lower.

\begin{algorithm}
\caption{Heuristic EKF–NLP iterative method}
\label{alg:ekf-nlp}
\begin{algorithmic}[1] 
\Require Warm-up dataset $\mathcal{D}_{\mathrm{wp}}$,
         identification dataset $\mathcal{D}_{\mathrm{fit}}$; initial parameters $\theta^{0}$; prior $(x_{0}, P_{0})$; tolerance $\varepsilon \in \mathbb{R}_{> 0}$; maximum number of iterations $j_{\mathrm{max}} \in \mathbb{N}$
\Ensure Parameter estimate $\hat{\theta}$
\Statex \hspace{\algorithmicindent}
\State $\hat{x}_{\mathrm{EKF}} \gets \textsc{WarmUp-Phase}(\theta^{0}, x_{0}, P_{0}, \mathcal{D}_{\mathrm{wp}})$
        \Comment{EKF} %filtered handoff state
\State $J_0 \gets J(\theta^{0} \,;\, \hat{x}_{\mathrm{EKF}}, \mathcal{D}_{\mathrm{fit}})$
\State $\hat{\theta} \gets \theta^{0}$
\For{$j = 1,2, \dots, j_{\mathrm{max}}$}
    \State $\theta^{+} \gets \textsc{Identification-Phase}(\hat{x}_{\mathrm{EKF}}, \mathcal{D}_{\mathrm{fit}}, \hat{\theta})$
        \Comment{solve \eqref{eq:gray_box_weighted_optimization} separately for the building and radiator parameters with $\hat{x}_{\mathrm{EKF}}$ fixed}
    \State $\hat{x}_{\mathrm{EKF}}^+ \gets \textsc{WarmUp-Phase}(\theta^+, x_{0}, P_{0}, \mathcal{D}_{\mathrm{wp}})$
    \State $J^{+} \gets J(\theta^{+} \,;\, \hat{x}_{\mathrm{EKF}}^+, \mathcal{D}_{\mathrm{fit}})$
    \Statex \hspace{\algorithmicindent} \textit{--- Acceptance test ---}
    \If{$J^{+} \le J_{0}$} \label{alg:ekf-nlp:line8} 
        \State $J_{0} \gets J^{+}$
        \If{$\lVert \theta^{+} - \hat{\theta} \rVert / \lVert \hat{\theta} \rVert < \varepsilon$}
            \State $\hat{\theta} \gets \theta^{+}$
            \State \textbf{break}
        \EndIf
        \State $\hat{\theta} \gets \theta^{+}$
        \State $\hat{x}_{\mathrm{EKF}} \gets \hat{x}_{\mathrm{EKF}}^+$
    \EndIf
\EndFor
\State \Return $\hat{\theta}$
\end{algorithmic}
\end{algorithm}

\begin{remark}
    Note that, unlike the two-block optimization method discussed in Section~\ref{sec:theoretical_approach}, this approach does not guarantee a nonincrease in the cost over the iterations. This is because the objective function $J$ in \eqref{eq:gray_box_weighted_optimization} changes at each iteration as the EKF state estimate $\hat{x}_{\mathrm{EKF}}$ is updated. Therefore, in Line~\ref{alg:ekf-nlp:line8} of Algorithm~\ref{alg:ekf-nlp}, we discard the iterations for which the cost increases.
\end{remark}

\begin{figure}
    \centering
    \includegraphics[width=0.55\linewidth]{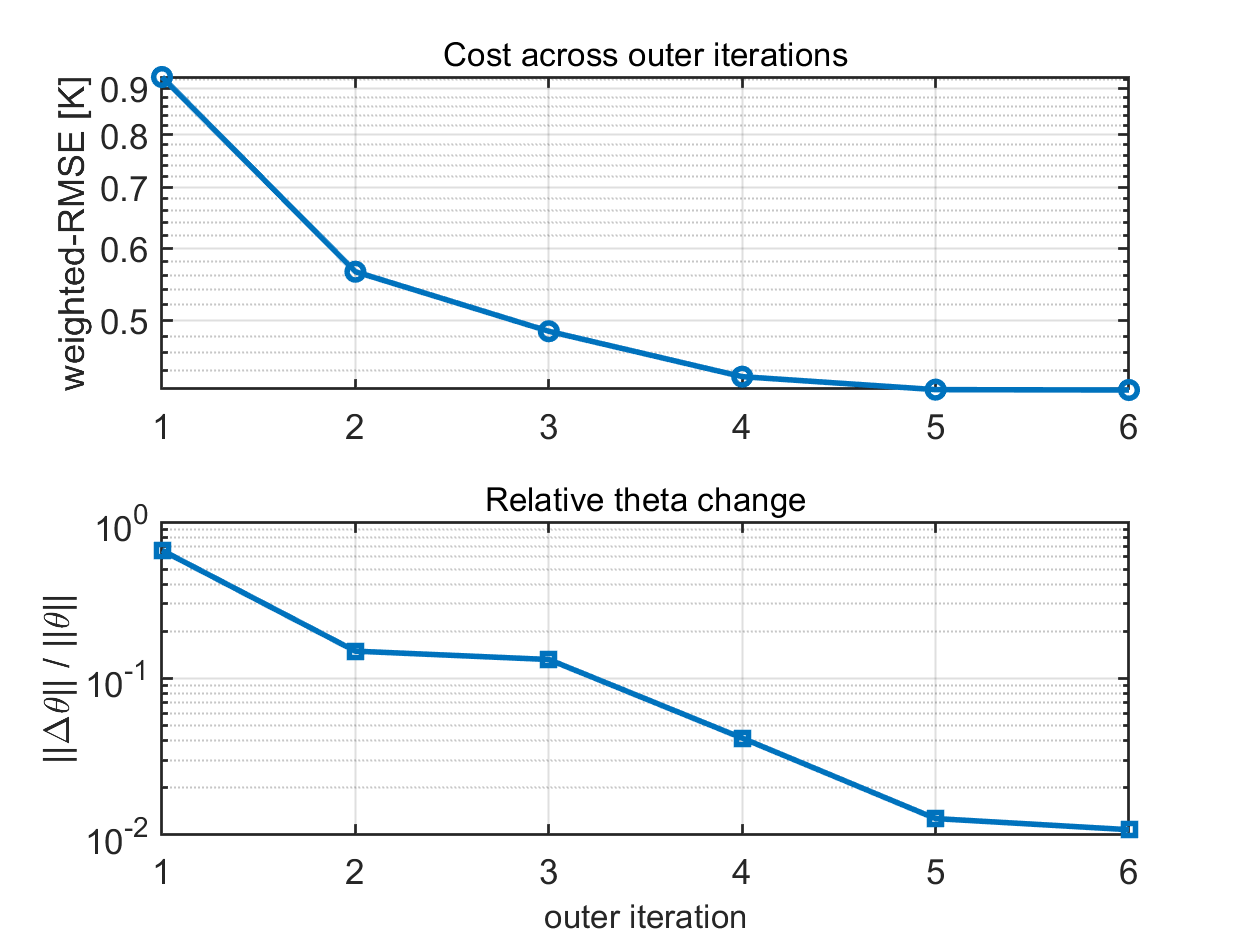}
    \caption{Weighted root-mean-square error (RMSE), computed as the square root of the cost in \eqref{eq:gray_box_weighted_optimization} using the building-parameter identification weights, and relative change in the targeted parameters $\theta$ over the outer-loop iterations of the heuristic EKF--identification method.}
    \label{fig:grey_box_cost}
\end{figure}

\begin{figure}
     \centering
     \begin{subfigure}[t]{0.44\linewidth}
         \centering
         \includegraphics[height=0.25\textheight, keepaspectratio]{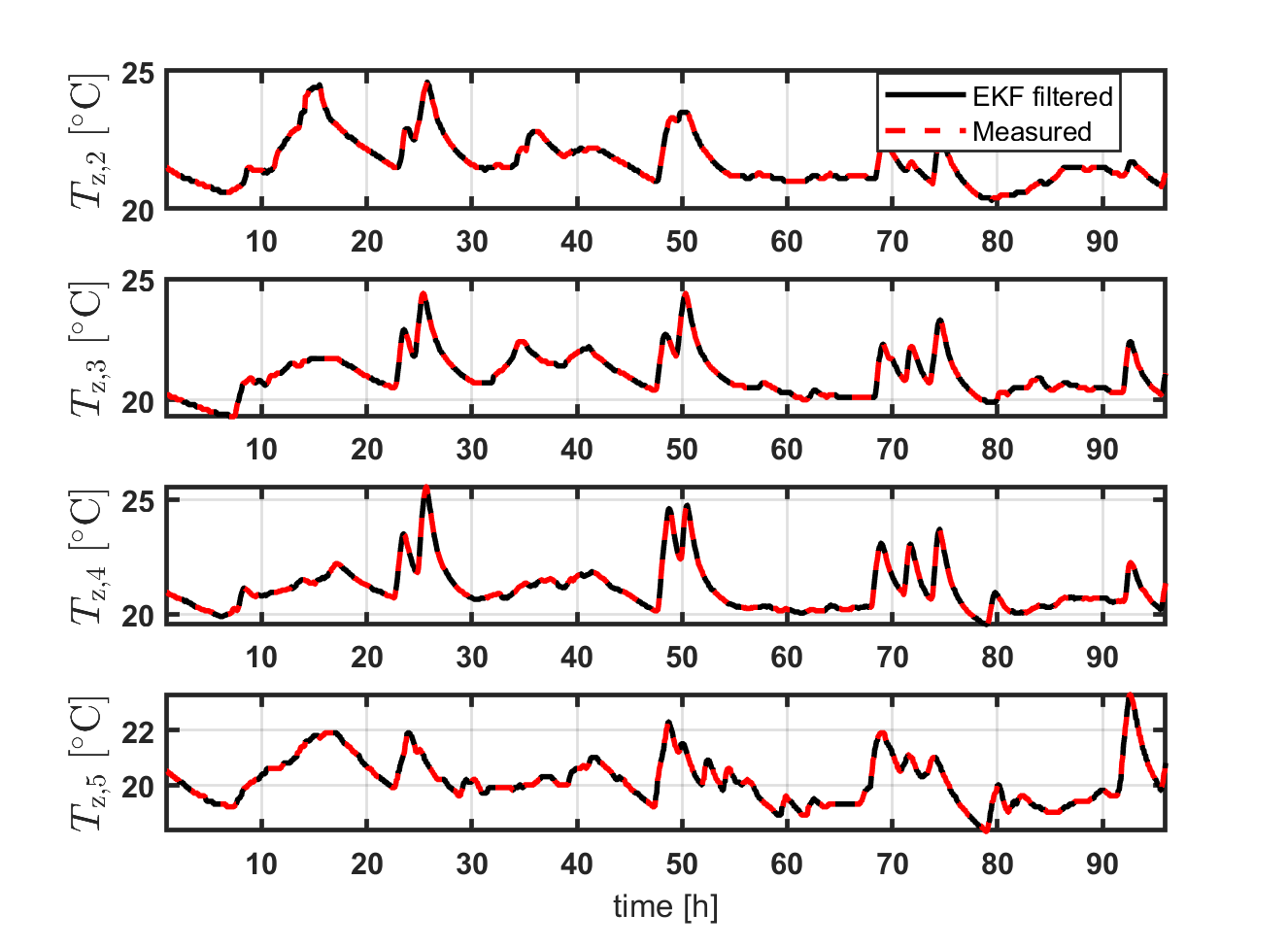}
         \caption{The warm-up phase (four days) uses an extended Kalman filter (EKF) to estimate the states at the end of this phase, which are then used to initialize the identification phase.}
         \label{fig:grey_box_wu}
     \end{subfigure}
     \hfill
    \begin{subfigure}[t]{0.48\linewidth}
        \centering
        \includegraphics[height=0.25\textheight, keepaspectratio]{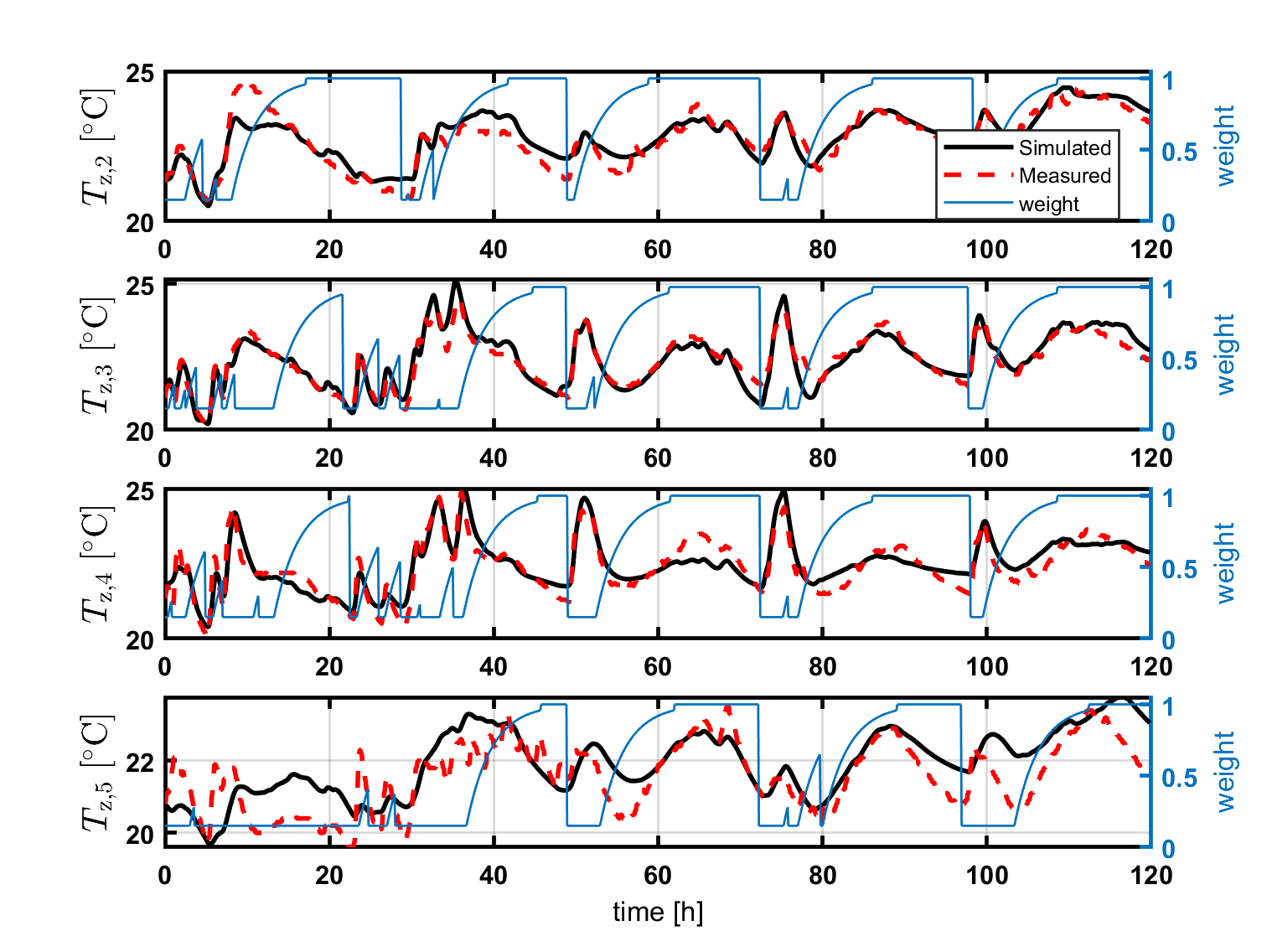}
        \caption{Five-day open-loop simulation initialized from the final warm-up state using the optimized $\theta$, together with the corresponding real measurements and weights used in the identification phase for the building-related system parameters.}
        \label{fig:grey_box_ol}
    \end{subfigure}
        \caption{The proposed heuristic iterative EKF–identification method discussed in Section~\ref{sec:practical_approach} for identifying the building-related system parameters of the case-study building.}
    \label{fig:gre_box_phase}
\end{figure}

\begin{figure}
    \centering
    \includegraphics[width=0.6\linewidth]{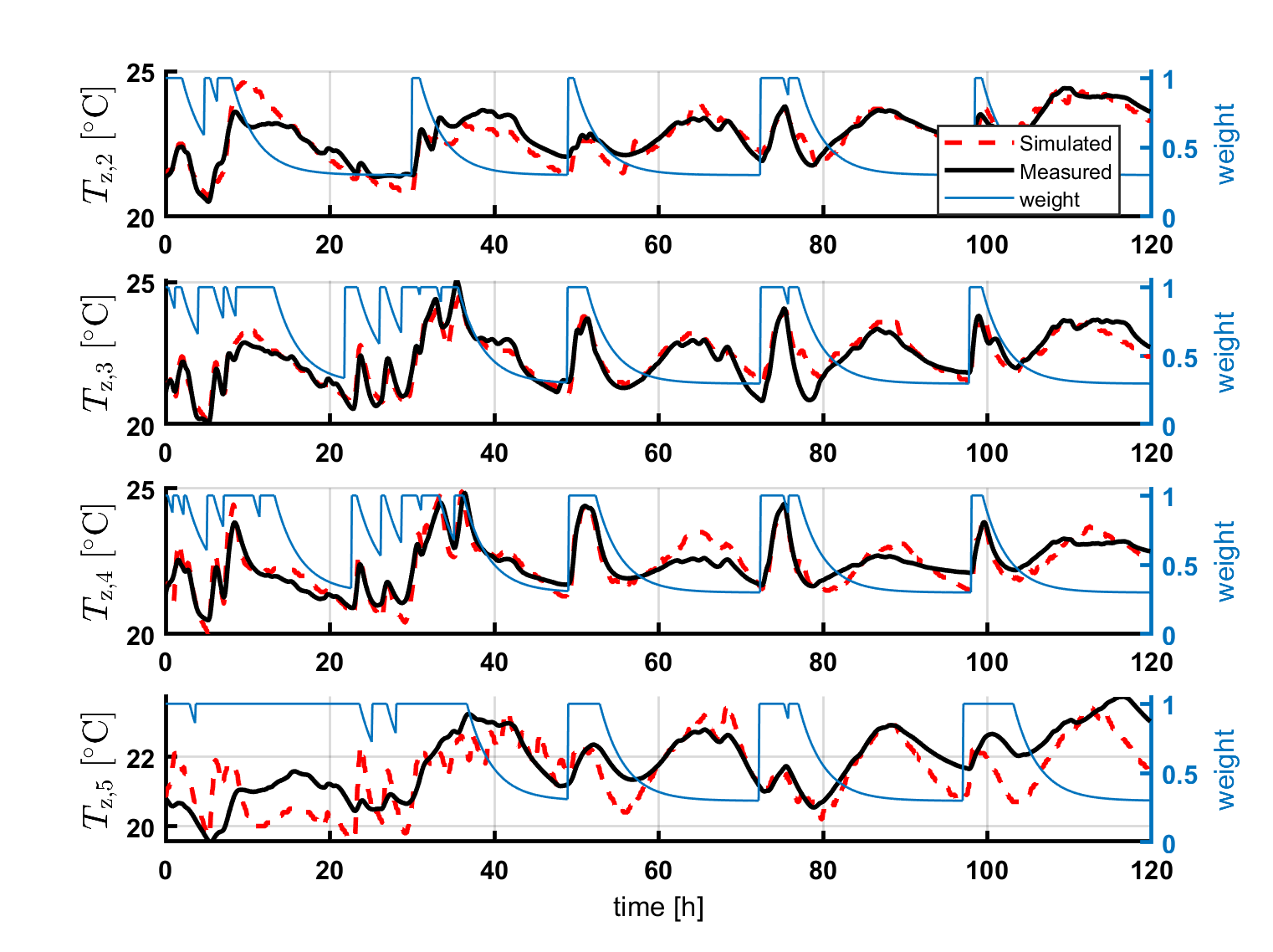}
    \caption{Five-day open-loop simulation, corresponding measurements, and identification weights used to estimate the radiator-related system parameters with the proposed heuristic iterative EKF–identification method.}
    \label{fig:gre_box_phase_radiator}
\end{figure}

\begin{table}[t]
    \centering
    \caption{Per-zone errors between the modeled and measured zone temperatures using the parameters $\theta$ estimated with the \textit{heuristic iterative EKF–identification method}.}
    \label{tab:per_zone_open_loop_fit}
    \begin{tabular}{ccccc}
        \toprule
        Zone & RMSE & Weighted RMSE (building) & Weighted RMSE (radiator) & maximum error \\
        & [\si{\degreeCelsius}] & [\si{\degreeCelsius}] & [\si{\degreeCelsius}]
        & [\si{\degreeCelsius}] \\
        \midrule
        2 & $0.385$ & $0.363$ & 0.405 & $1.600$ \\
        3 & $0.332$ & $0.298$ & 0.361 & $1.525$ \\
        4 & $0.359$ & $0.367$ & 0.352 & $1.076$ \\
        5 & $0.725$ & $0.634$ & 0.784 & $2.212$ \\
        \midrule
        \multicolumn{4}{r}{Total computation time} & < \SI{1}{\hour} \\
        \bottomrule
    \end{tabular}
\end{table}

\subsection{Effect of Hydraulic Interactions on Grey-Box Modeling Accuracy} \label{sec:effec_hydraulic}
To motivate the inclusion of hydraulic interactions among the radiators discussed in Section~\ref{sec:hydraulic_model}, we estimate the parameters of model~\eqref{eq:overall_ct_system} using the two-block optimization method and evaluate its errors on a dataset different from that used for identification. The evaluation consists of a five-day warm-up phase, followed by a seven-day open-loop simulation. The resulting errors are then compared with those obtained using a model in which hydraulic interactions are neglected. In particular, in the latter model, a nominal water flow rate $\dot{m}^{(i)}_{\mathrm{wat}, \mathrm{nom}} \in \mathbb{R}$ is assigned to each radiator~$i$, and its actual flow rate is assumed to satisfy
\begin{align}
\dot{m}^{(i)}_{\mathrm{wat}}
=
\dot{m}^{(i)}_{\mathrm{wat}, \mathrm{nom}}\eta^{(i)}, \quad i \in \{1,2,\ldots, n_{\mathrm{rad}}\},
\nonumber
\end{align}
where $\eta^{(i)} \in [0,1] \subset \mathbb{R}$ denotes the corresponding valve position. The nominal flow rate $\dot{m}^{(i)}_{\mathrm{wat}, \mathrm{nom}}$ is initially estimated from the nominal heat output $Q_{\mathrm{nom}}^{(i)} \in \mathbb{R}$ using the steady-state forms of \eqref{eq:rad_dynamics} and~\eqref{eq:water_dynamics_revised}, yielding
\begin{align}
\dot{m}^{(i)}_{\mathrm{wat}, \mathrm{nom}}
=
\frac{Q_{\mathrm{nom}}^{(i)}}
{c_{\mathrm{p},\mathrm{wat}}
\left(
T_{\mathrm{sup}, \mathrm{nom}}
-
T_{\mathrm{ret}, \mathrm{nom}}
\right)}.
\label{eq:nominal_flow_rate}
\end{align}
It is then included in the vector of system parameters $\theta$ and adjusted using the proposed grey-box modeling approach.
In \eqref{eq:nominal_flow_rate}, $T_{\mathrm{sup}, \mathrm{nom}}, T_{\mathrm{ret},\mathrm{nom}} \in \mathbb{R}$ denote the nominal supply and return temperatures at which $Q^{(i)}_{\mathrm{nom}}$ is reported in the radiator datasheet. These temperatures are typically given as $T_{\mathrm{sup},\mathrm{nom}}=70\,^\circ\mathrm{C}$ and $T_{\mathrm{ret},\mathrm{nom}}=55\,^\circ\mathrm{C}$. Table~\ref{tab:hydraulic_balancing_rmse_comparison} compares the per-zone and overall RMSE values obtained with and without accounting for hydraulic balancing. It can be observed that accounting for hydraulic interactions reduces the overall RMSE from $1.03\,[\si{\degreeCelsius}]$ to $0.92\,[\si{\degreeCelsius}]$, corresponding to a reduction of approximately $11\%$.

\begin{table}[t]
\centering
\caption{Comparison of the open-loop RMSE obtained with and without
         considering hydraulic balancing.}
\label{tab:hydraulic_balancing_rmse_comparison}
\begin{tabular}{ccc}
\toprule
Zone & \multicolumn{2}{c}{RMSE $[\si{\degreeCelsius}]$} \\
\cmidrule(lr){2-3}
     & without hydraulic interactions & with hydraulic interactions \\
\midrule
2 & 0.577 & 0.569 \\
3 & 0.618 & 0.565 \\
4 & 1.117 & 0.909 \\
5 & 1.508 & 1.395 \\
\midrule
All & 1.03 & 0.92 \\
\bottomrule
\end{tabular}
\end{table}

\begin{remark}
    Note that this $11 \%$ reduction in RMSE is obtained in a relatively small building with six zones. This effect can be significantly more important in larger buildings with many more zones.
\end{remark}

By modeling the system structure as described in Section~\ref{sec:model} and using grey-box modeling to adjust its parameters, as discussed in Section~\ref{sec:grey-box}, we incorporate the resulting model into the Model Predictive Control (MPC) framework in the next section to regulate the radiator heat output and achieve the objectives for the case-study building.

%% file: s4.tex
At each sampling instant $t_k \coloneqq kt_s$, where \(k\in\mathbb{N}_0\) and $t_s \in \mathbb{R}_{> 0}$ is the sampling time, the current state $\hat{x}_k \in \mathbb{R}^n$ is first estimated using the available zone-temperature and valve-position measurements through the EKF described in Section~\ref{sec:sec3:warm-up}. Starting from this state estimate and using the available forecasts of the exogenous inputs, MPC predicts the future system behavior over a predefined prediction horizon $N_{\mathrm{p}}\in\mathbb{N}$; see Figure~\ref{fig:control_scheme}. Based on these predictions, MPC determines the optimal sequence of control inputs to maintain thermal comfort while reducing energy consumption. In the proposed framework, this optimal input sequence is given by
\[
U^*_k
\coloneqq
\begin{bmatrix}
u_{0|k}^{*\top} &
u_{1|k}^{*\top} &
\ldots &
u_{N_{\mathrm{p}}-1|k}^{*\top}
\end{bmatrix}^{\top},
\]
where
\begin{equation}
u^*_{i|k}
\coloneqq
\begin{bmatrix}
T^{*}_{\mathrm{sup},i|k} &
T^{*\top}_{\mathrm{des},i|k}
\end{bmatrix}^{\top}
\in\mathbb{R}^{1+n_{\mathrm{z}}},
\qquad
i \in \{0,1,\ldots,N_{\mathrm{p}}-1\}.
\nonumber
\end{equation}
Here, $T^*_{\mathrm{sup},i|k}\in\mathbb{R}$ is the optimal boiler supply temperature, and $T^*_{\mathrm{des},i|k}\in\mathbb{R}^{n_{\mathrm{z}}}$ denotes the vector of optimal temperature setpoints for the zones, both computed at time step $k$ for time step $k+i$. The first control input $u^*_{0|k}$ is then applied to the real system, and the optimization procedure is repeated at the next time instant using the newly available measurements. Since MPC relies on a mathematical prediction model, the model in \eqref{eq:overall_ct_system} can be used to formulate the MPC problem. However, because \eqref{eq:overall_ct_system} is nonlinear, we linearize the system to reduce the computational burden and enable the MPC problem to be formulated as a quadratic program. In particular, the nonlinear vector field $f$ in \eqref{eq:overall_ct_system} is approximated by its first-order Taylor expansion around a given operating point $\bar{p} \coloneqq \left(\bar{x},\bar{u},\bar{d} \right)$ as
\begin{equation}
\dot{x}(t) \approx
A_{\mathrm{c}}(\bar{p}) x(t) + B_{\mathrm{c}}(\bar{p}) u(t)
+ E_{\mathrm{c}}(\bar{p}) d(t) + c_{\mathrm{c}}(\bar{p}),
\label{eq:ct_lin}
\end{equation}
where
\begin{equation}
A_{\mathrm{c}}(\bar{p})
\coloneqq
\left.\frac{\partial f}{\partial x}\right|_{\bar{p}},
\qquad
B_{\mathrm{c}}(\bar{p})
\coloneqq
\left.\frac{\partial f}{\partial u}\right|_{\bar{p}},
\qquad
E_{\mathrm{c}}(\bar{p})
\coloneqq
\left.\frac{\partial f}{\partial d}\right|_{\bar{p}}, \nonumber
\end{equation}
and the affine term is given by
\begin{equation}
c_{\mathrm{c}}(\bar{p})
\coloneqq
f \left(\bar{x},\bar{u},\bar{d} \right)
-
A_{\mathrm{c}} \left(\bar{p} \right)\bar{x}
-
B_{\mathrm{c}} \left (\bar{p} \right)\bar{u}
-
E_{\mathrm{c}} \left(\bar{p} \right)\bar{d}. \nonumber
\end{equation}
The term $c_{\mathrm{c}}$ ensures that the affine model \eqref{eq:ct_lin} matches the nonlinear dynamics \eqref{eq:overall_ct_system} exactly at $\bar{p}$.
Assuming that $\bar{p}$ is fixed and that $u(t) = u_k$ and $d(t) = d_k$ are constant over each sampling interval $[kt_s, ~(k+1)t_s]$, the continuous-time affine model \eqref{eq:ct_lin} can be exactly discretized under zero-order hold as
\begin{equation}
x_{k+1}
=
A(\bar{p}) x_k + B(\bar{p}) u_k + E(\bar{p}) d_k + c(\bar{p}),
\label{eq:dt_lin_matrices}
\end{equation}
where
\begin{equation}
A(\bar{p}) = e^{A_{\mathrm{c}}(\bar{p})t_s},
\qquad
\begin{bmatrix}
B(\bar{p}) & E(\bar{p}) & c(\bar{p})
\end{bmatrix}
=
\left(\int_{0}^{t_s}
e^{A_{\mathrm{c}}(\bar{p})\tau}\mathrm{d}\tau\right)
\mathcal{B}_{\mathrm{c}}(\bar{p}), \nonumber
\end{equation}
with $\mathcal{B}_{\mathrm{c}}(\bar{p})
\coloneqq
\begin{bmatrix}
B_{\mathrm{c}}(\bar{p}) & E_{\mathrm{c}}(\bar{p}) & c_{\mathrm{c}}(\bar{p})
\end{bmatrix}
\in
\mathbb{R}^{n\times(n_u+n_d+1)}$.
\begin{remark}
    Alternatively, the nonlinear system \eqref{eq:overall_ct_system} could be discretized using for example RK4 \citep{RK4} and incorporated into a nonlinear MPC formulation, which could then be compared with the MPC based on the linearized model. However, such a comparison is beyond the main focus of this paper. Instead, we aim to show that, although accounting for the radiator valves and hydraulic interactions among the radiators introduces nonlinearities into the model, the resulting computational complexity can be alleviated through linearization.
\end{remark}
\begin{figure}
    \centering
    \includegraphics[width=0.75\linewidth]{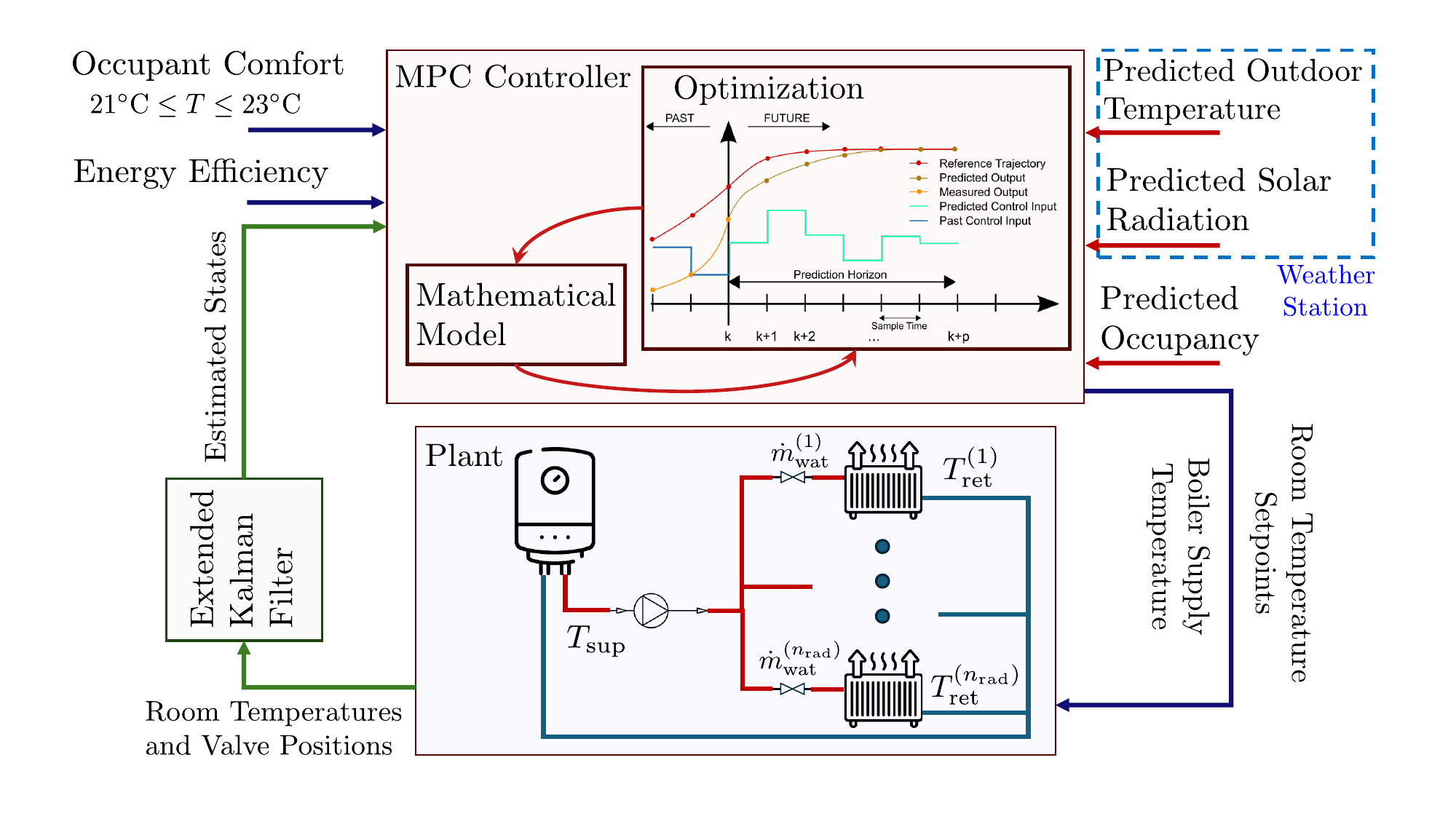}
    \caption{Overall topology of the proposed MPC scheme.}
    \label{fig:control_scheme}
\end{figure}

However, using a single linearized model over the entire MPC horizon may lead to inaccurate predictions. Therefore, we use a linear time-varying affine approximation, obtained by linearizing the nonlinear system \eqref{eq:overall_ct_system} at different operating points along the prediction horizon. To construct the operating points
\begin{align}
\bar{p}_{i|k}
\coloneqq
\begin{bmatrix}
\bar{x}_{i|k} & \bar{u}_{i|k} & d_{i|k}
\end{bmatrix},
\qquad
i \in \{0,1,\ldots,N_{\mathrm{p}}-1\}, \nonumber
\end{align}
a nominal input trajectory is first constructed at each time step~$k$. This trajectory is obtained by shifting the optimal input trajectory from MPC at previous time step one step forward and repeating its final input to preserve the prediction-horizon length. In particular, we consider the nominal input trajectory~$\bar{u}_{i|k}$, $i \in \{0,1,\ldots,N_{\mathrm p}-1\}$, as
\begin{align}
    \bar{u}_{i|k}
    \coloneqq
    \begin{bmatrix}
        \bar{T}_{\mathrm{sup},i|k} & \bar{T}_{\mathrm{des},i|k}^{\top}
    \end{bmatrix}^{\top}, \qquad
 \nonumber
\end{align}
where
\begin{align}
\bar{T}_{\mathrm{sup},i|k}
&\coloneqq
T_{\mathrm{sup},i+1|k-1}^{\ast}, \qquad
i \in \{0,1,\ldots,N_{\mathrm p}-2\}, \nonumber
\\
\bar{T}_{\mathrm{sup},N_{\mathrm p}-1|k}
&\coloneqq
T_{\mathrm{sup},N_{\mathrm p}-1|k-1}^{\ast}, \nonumber
\end{align}
and similarly
\begin{align}
\bar{T}_{\mathrm{des},i|k}
&\coloneqq
T_{\mathrm{des},i+1|k-1}^{\ast},
\qquad
i \in \{0,1,\ldots,N_{\mathrm p}-2\}, \nonumber
\\
\bar{T}_{\mathrm{des},N_{\mathrm p}-1|k}
&\coloneqq
T_{\mathrm{des},N_{\mathrm p}-1|k-1}^{\ast}. \nonumber
\end{align}
Using $\bar{u}_{i|k}$ together with the known exogenous inputs $d_{i|k}$, the corresponding nominal state trajectory is then propagated through the nonlinear system 
\begin{align}
\bar{x}_{i+1|k}
&\coloneqq
\bar{x}_{i|k}
+
\int_{0}^{t_{\mathrm{s}}}
f\!\left(
\bar{\xi}(\tau),\,
\bar{u}_{i|k},\,
d_{i|k}
\right)
\mathrm{d}\tau,
\qquad
i \in \{0,1,\ldots,N_{\mathrm{p}}-1\},
\nonumber
\end{align}
where $\bar{x}_{0|k} \coloneqq \hat{x}_{k}$, and
$\dot{\bar{\xi}}(\tau) = f\!\left(\bar{\xi}(\tau), \bar{u}_{i|k}, d_{i|k}\right)$, $\tau \in [0,t_{\mathrm{s}}]$, with $\bar{\xi}(0) \coloneqq \bar{x}_{i|k}$. Note that, at $k=0$, no previous optimal input trajectory is available. Therefore, the nominal input trajectory must be initialized, for example, using the midpoint of its feasible range.

 Using the linearized prediction model \eqref{eq:dt_lin_matrices}, the proposed MPC aims to provide occupant comfort in each targeted zone while minimizing the boiler supply temperature, which is related to the energy consumption of the building. To satisfy the former objective while avoiding possible infeasibility of the comfort requirements, we introduce the nonnegative slack-variable vector 
 \begin{align} 
 \delta_k \coloneqq 
 \begin{bmatrix} 
 \delta_{1|k} & \delta_{2|k} & \ldots & \delta_{N_{\mathrm{p}} | k}
 \end{bmatrix}^{\top} \in \mathbb{R}
_{\geq 0}^{N_{\mathrm{p}}}
, \nonumber 
\end{align} 
and impose the softened comfort constraints 
\begin{align} T_{\min,k+i+1}
- \delta_{i+1|k}\mathbb{1}_{n_{\mathrm{z, tar}}}
\leq C_{\mathrm{z, tar}}x_{i+1|k} \leq T_{\max,k+i+1}
+ \delta_{i+1|k}\mathbb{1}_{n_{\mathrm{z, tar}}}, \quad i \in \{0,1,\ldots,N_{\mathrm{p}}-1 \}, \nonumber
\end{align}
where \mbox{$\mathbb{1}_{n_{\mathrm{z, tar}}} \in \mathbb{R}^{n_{\mathrm{z, tar}}}$} is the vector of one, $n_{\mathrm{z,tar}} \in \mathbb{N}$ denotes the number of targeted zones whose temperatures are to be controlled, and $C_{\mathrm{z, tar}} \in \mathbb{R}^{n_{\mathrm{z, tar}} \times n}$ extracts their temperatures from the state vector. Moreover, $T_{\min,k+i+1}, T_{\max,k+i+1} \in \mathbb{R}^{n_{\mathrm{z,tar}}}$ denote the vectors of element-wise minimum and maximum comfort temperatures for these zones at time step $k+i+1$, respectively; thus, the temperature of each targeted zone is constrained within a time-varying comfort range, which can, for example, depend on whether the building is occupied or unoccupied.
To penalize any violation of the comfort range, we penalize the slack variable $\delta_k$ as
\begin{align}
    J_{\mathrm{comf}}(\delta_k) \coloneqq \delta_k^{\top} \Lambda_{\delta} \delta_k \nonumber
\end{align}
using diagonal weighting matrix \(\Lambda_{\delta}\in \mathbb{R}^{N_{\mathrm{p}} \times
N_{\mathrm{p}}}\). To satisfy the second objective, namely reducing energy consumption, high boiler supply temperatures over the prediction horizon are penalized in the MPC framework through the objective term
\begin{align}
J_{\mathrm{sup}}(U_{\mathrm{sup},k})
\coloneqq
U_{\mathrm{sup},k}^{\top}
\Lambda_{\mathrm{sup}} U_{\mathrm{sup},k},
\nonumber
\end{align}
where 
\begin{align}
 U_{\mathrm{sup},k} \coloneqq 
\begin{bmatrix}T_{\mathrm{sup},0|k} & T_{\mathrm{sup},1|k} & \dots  & T_{\mathrm{sup},N_{\mathrm{p}} - 1|k}
\end{bmatrix}^{\top} \in \mathbb{R}^{N_{\mathrm{p}}}, \nonumber
\end{align}
and
$\Lambda_{\mathrm{sup}} \in\mathbb{R}^{N_{\mathrm{p}}\times N_{\mathrm{p}}}$ is a diagonal weighting matrix. Moreover, each element of the boiler supply temperature sequence is also constrained to admissible values as
\begin{equation}
T_{\mathrm{sup},\min}\mathbb{1}_{N_{\mathrm{p}}}
\leq
U_{\mathrm{sup},k}
\leq
T_{\mathrm{sup},\max}\mathbb{1}_{N_{\mathrm{p}}}, \nonumber
\end{equation}
where \(T_{\mathrm{sup},\min}, T_{\mathrm{sup},\max}\in\mathbb{R}\) represent the lower and upper bounds on the admissible boiler supply temperature, respectively.

As a result, at each time step $k\in\mathbb{N}_0$, given the current state estimate $\hat{x}_k$ obtained using the EKF, the following MPC problem is solved:
\begin{subequations} \label{eq:proposed_MPC}
\begin{align} 
    (U^\ast_{k}, \delta^\ast_k) \coloneqq \text{arg}&\min_{U_k, \delta_k} \quad  
    J_{\mathrm{comf}}(\delta_k) + J_{\mathrm{sup}}(U_{\mathrm{sup},k})  \label{eq:MPC_cost} \\[2mm]
    \text{subject to} \qquad  \nonumber 
    & x_{0|k} \coloneqq \hat{x}_k, \\[2mm]
    & x_{i+1|k} = A(\bar{p}_{_{i|k}}) x_{i|k} + B(\bar{p}_{_{i|k}}) u_{i|k} + E(\bar{p}_{_{i|k}})d_{i|k} + c(\bar{p}_{_{i|k}}), \\[2mm]
    & T_{\mathrm{sup}, \min} \mathbb{1}_{N_{\mathrm{p}}}  \leq  U_{\mathrm{sup}, k} \leq T_{\mathrm{sup}, \max}\mathbb{1}_{N_{\mathrm{p}}},  \\[2mm]
    & T_{\min, (k + i + 1)} - \delta_{i+1|k}\mathbb{1}_{n_{\mathrm{z, tar}}} \leq C_{\mathrm{z, tar}}x_{i+1|k} \leq T_{\max, (k + i + 1)} + \delta_{i+1|k}\mathbb{1}_{n_{\mathrm{z,tar}}}, \\[2mm]
    &U_k
    \coloneqq
    \begin{bmatrix}
    u_{0|k}^{\top} &
    u_{1|k}^{\top} &
    \ldots &
    u_{N_{\mathrm{p}}-1|k}^{\top}
    \end{bmatrix}^{\top}, \quad u_{i|k}
    \coloneqq
    \begin{bmatrix}
    T_{\mathrm{sup},i|k} &
    T^{\top}_{\mathrm{des},i|k}
    \end{bmatrix}^{\top}, \\[2mm]
    &U_{\mathrm{sup},k} \coloneqq 
    \begin{bmatrix}T_{\mathrm{sup},0|k} & T_{\mathrm{sup},1|k} & \dots  & T_{\mathrm{sup},N_{\mathrm{p}} - 1|k}
    \end{bmatrix}^{\top},
\end{align}
\end{subequations}
where $i \in \{0,1, \dots, N_{\mathrm{p}}-1\}$.

To validate the proposed MPC framework~\eqref{eq:proposed_MPC}, we conduct numerical case studies and real-building experiments in Sections~\ref{sec:numerical_case_study} and~\ref{sec:experimental_validation} below, respectively, and compare its performance with existing MPC strategies from the literature.

\input{s4_1}

\input{s4_2}

%% file: s4_1.tex
\subsection{Numerical Case Study} \label{sec:numerical_case_study}
In the numerical case study, the objective is to evaluate the benefits of accounting for hydraulic interactions among the radiators and of incorporating valve control into the control design for dynamic hydraulic balancing. To this end, the performance of the proposed MPC~\eqref{eq:proposed_MPC} is compared against two benchmark MPC formulations, assuming that the actual system follows the dynamics in~\eqref{eq:overall_ct_system}, since Section~\ref{sec:effec_hydraulic} shows that \eqref{eq:overall_ct_system} is more representative of the actual system. The two benchmark formulations are then constructed to isolate the individual effects of valve control and hydraulic interactions. Specifically, Benchmark MPC~1 assumes that all radiator valves are fully open and are not available as control inputs, whereas Benchmark MPC~2 treats the valve positions as control inputs but neglects the hydraulic interactions among the radiators. In particular, it assumes a nominal flow rate for each radiator according to \eqref{eq:nominal_flow_rate}, and the actual flow rate through each radiator is determined as a proportion of this nominal flow rate based on the corresponding valve position.

To compare the performance between these MPC frameworks, we use several criteria. First, we use total comfort violation $[\si{\kelvin}]$ defined as
\begin{align}
    \sum_k \mathbb{1}_{n_{\mathrm{z,tar}}}^{\top} \left(\varepsilon_k^- + \varepsilon_k^+ \right), \nonumber
\end{align}
where $\varepsilon_k^-, \varepsilon_k^+ \in \mathbb{R}_{> 0}^{n_{\mathrm{z,tar}}}$ are lower and upper comfort violations, respectively, at sampling instant $t_k$, and are defined as
\begin{align}
    \varepsilon_k^- \coloneqq \max \left (T_{\mathrm{min},k} - \Bigl[T_{\mathrm{meas},k}^{(i)}\Bigr]_{i\in\mathcal{Z}},\, 0 \right), \qquad
    \varepsilon_k^+ \coloneqq  \max \left( \Bigl[T_{\mathrm{meas},k}^{(i)}\Bigr]_{i\in\mathcal{Z}} - T_{\mathrm{max},k},\, 0 \right). \nonumber
\end{align}
Here, the operator $\max(x,0)$ denotes the elementwise positive part of the vector $x$. Moreover, we use maximum comfort violation defined as $\varepsilon_{\mathrm{max}} = \max (\varepsilon_{\mathrm{max}}^-, \varepsilon_{\mathrm{max}}^+)$, where
\begin{align}
    \varepsilon_{\mathrm{max}}^- \coloneqq \max_k ~ \left\lVert \varepsilon_{k}^-  \right\rVert_\infty, \qquad
    \varepsilon_{\mathrm{max}}^+ \coloneqq \max_k ~ \left\lVert \varepsilon_{k}^+  \right\rVert_\infty, \nonumber
\end{align}
with $\Vert . \Vert_\infty$ denoting the infinity norm. We also evaluate the cumulative boiler supply temperature applied to the system over the simulation period, defined as
\begin{align}
    \sum_k T_{\mathrm{sup},0|k}^\ast. \nonumber
\end{align}
Lastly, we evaluate the accumulated realized cost defined as
\begin{align}
\sum_k \biggl( \lambda_{1, \mathrm{sup}} \left(T_{\mathrm{sup},0|k}^\ast \right)^2 + \lambda_{1,\delta} &\big\Vert \bigl[
    {\varepsilon_k^-} ^\top ~~ {\varepsilon_k^+}^\top
\bigr]^\top
\big\Vert_\infty^2 \biggr), \nonumber
\end{align}
where $\lambda_{1,\delta}, \lambda_{1, \mathrm{sup}} \in \mathbb{R}$ are the first diagonal entries of the weighting matrices \(\Lambda_{\delta}\) and \(\Lambda_{\mathrm{sup}}\), respectively, which are defined in the MPC cost function in \eqref{eq:MPC_cost}. 

All MPC frameworks are implemented using the same weighting matrices, \(\Lambda_{\delta}\) and \(\Lambda_{\mathrm{sup}}\), in \eqref{eq:MPC_cost}, with a sampling time of $t_s = 15 \, \si{\min}$, and a prediction horizon of $N_{\mathrm{p}} = 15$. The comfort bounds are set to $T_{\mathrm{min}} = 20\,\si{\degreeCelsius}$ and $T_{\mathrm{max}} = 22\,\si{\degreeCelsius}$ during unoccupied hours, and $T_{\mathrm{min}} = 22.5\,\si{\degreeCelsius}$ and $T_{\mathrm{max}} = 24\,\si{\degreeCelsius}$ during occupied hours. Note that the comfort bounds are chosen to be more demanding than typical values to create more challenging scenarios for the MPC and to match the experimental case studies presented in Section~\ref{sec:experimental_validation} below. Using the above mentioned criteria, their performances are summarized in Table~\ref{tab:mpc_benchmark_comparison}. The results show that the proposed MPC framework significantly outperforms both benchmark formulations. In particular, relative to the proposed MPC, the realized cost is approximately \(92\%\) higher for Benchmark MPC~1 and approximately \(33.5\%\) higher for Benchmark MPC~2. Moreover, the proposed MPC reduces the total comfort violation from $70.90\,\si{\degreeCelsius}$ for Benchmark MPC~1 and $47.22\,\si{\degreeCelsius}$ for Benchmark MPC~2 to $34.31\,\si{\degreeCelsius}$, corresponding to reductions of $51.6\%$ and $27.3\%$, respectively. The corresponding zone temperatures are shown in Figures~\ref{fig:bench3}, \ref{fig:bench1}, and~\ref{fig:bench2}, respectively. Moreover, Figure \ref{fig:bench3_Tsup} depicts the optimized boiler supply temperature obtained using the proposed MPC framework, along with the outdoor temperature and the solar irradiance incident on the south-facing walls. The south-facing orientation is shown as a representative example, since the incident solar radiation varies with wall orientation.

\begin{remark}
    When the radiator valves are not included as control inputs in the prediction model in \eqref{eq:overall_ct_system} and are instead assumed to remain fixed, the resulting system behaves more closely to a linear system. Consequently, Benchmark MPC~1 shows closer agreement between the predicted and simulated zone-temperature trajectories, as observed in Figure~\ref{fig:bench1}, than the proposed MPC and Benchmark MPC~2, as shown in Figures~\ref{fig:bench3} and \ref{fig:bench2}, respectively.
\end{remark}

\begin{remark}
    Note that, for each MPC framework, the corresponding modeling assumptions are also incorporated into the EKF observer used to estimate the states at each time step. For example, for Benchmark MPC~2, the corresponding observer model neglects the hydraulic interactions among the radiators. Moreover, before running the MPC frameworks, a warm-up period of approximately four days is simulated prior to $k=0$ to warm-start the observer and obtain a valid initial state estimate~$\hat{x}_0$.
\end{remark}

\begin{table}[pos=t]
    \centering
    \caption{Performance comparison of the proposed MPC for dynamic hydraulic balancing with two benchmark MPC formulations.}
    \label{tab:mpc_benchmark_comparison}

    \small
    \setlength{\tabcolsep}{10pt}

    \begin{tabular}{lcccc}
        \toprule
        MPC benchmark
        & \makecell{Total comfort\\violation}
        & \makecell{Maximum comfort\\violation}
        & $\sum T_{\mathrm{sup}}$
        & \makecell{Rel. cost\\change} \\
        & [\si{\degreeCelsius}]
        & [\si{\degreeCelsius}]
        & [\si{\degreeCelsius}]
        & [-] \\
        \midrule

        Proposed MPC
        & 34.31
        & 1.35
        & 6746
        & 0\% \\

        No-valve-control MPC
        & 70.90
        & 1.44
        & 6228
        & 92\% \\

        Nominal-flow-based MPC
        & 47.22
        & 1.40
        & 6798
        & 33.5\% \\
        \bottomrule
    \end{tabular}
\end{table}

\begin{figure}
    \centering
    \includegraphics[width=0.85\linewidth]{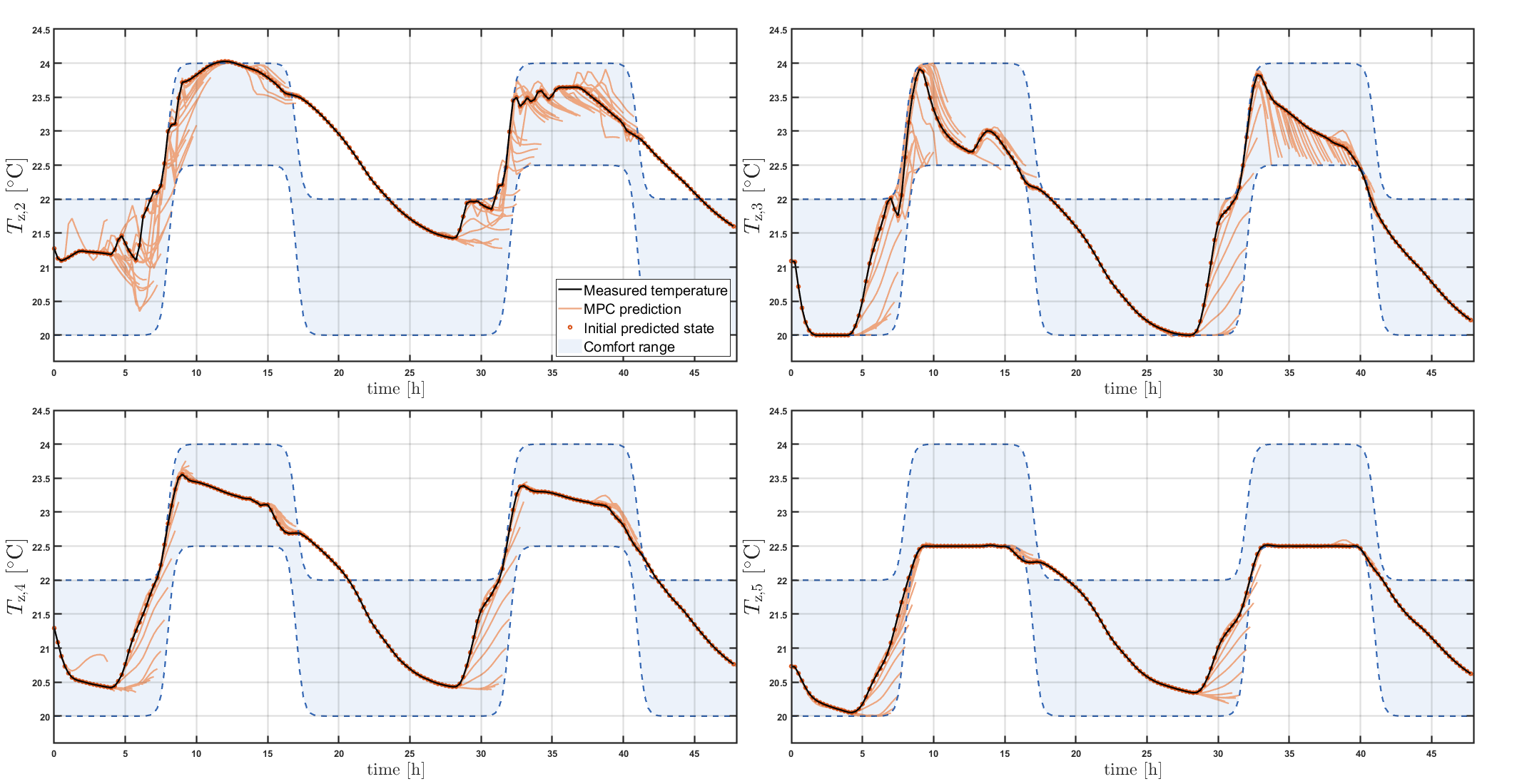}
    \caption{Zone temperatures under \textit{the proposed MPC~\eqref{eq:proposed_MPC} for dynamic hydraulic balancing}, together with the considered comfort bounds and the corresponding MPC-predicted temperature trajectories.}
    \label{fig:bench3}
\end{figure}

\begin{figure}
    \centering
    \includegraphics[width=0.62\linewidth]{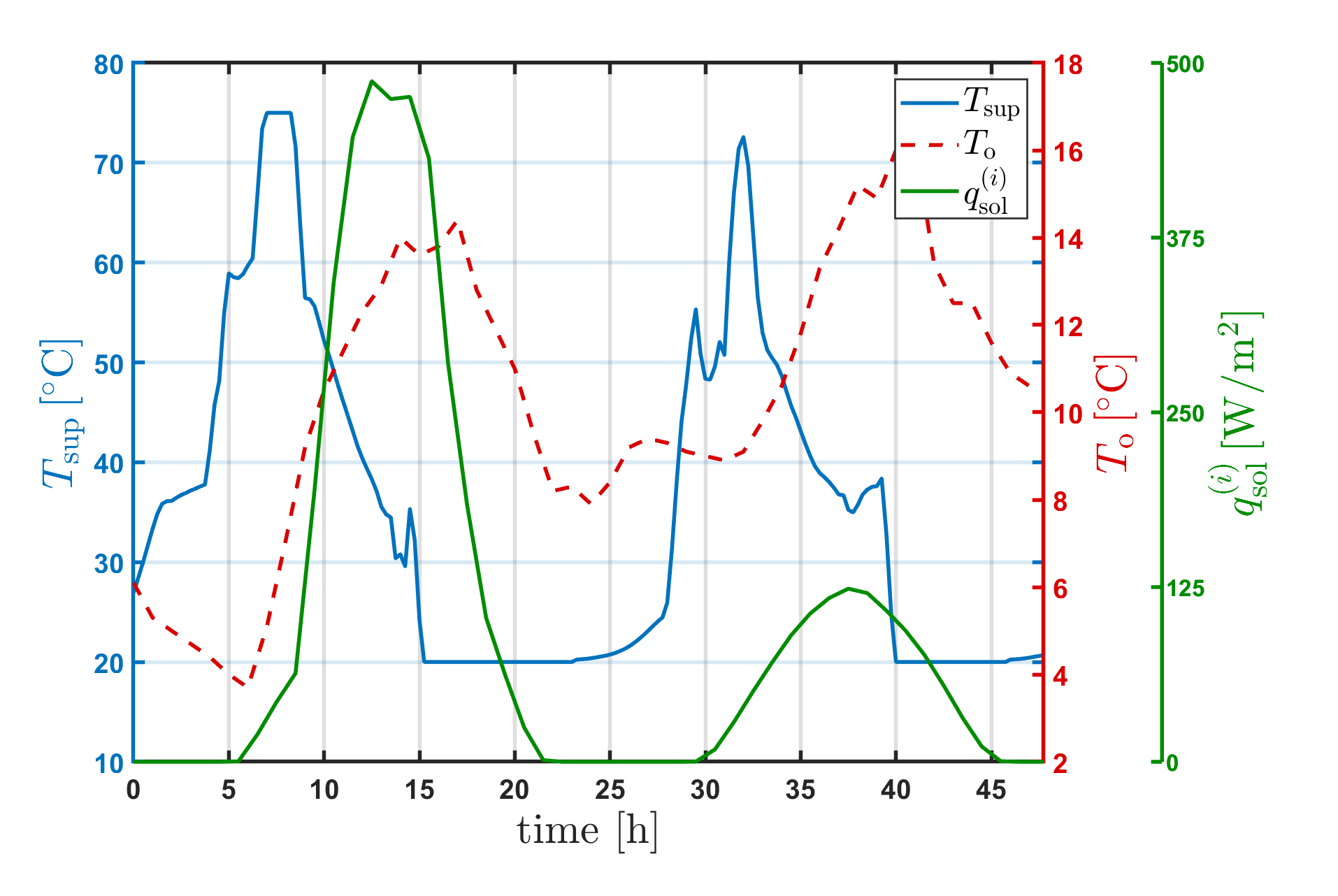}
    \caption{Supply temperature obtained using \textit{the proposed MPC~\eqref{eq:proposed_MPC}}, together with the outdoor temperature and the solar irradiance incident on south-facing walls.}
    \label{fig:bench3_Tsup}
\end{figure}

\begin{figure}
    \centering
    \includegraphics[width=0.85\linewidth]{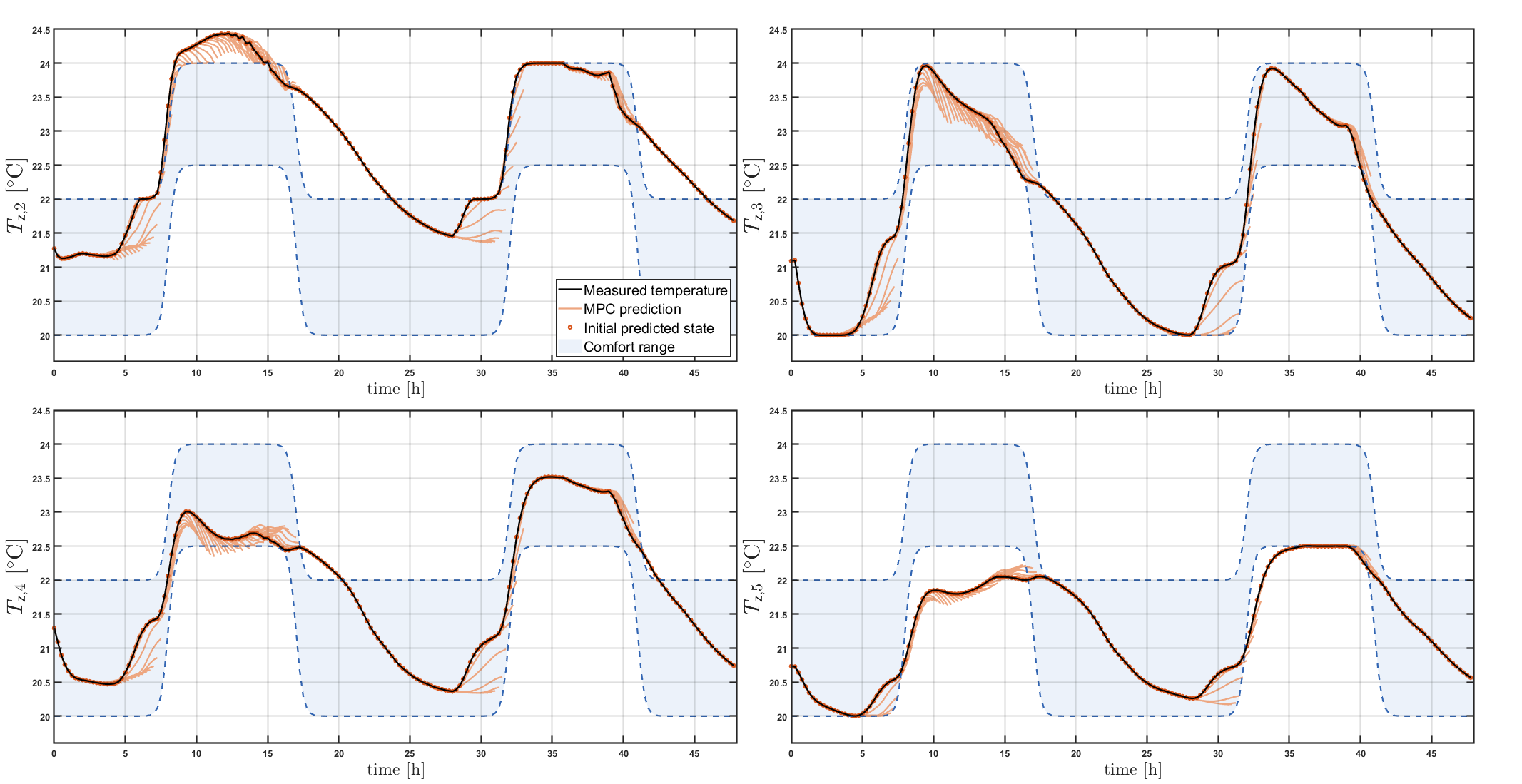}
    \caption{Zone temperatures under \textit{the MPC formulation without radiator-valve control}, where all radiator valves are fully open, together with the considered comfort bounds and the corresponding MPC-predicted temperature trajectories.}
    \label{fig:bench1}
\end{figure}

\begin{figure}
    \centering
    \includegraphics[width=0.85\linewidth]{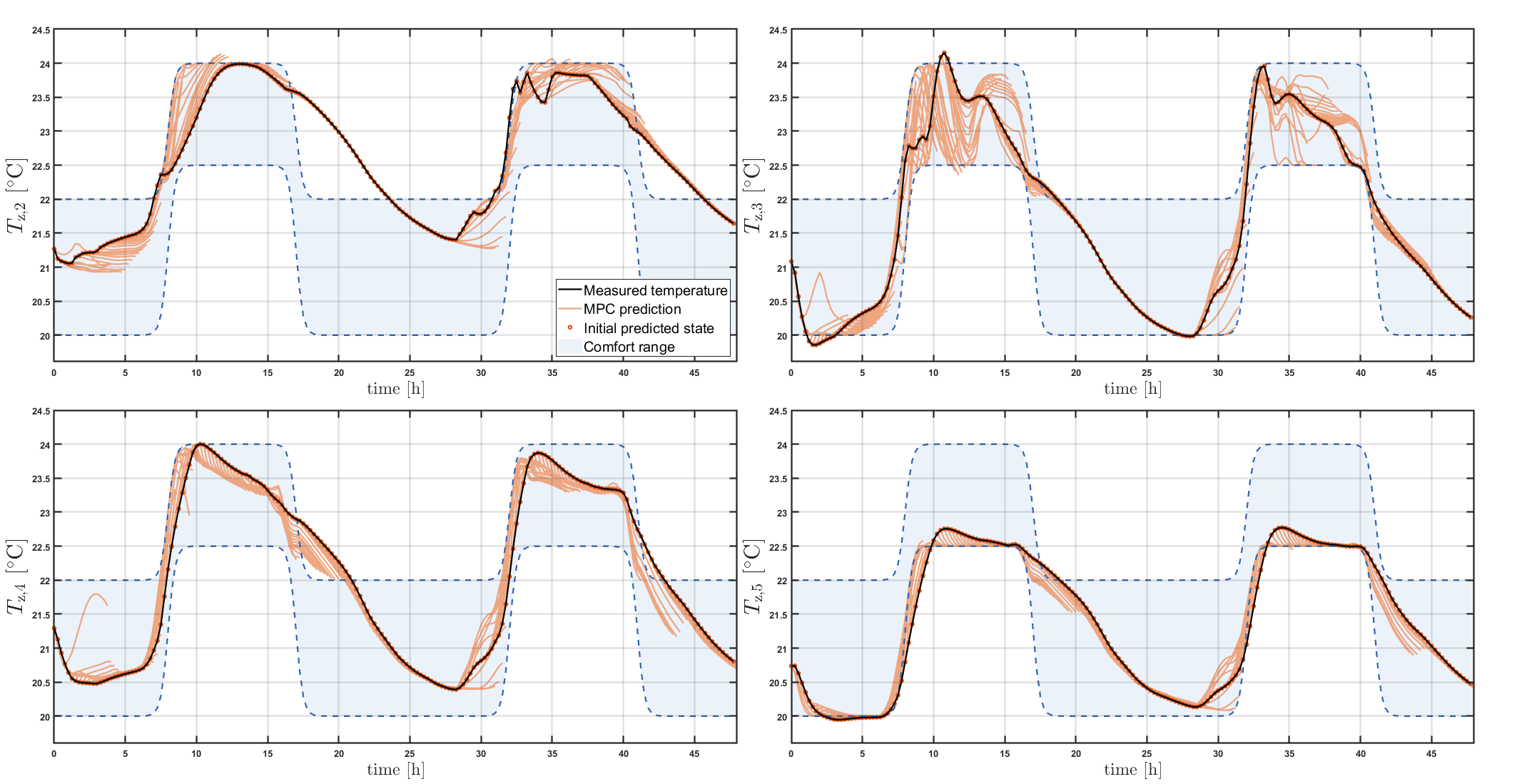}
    \caption{Zone temperatures under \textit{the MPC that ignores hydraulic interactions among the radiators}, together with the considered comfort bounds and the corresponding MPC-predicted temperature trajectories.}
    \label{fig:bench2}
\end{figure}

%% file: s4_2.tex
\subsection{Experimental Validation} \label{sec:experimental_validation}

To validate the proposed MPC \eqref{eq:proposed_MPC} for dynamic hydraulic balancing, we also implement it on the Rensen building (see Figure \ref{fig:exterior_scheme}). The building is equipped with a \textit{Priva} building management system \citep{priva_bms}, through which all relevant sensors are connected, including the zone-temperature sensors, valve-position sensors, and the boiler supply- and return-temperature sensors. The platform enables real-time access to and extraction of the corresponding measurement data. It also allows us to provide setpoints for the boiler supply temperature, $T_{\mathrm{sup}}$, and the zone temperature setpoints, $T_{\mathrm{des}}^{(i)}$, $i \in \{2,3,4,5\}$; see Figure~\ref{fig:platform_rensen}. Thus, the proposed MPC in \eqref{eq:proposed_MPC} reads the sensor measurements at each time step through this platform and provides the corresponding setpoints.

\begin{figure}
    \centering

    \begin{subfigure}[t]{0.48\linewidth}
        \centering
        \includegraphics[width=\linewidth]{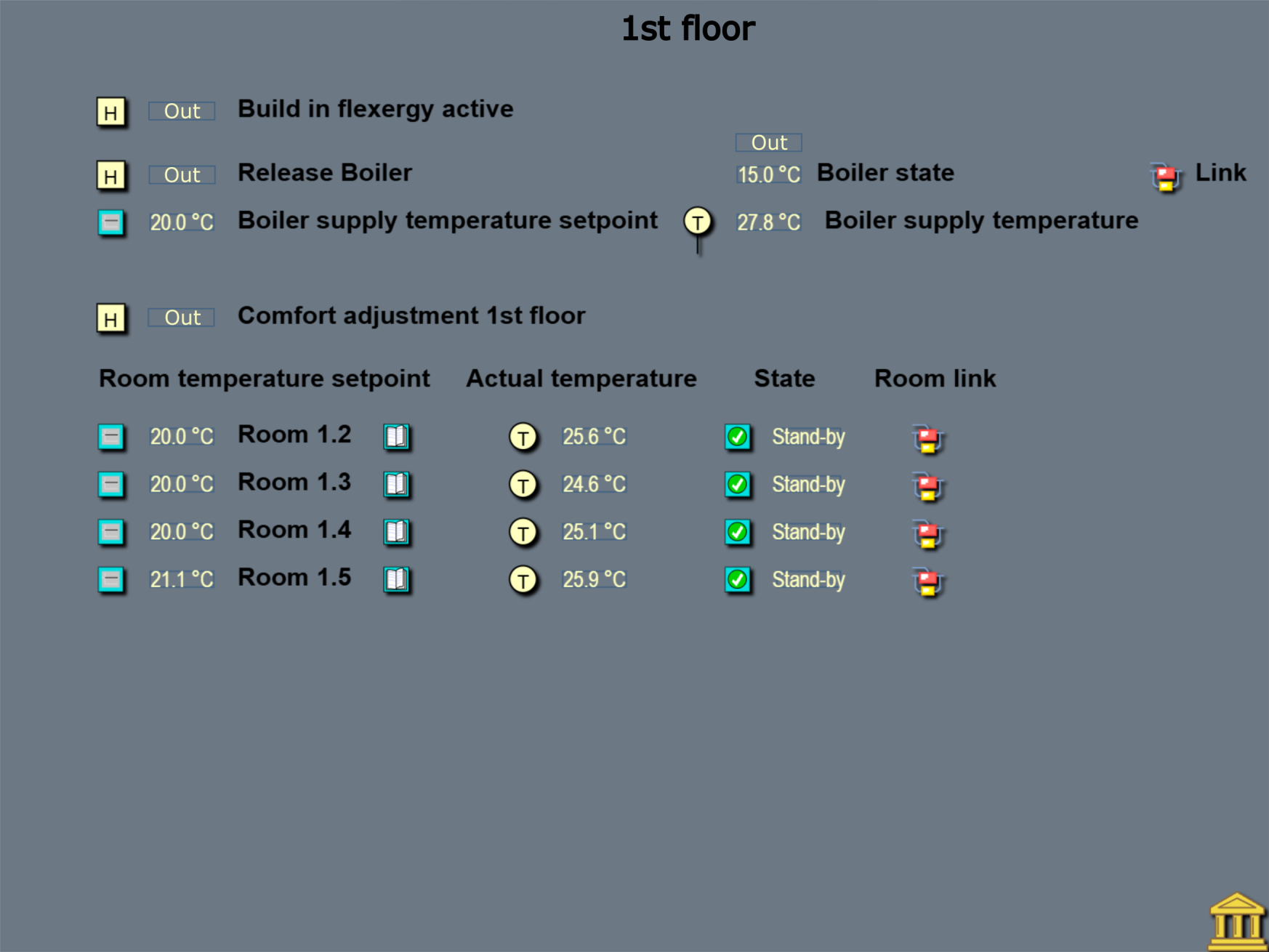}
        \caption{Boiler supply-temperature setpoint, actual boiler supply temperature, zone-temperature setpoints, and actual zone temperatures.}
        \label{fig:platform_rensen_room}
    \end{subfigure}
    \hfill
    \begin{subfigure}[t]{0.48\linewidth}
        \centering
        \includegraphics[width=\linewidth]{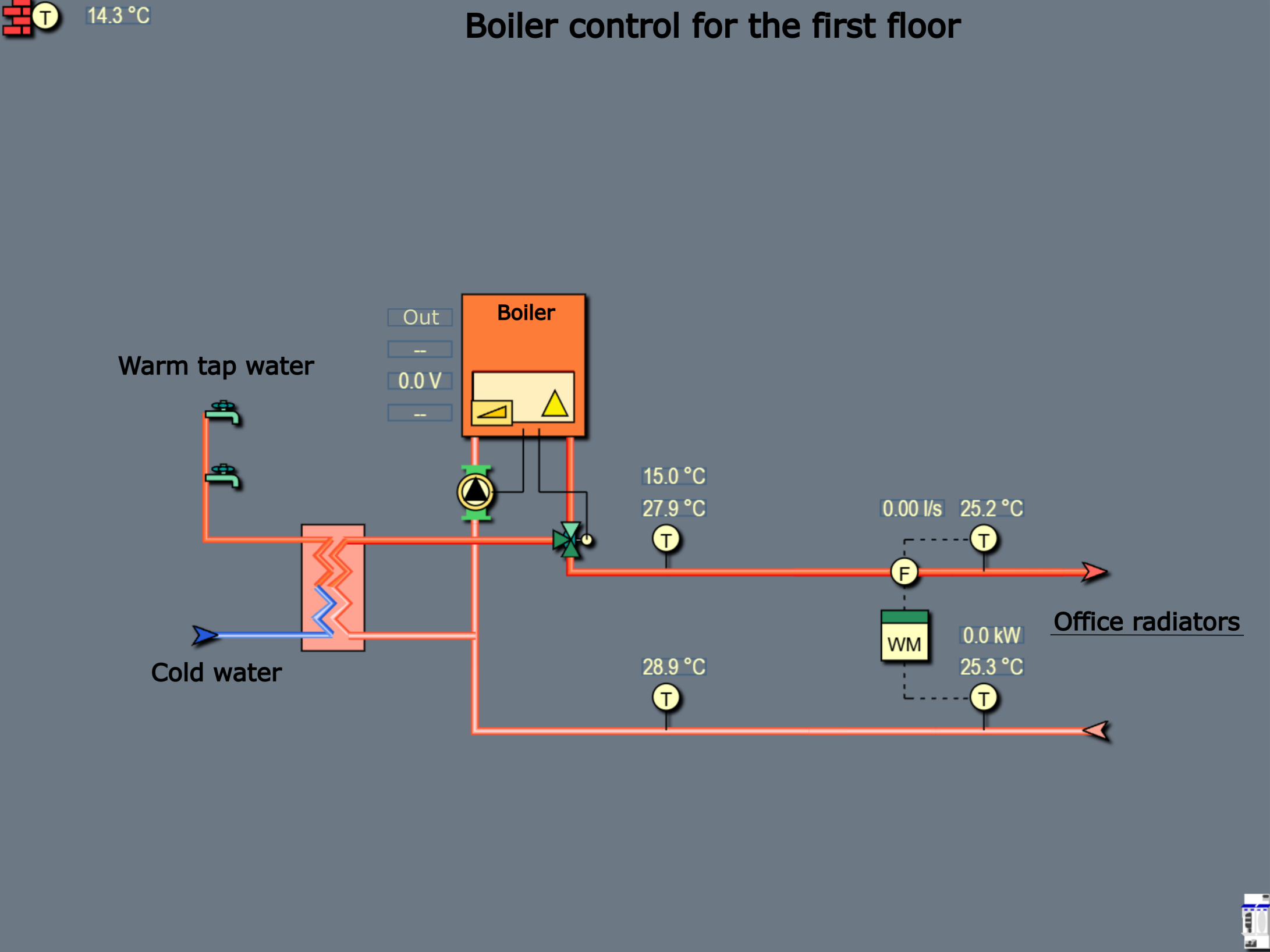}
        \caption{The boiler supply temperature, water flow rate, and return-water temperature to the boiler.}
        \label{fig:platform_rensen_boiler}
    \end{subfigure}

    \caption{Priva building control platform through which the sensors and actuators of the case-study building are connected.}
    \label{fig:platform_rensen}
\end{figure}

Given that the experiments were conducted in May in the Netherlands, when the daytime outdoor temperature was relatively high, we performed them at night. Moreover, to create a more challenging scenario for the MPC controllers, the comfort bounds are set higher than typical values, namely, to $[20\,^\circ\mathrm{C},\,22\,^\circ\mathrm{C}]$ during unoccupied hours and $[22.5\,^\circ\mathrm{C},\,24\,^\circ\mathrm{C}]$ during occupied hours. The zone temperatures and boiler supply temperature obtained using the proposed MPC in \eqref{eq:proposed_MPC} are depicted in Figure~\ref{fig:MPC_case_study}. To demonstrate the effectiveness of the proposed MPC, we compare its performance with that of an MPC that does not control the radiator valves, with all valves kept fully open. Using the same performance criteria introduced in Section~\ref{sec:numerical_case_study}, the comparison results are summarized in Table~\ref{tab:exp_mpc_benchmark_comparison}. As shown in Table~\ref{tab:exp_mpc_benchmark_comparison}, the proposed MPC~\eqref{eq:proposed_MPC} reduces the total comfort violation from $18.2\,\si{\degreeCelsius}$ to $4.8\,\si{\degreeCelsius}$ compared with the benchmark MPC, a 
$73.6\%$ reduction, while operating at a lower cumulative boiler supply temperature. The main reason is that the benchmark MPC cannot control the valves, thus raising the boiler supply temperature to maintain comfort in one zone (here, zone~5) simultaneously raises the temperature in the other zones, pushing them outside the comfort range, since the individual flow rates cannot be adjusted; see Figure~\ref{fig:MPC_no_valve_room}.

\begin{remark}
    Comparing different controllers over a short period in real experiments may not be entirely conclusive, since the exogenous inputs (e.g., outdoor temperature) acting on the system are not completely identical across the two experiments. However, the experiments are conducted on two consecutive days to reduce these differences.
\end{remark}

\begin{figure}
    \centering

    \begin{subfigure}[t]{0.49\linewidth}
        \centering
        \includegraphics[width=\linewidth]{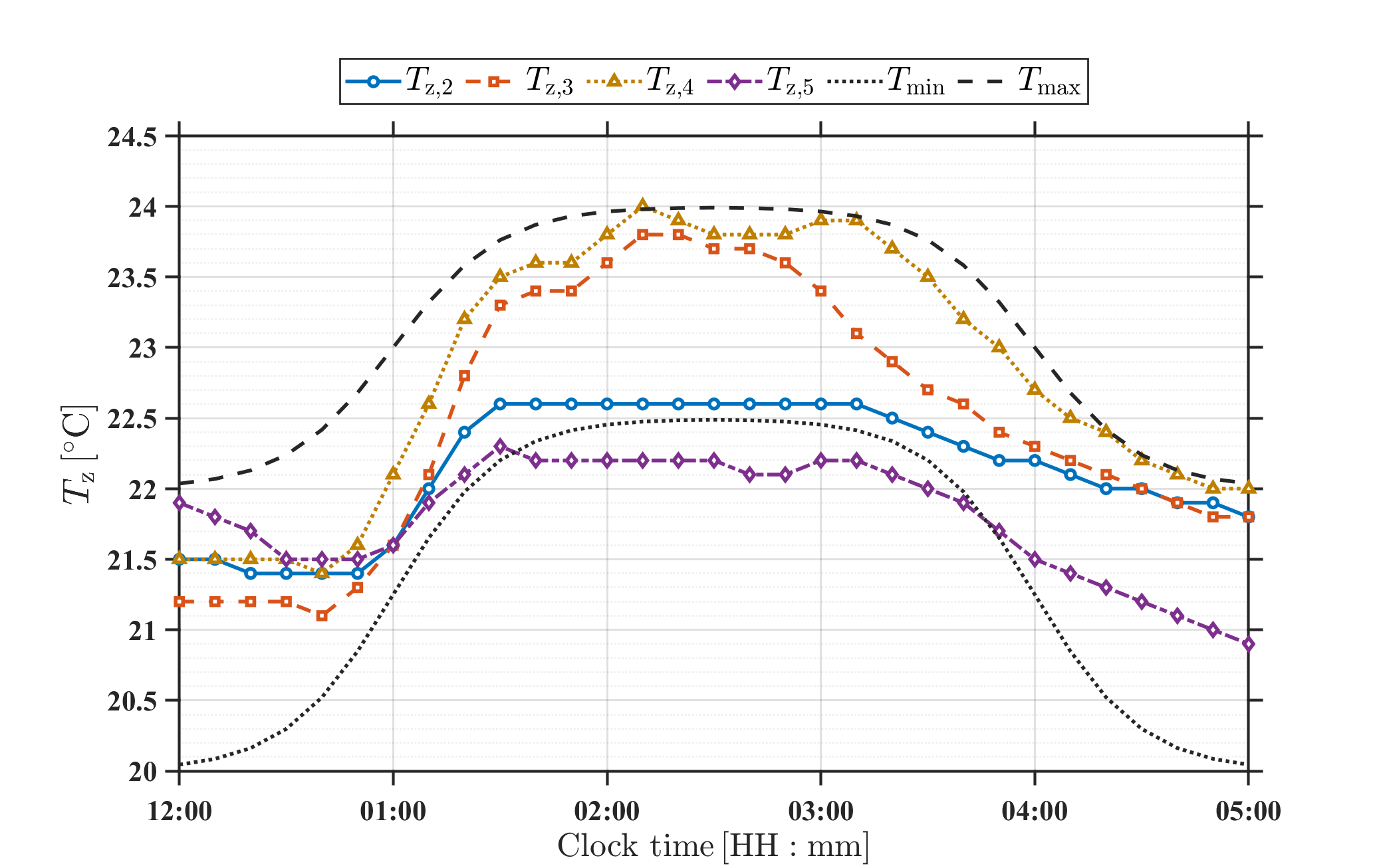}
        \caption{The zone temperatures.}
        \label{fig:MPC_case_study_room}
    \end{subfigure}
    \hfill
    \begin{subfigure}[t]{0.49\linewidth}
        \centering
        \includegraphics[width=\linewidth]{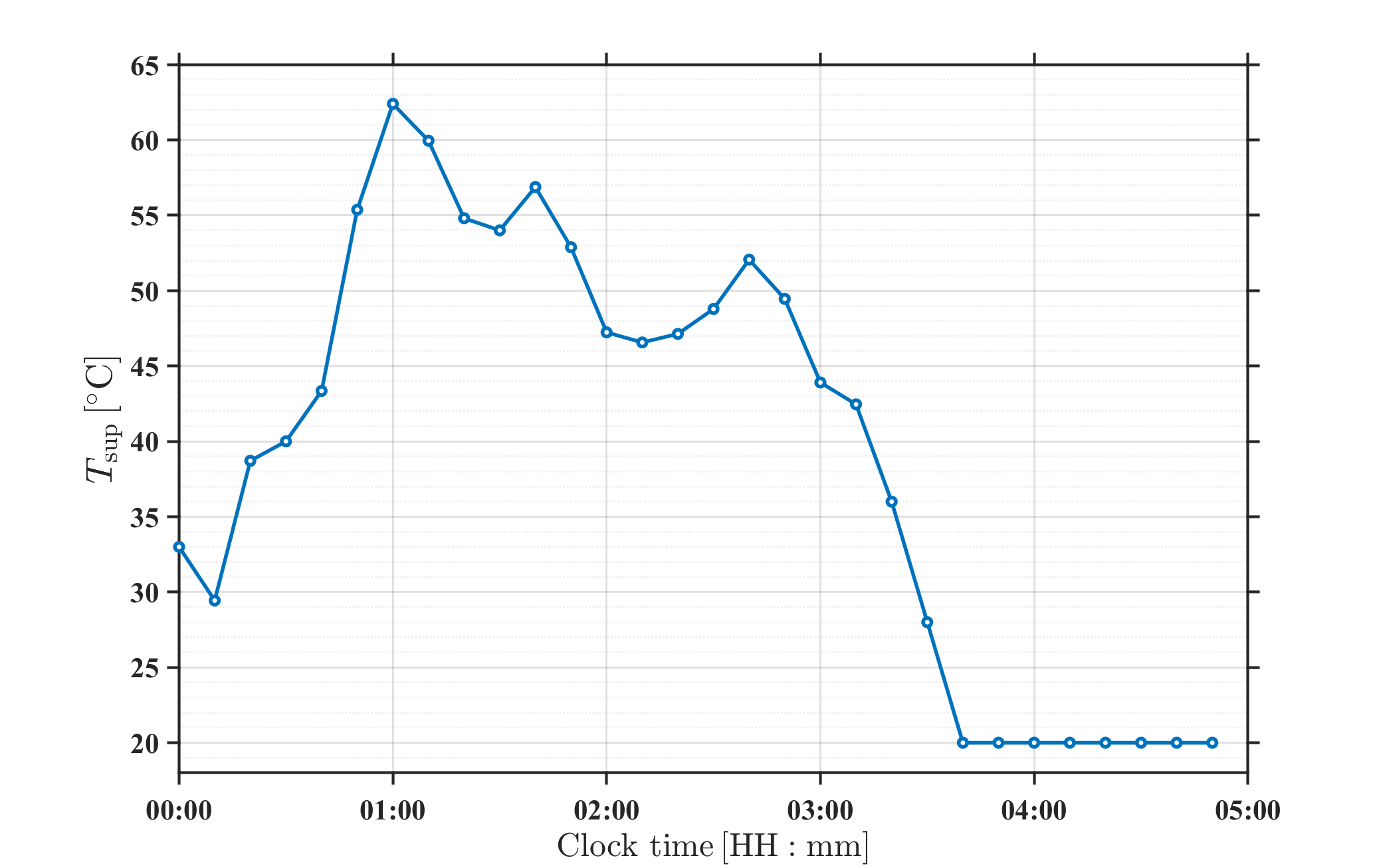}
        \caption{The boiler supply temperature.}
        \label{fig:MPC_case_study_boiler}
    \end{subfigure}

    \caption{Experimental results obtained by applying the proposed MPC \eqref{eq:proposed_MPC} to the case-study building.}
    \label{fig:MPC_case_study}
\end{figure}

\begin{figure}
    \centering
    \includegraphics[width=0.6\linewidth]{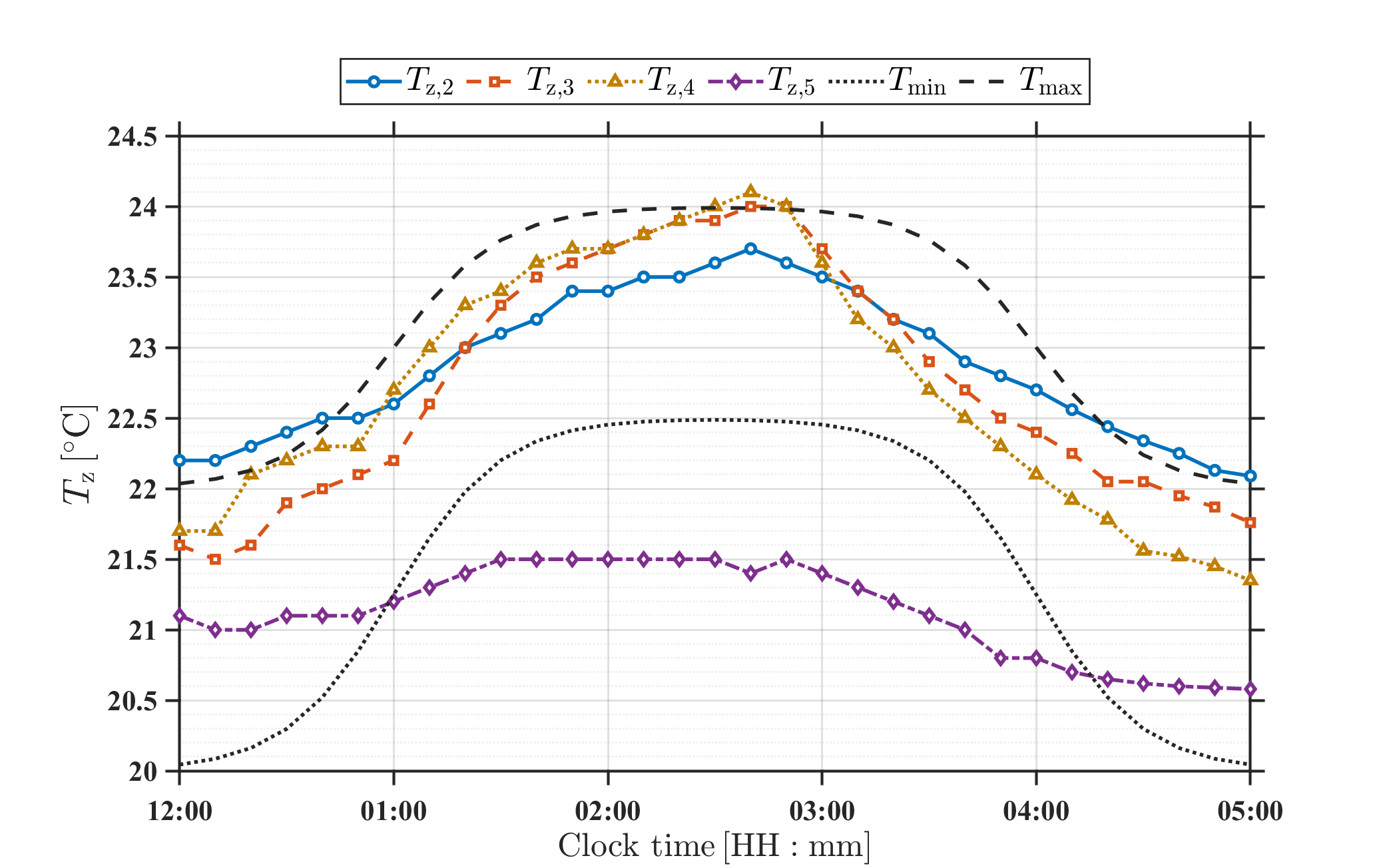}
    \caption{Experimental zone temperatures in the case-study building under the benchmark MPC, which does not adjust the radiator valves and can only control the boiler supply temperature. Shown for comparison with the proposed MPC in Figure~\ref{fig:MPC_case_study}.}
    \label{fig:MPC_no_valve_room}
\end{figure}

\begin{table}[t]
    \centering
    \caption{Experimental performance comparison of the proposed MPC and the benchmark MPC on the Rensen building.}
    \label{tab:exp_mpc_benchmark_comparison}
    \begin{tabular}{lcccc}
        \toprule
        MPC benchmark
        & Total comfort violation
        & Maximum comfort violation
        & $\sum T_{\mathrm{sup}}$ \\
        & [\si{\degreeCelsius}]
        & [\si{\degreeCelsius}]
        & [\si{\degreeCelsius}] \\
        \midrule
        Proposed MPC
        & 4.8
        & 0.38
        & 197 \\
        No-valve-control MPC
        & 18.2 
        & 1.13
        & 270 \\
        \bottomrule
    \end{tabular}
\end{table}

%% file: s5.tex
In this article, we propose a physics-based model that captures both the hydraulic and thermal behavior of hydronic radiators in buildings while remaining suitable for MPC frameworks. For a case-study building, namely the Rensen building in the Netherlands, we also use grey-box modeling to estimate the parameters of the introduced model that are difficult to obtain in practice. Specifically, two grey-box modeling approaches are discussed and compared using real measurements from the building: one offers theoretical guarantees, while the other is heuristic but computationally efficient, which makes it attractive for buildings with multiple zones. It is shown that both approaches have approximately the same results. We then show that accounting for hydraulic interactions among the radiators in the model reduces the root mean square error (RMSE) between the measured and modeled zone temperatures by 11\% compared with a model that neglects these interactions. Using this model, we propose a model predictive control (MPC) framework for dynamic hydraulic balancing. Through both numerical case studies and real experiments in the six-zone building, it is shown that, by jointly adjusting the boiler supply temperature and radiator-valve positions to achieve dynamic hydraulic balancing, the proposed MPC can effectively distribute water flow among the zones according to their heat demands. In particular, the proposed MPC reduces the total comfort violations by at least $27\%$ compared with existing MPC strategies that either neglect hydraulic interactions among the radiators or do not control the radiator valves, while requiring a lower or nearly identical cumulative boiler supply temperature, used here as a proxy for energy use.

For future work, we aim to account for pipe heat losses and model inaccuracies. We also aim to conduct experiments over longer periods to better assess the long-term effects of the proposed MPC on energy savings and thermal comfort.